%% file: main.tex
\documentclass[11pt]{article}
\usepackage[top=1in, bottom=1.25in, left=1.25in, right=1.25in]{geometry}
\usepackage{amsmath, amsthm, amsfonts, amssymb, amscd, bm, enumerate}

\usepackage{amsmath, amsfonts, amssymb, mathrsfs}
\DeclareMathOperator*{\argmin}{arg\,min}
 
\usepackage{bm} 
\usepackage{cool} 
\usepackage{mathtools}
\usepackage{bbm}
\DeclareMathOperator{\tr}{tr}
\counterwithin{equation}{section}
\makeatletter 
\providecommand*{\diff}%
	{\@ifnextchar^{\DIfF}{\DIfF^{}}}
\def\DIfF^#1{%
	\mathop{\mathrm{\mathstrut d}}%
	\nolimits^{#1}\gobblespace}
\def\gobblespace{%
	\futurelet\diffarg\opspace}
\def\opspace{%
	\let\DiffSpace\!%
	\ifx\diffarg(%
		\let\DiffSpace\relax
	\else
		\ifx\diffarg[%
			\let\DiffSpace\relax
		\else
			\ifx\diffarg\{%
				\let\DiffSpace\relax
			\fi\fi\fi\DiffSpace}
\providecommand*{\deriv}[3][]{%
	\frac{\diff^{#1}#2}{\diff #3^{#1}}}
\usepackage{xcolor}
\definecolor{oxford_blue}{RGB}{14,31,71}
\usepackage{graphicx}
\usepackage{wrapfig}
\usepackage{booktabs}
\usepackage{multirow}
\usepackage{multicol}
\usepackage{cite}
\usepackage[numbers]{natbib}
\usepackage[labelfont=bf]{caption} 
\usepackage{subcaption} 
\usepackage[ruled,vlined,linesnumbered,nofillcomment]{algorithm2e}

\SetCommentSty{mycommfont}
\SetKw{KwBy}{by}

\usepackage{enumerate, enumitem}
\usepackage[colorlinks]{hyperref}
\usepackage{etoolbox} 
\usepackage{url}
\usepackage{tikz}
\usetikzlibrary{shapes,intersections,arrows,arrows.meta,positioning,fit,chains,decorations.markings,shapes.geometric}
\usetikzlibrary{matrix,decorations.pathreplacing,calc}
\tikzset{%
  >={stealth[width=1mm,length=1mm]},
            base/.style = {rectangle, rounded corners, draw=black,
                           minimum width=1.5cm, minimum height=1cm,
                           text centered, font=\sffamily},
  			mass/.style = {base, text width=3cm, thick},
         process/.style = {base, on chain,
         				   minimum width=1.2cm, fill=white,
         				   text width=1.2cm},
            data/.style = {trapezium, draw, on chain,
            			   anchor=center,
            			   minimum width=1.5cm, minimum height=.75cm,
            			   trapezium left angle=75, trapezium right angle=105,
            			   text centered, font=\ttfamily,
            			   text width=1.6cm,
            			   fill=red!30},
           	  or/.style = {draw, circle, minimum size=0.05cm, fill=blue!30},
}
\usetikzlibrary{calc,quotes,angles}
\usepackage{xparse}
\NewDocumentCommand\Perp{O{0.5}mmmO{1cm}}{%
  \coordinate (#4) at
    ($ ($ #2!#1!#3 $) ! {sin(90)} ! 90:#3 $) {};
  \draw[-latex, red] ($ #2!#1!#3 $) -- ($ ($ #2!#1!#3 $) ! #5 ! (#4)$);
  \coordinate (#4) at ($ ($ #2!#1!#3 $) ! #5 ! (#4)$);  
}
\tikzset{%
  square/.style={
    regular polygon,
    regular polygon sides=4,
    draw,
    minimum size=1.5cm
  },
  every neuron/.style={
    circle,
    draw,
    minimum size=1cm
  },
  trans neuron/.style={
    square,
    draw,
    minimum size=1.5cm
  },
  neuron missing/.style={
    draw=none, 
    scale=4,
    text height=0.333cm,
    execute at begin node=\color{black}$\vdots$
  },
  neuron empty/.style={
    draw=none, 
    scale=4,
    text height=0.333cm
  },
}

\newtheorem{theorem}{Theorem}
\newtheorem{proposition}{Proposition}

\theoremstyle{remark}
\newtheorem{remark}{Remark}

\begin{document}

\title{Arbitrage-free neural-SDE market models
}

\author{%
Samuel N. Cohen \and Christoph Reisinger \and Sheng Wang \and 
Mathematical Institute, University of Oxford \\
\texttt{ \{samuel.cohen, christoph.reisinger, sheng.wang\} }\\ 
\texttt{@maths.ox.ac.uk}
}

\date{}

\maketitle

\begin{abstract}

Modelling joint dynamics of liquid vanilla options is crucial for arbitrage-free pricing of illiquid derivatives and managing risks of option trade books. This paper develops a nonparametric model for the European options book respecting underlying financial constraints while being practically implementable. We derive a state space for prices which are free from static (or model-independent) arbitrage and study the inference problem where a model is learnt from discrete time series data of stock and option prices. We use neural networks as function approximators for the drift and diffusion of the modelled SDE system, and impose constraints on the neural nets such that no-arbitrage conditions are preserved. In particular, we give methods to calibrate \textit{neural SDE} models which are guaranteed to satisfy a set of linear inequalities. We validate our approach with numerical experiments using data generated from a Heston stochastic local volatility model.

\end{abstract}



SDE; constrained diffusions; statistical inference

\section{Introduction}

Consider a financial market where the following assets are liquidly traded: a stock\footnote{Although we refer to the asset as ``stock'' throughout the paper, our methods are equally applicable to other asset classes such as currencies or commodities.} $S$ and a collection of European call options $C(T,K)$ on $S$ with various expiries $T \in \mathcal{T}$ and strikes $K \in \mathcal{K}$. Pricing and risk management generally require a statistical model for these assets' prices, which are known to be related to each other in complex ways. For example, the volatility of the stock price through time is typically related to the value of call options, in a monotone but nonlinear fashion. These stylised relationships suggest that there is significant statistical information captured in the interrelated prices of the stock and options, however modelling these jointly is challenging. Furthermore, there are various constraints on prices which must hold in the absence of arbitrage, and these should be reflected in any statistical model.

We aim to construct a class of models for the stock and options that should
\begin{enumerate}[leftmargin=*, label=(\roman*)]
	\setlength\itemsep{1pt}
    \item permit no arbitrage;
    \item allow exact cross-sectional calibration; and 
    \item reflect stylised facts observed from market price dynamics.
\end{enumerate}
Importantly, it should be practically convenient to estimate these models, given that observations are discrete time series of prices for a large but finite collection of options. We aim to exploit the recent successes in the use of neural networks as function approximators in order to give a flexible class of models. However, the no-arbitrage constraints on option prices imply that we need to fit neural-nets where the resulting behaviour (when used in an SDE) will satisfy a family of linear (in)equalities. In this paper we will develop methods to solve this challenge, and demonstrate their effectiveness when building a financial model.

\subsection{Martingale and market model approaches}

Many models in the literature are derived using the \textit{martingale approach}. In this approach, one specifies the dynamics of the underlying $S$ (usually in the form of an SDE) under some pricing measure $\mathbb{Q}$, and uses no-arbitrage arguments and It\^o calculus to derive option prices, written as discounted conditional expectation of options' payoffs under $\mathbb{Q}$. For example, $C_t(T,K) = D_t(T) \cdot \mathbb{E}^\mathbb{Q} [(S_T - K)^+ | \mathscr{F}_t]$, where $D_t(T)$ is the time $t$ price of the zero-coupon bond expiring at time $T$. Prestigious examples of martingale models are the Black--Scholes (BS) \cite{BlackScholes1973} model, the local volatility (LV) model of Dupire \cite{Dupire1994} and Derman and Kani \cite{Derman1994}, stochastic volatility (SV) models such as Heston \cite{Heston1993} and SABR \cite{Hagan2002, Hagan2014}, and stochastic local volatility (SLV) models \cite{Jex1999}.
A data-driven approach to learn an overparametrised martingale model with
SDE coefficients expressed as neural-nets
is taken by Gierjatowicz, Sabate-Vidales, \v{S}i\v{ska}, Szpruch and \v{Z}uri\v{c} \cite{gierjatowicz2020robust}.

The martingale approach guarantees no-arbitrage by the First Fundamental Theorem of Asset Pricing (FFTAP) \cite{harrison1979}. Therefore, the martingale approach eliminates arbitrage by construction, and the LV model and SV models can calibrate exactly for \textit{static} price surfaces. However, some challenges remain:
\begin{enumerate}[leftmargin=*, label=(\roman*)]
\setlength\itemsep{1pt}
\item martingale models do not give explicit expressions for option prices, often requiring heavy, model-specific numerical methodology to calibrate these models to market data;
\item the values of calibrated model parameters are observed to change over time, even though they are assumed to be constant by the posited model;
\item martingale models are naturally posed under the (risk-neutral) measure $\mathbb{Q}$, and additional steps are often needed to use them to model dynamics under the historical measure $\mathbb{P}$, as is needed for risk management.
\end{enumerate}

An alternative family of models admitting exact calibration are \textit{market models}, where one specifies the joint dynamics of all liquid tradables simultaneously, and imposes additional conditions to exclude arbitrage, as option prices are no longer automatically discounted conditional expectations. This idea originates from the framework of Heath--Jarrow--Morton (HJM) \cite{HJM1992} for interest rate modelling, which establishes no-arbitrage through drift restrictions. The case for options is more complicated, as the contract specifications at maturity enforce convoluted relationships between option prices and the underlying price. In other words, the state space of these processes is heavily constrained. 

There is a useful distinction to be drawn, in this context, between ``model-free'' and ``model-based'' arbitrage. If a market model produces prices which display model-free arbitrages, then these are easily exploited, as they do not depend on the dynamic model used. A counterparty simply needs to observe prices inconsistent with a set of restrictions (in practice given by a finite set of linear inequalities), and can quickly identify arbitrage opportunities. On the other hand, if a market model only has the risk of producing model-based arbitrage, then this may be less of a concern; to exploit this opportunity, a counterparty would typically need to know the model being used (including all conventions around interpolation of prices). Indeed, if no model-free arbitrages are present, then it is known that there exists a dynamic model which is arbitrage free and replicates these prices (see, for example, Carr and Madan \cite{Carr2005}, Davis and Hobson \cite{davis2007}).

Conveniently, the model-free arbitrage restrictions correspond to \emph{static} arbitrage constraints, that is, to restrictions on the state-space of the price process. In previous work \cite{Cohen2020} (drawing on Cousot \cite{cousot2007}), we have reduced these constraints, for a collection of arbitrary strikes and maturities, to an efficient set of linear inequalities.

In this paper we will focus our attention, therefore, on estimating (under the historical measure $\mathbb{P}$) a market model where the model-free arbitrage constraints will be satisfied. We will see that this restriction of the state-space of our model yields significant benefits to the task of statistical calibration.

\subsection{Factor-based arbitrage-free market models}

To allow practical calibration, we will focus on constructing \textit{market models} for \textit{finitely} many options that may have an \textit{arbitrarily-shaped} lattice of strikes and expiries.

Typically, the range of strikes actively quoted on the market is much smaller for short-expiry options than for long-expiry options, as seen in Figure \ref{fig:optionLattice}\footnote{Unsurprisingly, the ranges of quoted Black--Scholes deltas appear quite similar across expiries, which is consistent with the OTC FX convention of quoting option price by delta (see Reiswich and Wystup \cite{Reiswich2009}).}. This causes problems for methods based on a \textit{rectangular} lattice\footnote{The lattice of strikes and expiries is called rectangular if prices are specified for a finite collection of discrete strikes and, for each strike, for a common finite set of expiries.}. One might augment market data using arbitrage-free interpolation techniques (see e.g. Kahal\'{e} \cite{Kahale2003} and Fengler \cite{Fengler2009}), but these are prone to extrapolate unrealistic prices for deep ITM and OTM options. Our models allow an arbitrarily shaped lattice of options, which can be specified according to the stylised liquidity profile of the options market. 

\begin{figure}[!h]
\centering
\includegraphics[scale=0.66]{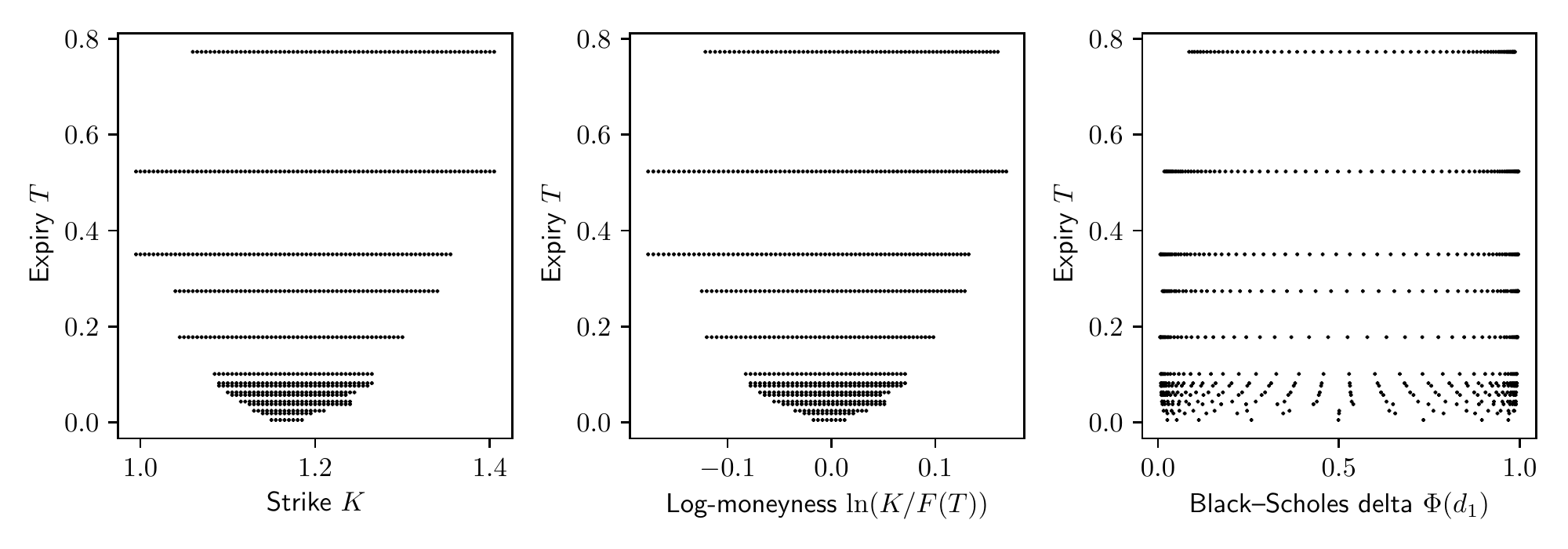}
\caption{The quoted strikes/log-moneynesses/Black--Scholes deltas (horizontal axis) for all expiries (vertical axis) of CME weekly- and monthly-listed EURUSD European call options as of 31st May, 2018. Each dot represents a market quote.}
\label{fig:optionLattice}
\end{figure}

Imposing no-arbitrage conditions is a key component in our models. To achieve this, our models first ensure that for each fixed time, prices take values in a constrained state space such that options with different strikes (moneynesses) and expiries are \textit{statically} arbitrage-free. Unlike the neat static arbitrage constraints for options on a rectangular lattice, e.g. see Carr and Madan \cite{Carr2005},  the flexibility of modelling options on an arbitrarily-shaped lattice induces challenges for deriving static arbitrage constraints. Based on the work of Davis and Hobson \cite{davis2007} and Cousot \cite{cousot2007}, our previous work \cite{Cohen2020} produces an efficient construction of static arbitrage constraints for large-scale practical problems. Second, we aim both to guarantee the statistical accuracy of our model (by training with historical data) and to minimise the opportunities for dynamic arbitrage -- the discounted price process of each traded asset, i.e.\ stock and all options, should be a martingale under \emph{some} joint risk-neutral measure (which is not generally the measure from which our training data are sampled). This involves compatibility restrictions on the model coefficients, analogous to the HJM \cite{HJM1992} drift restrictions.

The static arbitrage constraints are \textit{linear} with respective to option prices, and enforce simple geometric shape constraints on the surface $(T,K) \mapsto C(T,K)$, including positivity, monotonicity and convexity. This effectively restricts the evolution of the option price surface to a low-dimensional submanifold; see for example Cont, Fonseca and Durrleman \cite{Cont2002, Cont20022} for the use of Karhunen--Lo\`{e}ve decomposition on implied volatility (IV) surfaces to extract low-dimensional factor models. Rather than decomposing IVs, we decode factors directly from option prices, preserving the linearity of the static arbitrage constraints in the factor space. This results in a convex polytope state space for the factors.

We then assume that the evolution of the collection of option prices is driven by a small number of latent \textit{market factors}, for which we specify diffusion-like dynamics. To summarise, from the no-arbitrage conditions, we derive a class of low-dimensional market models with drift restrictions and a convex polytope state space. We shall see that this low-dimensional representation has various advantages, in particular it is possible to calibrate non-parametric models for the factor dynamics, the estimated models are noticeably more stable, and the number of constraints that need to be verified is significantly reduced.

\subsubsection*{Comparison with codebook models}

Most work on arbitrage-free market models for option prices has sought a convenient parametrisation, or \textit{codebook} in the language of Carmona \cite{Carmona2007}, such that the codebook processes have a simple state space and yet capture all the static arbitrage constraints. Specifying the dynamics of the codebook then leads to a tractable arbitrage-free dynamic model for options.

A proper choice of codebook processes improves the descriptive capability of the model in terms of the observed strikes and expiries ($\mathcal{K}$ and $\mathcal{T}$). Sch\"{o}nbucher \cite{Schonbucher1999} proposes a market model in terms of the BS \textit{implied volatility} $\sigma_\text{imp}$ for a single option, i.e.\ the case $\mathcal{K} = \{K\}$ and $\mathcal{T} = \{T\}$. The model is free from static arbitrage if $S>0$ and $\sigma_\text{imp} \geq 0$. However, as discussed by Schweizer and Wissel \cite{Schweizer2008}, market models of implied volatilities cannot be easily extended to the general case with more than one option, because the absence of static arbitrage between different options enforces awkward constraints between the corresponding implied volatilities (see, for example, Lemma 2.2 in \cite{Gatheral2014}). For the term structure case,  $\mathcal{K} = \{K\}$ and $\mathcal{T} = [0, +\infty)$, Sch\"{o}nbucher \cite{Schonbucher1999} and Schweizer and Wissel \cite{Schweizer2007} use the \textit{forward implied volatilities}, defined as $\sigma_\text{fw}^2(T):= \partial ((T-t)\sigma_\text{imp}^2(T)) / \partial T$, as the codebook. Their model is statically arbitrage-free if $S>0$ and $\sigma_\text{fw}^2(T) \geq 0$ for all $T \in \mathcal{T}$. Jacod and Protter \cite{Jacod2010} directly use \textit{call option prices} as the codebook and work in a more general setting of jump processes. Static arbitrage is ruled out by imposing a non-decreasing term structure on the call option prices. For the multi-strike case $\mathcal{K} = [0, +\infty)$ and $\mathcal{T} = \{T\}$, Schweizer and Wissel \cite{Schweizer2008} introduce a new parametrisation of option prices called \textit{local implied volatilities} $X(K)$ and \textit{price level} $Y$, see Definition 4.1 and 4.4 of \cite{Schweizer2008}. Their model does not admit static arbitrage if $X(K)>0$ for all $K \in \mathcal{K}$.

The surface case $\mathcal{K} = [0, +\infty)$ and $\mathcal{T}= [0, +\infty)$ has been considered by Derman and Kani \cite{Derman1998}, Carmona and Nadtochiy \cite{Carmona2009, Carmona2012} and Kallsen and Kr\"{u}hner \cite{Kallsen2013}. Carmona and Nadtochiy use Dupire's local volatilities $\sigma_\text{loc}(T,K)$ as the codebook processes and rigorously analyse the dynamic arbitrage conditions derived by Derman and Kani. Later, these authors independently built market models relying on time-inhomogeneous L\'{e}vy processes, which allow jumps and may hence be particularly suitable for short-term options. Both parametrisations make the static arbitrage constraints naturally hold.

The non-trivial construction of codebook processes, i.e.\ a PDE for local volatility \cite{Carmona2009} and Fourier transforms for L\'{e}vy process \cite{Carmona2012, Kallsen2013}, add complexities for model calibration. By considering a continuous spectrum of strikes and expiries, one is usually forced to consider an infinite-dimensional problem, where it is difficult to prove the existence of models and computationally infeasible to implement them exactly.  The only work, to our knowledge, that considers market models for a finite family of strikes and expiries is Wissel \cite{Wissel2008}. Wissel combines ideas from the LV model of Dupire with the market model of Schweizer and Wissel \cite{Schweizer2008} and uses local implied volatilities and price level as a codebook for parametrising statically arbitrage-free option prices on a rectangular lattice. The rectangular lattice setup is necessary for the codebook construction but is a rather restrictive assumption in practice. In addition, there is not a straightforward model calibration method for general cases.

\subsection{Model inference and constrained neural SDE}

For our models, inference consists of two independent steps: factor decoding and SDE model calibration.

The factor decoding step is to extract a smaller number of market factors from prices of finitely many options. These factors are built to reflect the joint goals of eliminating static and dynamic arbitrage in reconstructed prices and guaranteeing statistical accuracy.
Expressing option prices as affine functions of market factors, we use a modified form of principal component analysis (PCA) to decompose option prices into a small number of driving factors.

In the SDE model calibration step, we represent the drift and diffusion functions by neural networks, referred to as \textit{neural SDE} by Cuchiero, Khosrawi and Teichmann \cite{Cuchiero2020} and Gierjatowicz et al.\ \cite{gierjatowicz2020robust}. By leveraging deep learning algorithms, we train the neural networks by maximising the likelihood of observing the factor paths, subject to the derived arbitrage constraints. As argued by Gierjatowicz et al.\ in their framework of neural martingale models, this allows for calibration and model selection to be done simultaneously.

No-arbitrage conditions are embedded as part of model inference. Specifically, static arbitrage constraints are characterised by a convex polytope state space for the market factors; we identify sufficient conditions on the drift and diffusion to restrict the factors to their arbitrage-free domain, using the classic results of Friedman and Pinsky \cite{friedman1973}. Consequently, the neural network that is used to parameterise the drift and diffusion functions needs to be constrained\footnote{Chatatigner, Cr\'epey and Dixon \cite{Chataigner2020} compare the \textit{hard constraint} and \textit{soft constraint} approaches for imposing no-arbitrage constraints on put option prices. To enforce hard constraints, one modifies the network structure to embed the constraints (see Dugas, Benigo, B\'elisle, Dadeau and Garcia \cite{Dugas2009}), while the soft constraint approach introduces penalty terms favouring the satisfaction of the constraints (see Itkin \cite{itkin2019deep}). Though all these works are concerned with enforcing no-arbitrage conditions on neural networks that generate option prices, they are all solving cross-sectional calibration problems, rather than learning dynamic models.}. We achieve this by developing appropriate transformations for the output of the neural network.

Our models are friendly for implementation for the following reasons:
    \begin{enumerate}[leftmargin=*, label=(\roman*)]
    \setlength\itemsep{1pt}
        \item They are flexible enough to model options on an arbitrarily-shaped lattice of strikes and expiries. One can specify the lattice according to the stylised liquidity profile of the target option market.
        \item The model inference procedure includes both cross-sectional calibration (i.e.\ the factor decoding step) and time-series estimation (i.e.\ the neural SDE calibration step). Neural SDEs enable a data-driven model selection approach, and mitigate model risk arising from restrictive parametric forms.
        \item Further economic considerations (for example bounds on statistical arbitrage) can be incorporated into the training of these models through penalization.
    \end{enumerate}

\subsection{Key contributions}

We summarise the key contributions\footnote{We have implemented the factor decoding algorithm, constrained neural network training and market model simulation in Python in the repository \url{https://github.com/vicaws/neuralSDE-marketmodel}. } of this paper as follows.

\begin{enumerate}[leftmargin=*, label=(\roman*)]
\setlength\itemsep{1pt}
    \item We construct a family of factor-based market models (Section \ref{sec:market_model}), where the factor representation in principle allows exact static calibration, and the joint dynamics of options are straightforwardly available once the dynamics of the factors are specified. The models are given by a finite system of SDEs for the factors and the stock price. 
    
    \item We derive an HJM-type drift condition on the factor SDEs which guarantees freedom from dynamic arbitrage, and the state space of the market factor processes where the models are free from static arbitrage.
    
    \item To calibrate our models, we need to tackle the problem of calibrating a neural SDE with a convex polytope state space. We propose a novel hard constraint approach that modifies the network to respect sufficient conditions on the drift and diffusion to restrain the process within the polytope. Readers who are interested in this calibration problem, rather than market models, can directly jump to Section \ref{sec:neural_sde_polytope}. 
\end{enumerate}

\section{Arbitrage-free market models}
\label{sec:market_model}

We assume a frictionless market in continuous-time, with a distant fixed finite horizon $T^*$. We allow non-zero interest and dividends\footnote{When applying our methods to other asset classes, dividends of stocks are comparable to foreign currency interest for FX, or convenience yield for commodities.}, and assume a deterministic interest rate $r_t$ and dividend yield $q_t$ at time $t$. We use $D_t(T) = \exp (-\int_t^T r_s \diff s)$ to denote the market discount factor for time $T$, and $\Gamma_t(T) = \exp (\int_t^T q_s \diff s)$ to account for reinvestment of dividends. There is a model-independent, arbitrage-free forward price, $F_t(T) = S_t / (\Gamma_t(T)D_t(T))$, for delivery of the asset at $T$, where $S_t$ is the price of the underlying asset. The primary traded securities in the market are zero coupon bonds, the stock, and a collection of forwards and European call options written on the stock. We assume no derivatives expire after $T^*$.

\subsection{Liquid options}

While it is theoretically reasonable to search for a model to produce arbitrage-free option prices for every possible $T$ and $K$ at any time $t$, doing so can be both practically unnecessary and restrictive. It is of limited use to ensure absence of arbitrage with regards to deep in or out-of-the-money options, given these are not liquidly traded, and modelling these would impose significantly more constraints when calibrating to data. Our models focus on the joint dynamics of \textit{liquid} options, rather than options of \textit{any} strikes and expiries. Liquid options are usually those that are neither too far from the money nor too close to expiry (typically a few days). In addition, the range of liquid strikes typically broadens for longer expiries, as seen in Figure \ref{fig:optionLattice}.

Over time, the range of liquid options in $(T,K)$-coordinates tends to change \textit{stochastically}. For example, a liquid at-the-money option as of today will become deep out-of-the-money  if the underlying price plunges quickly. It is empirically advantageous, therefore, to cast the liquid range into relative coordinates $(\tau, m)$, where $\tau=T-t$ is time-to-expiry and $m$ denotes moneyness. There is usually a stable range of time-to-expiries and moneynesses for which the options are actively quoted and therefore price data are most readily available. A convenient parameterisation is given by the \textit{forward log-moneyness} $m_t$, which is
\begin{equation}
    m_t = M(K; F_t(T)) := \ln\left( \frac{K}{F_t(T)} \right).
\end{equation}
The function $M(\cdot)$ is referred to as the moneyness function. In Figure \ref{fig:dynamicsKT}, we highlight the trajectories of $(\tau_t, m_t)$ for three differently struck options (with the same expiry). The option that is very far from the money, i.e.\ $K=1.3$ (green line), moved out of the liquid range (which could be approximated by the gray dotted area) and was never traded after its time-to-expiry became shorter than $0.4$ years.
\begin{figure}[!ht]
\centering
\includegraphics[scale=0.66]{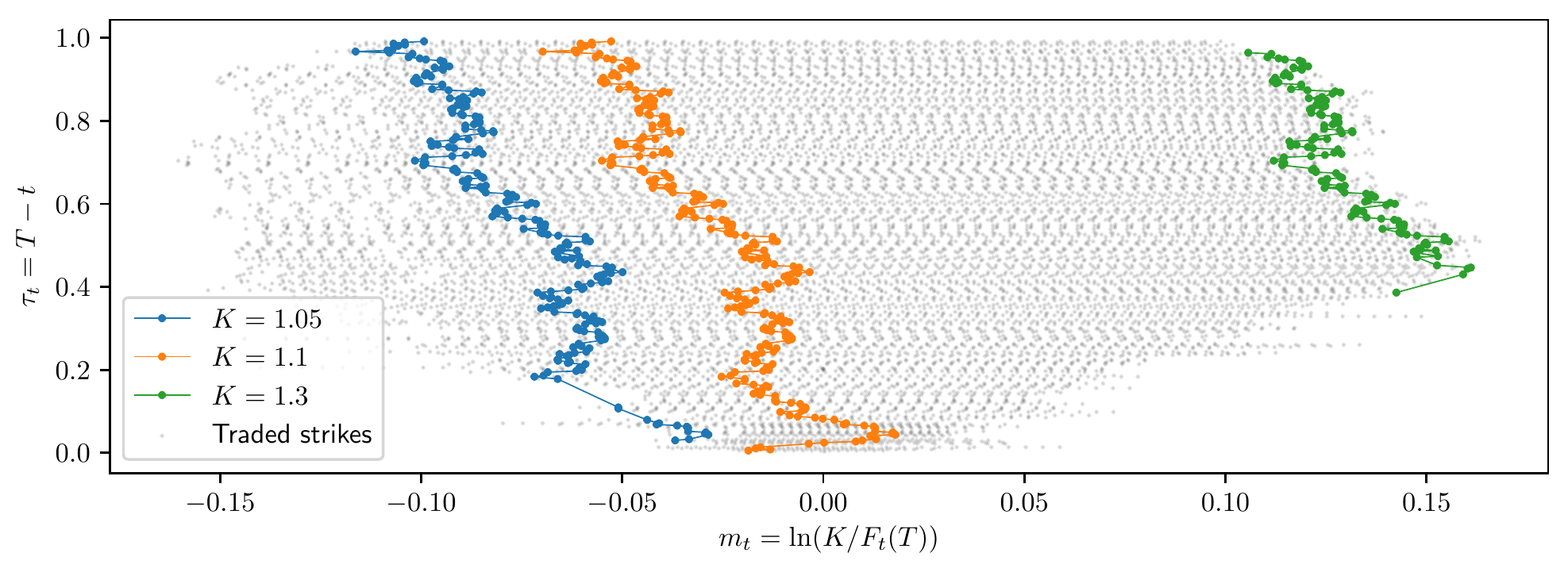}
\caption{An example of the trajectory of $(\tau_t, m_t)$ for fixed $(T, K)$'s. Grey dots indicate the traded strikes for CME EURUSD option contracts with expiry date 6th March, 2020, from their first listing day until expiry. We highlight the trajectories for three differently struck contracts.}
\label{fig:dynamicsKT}
\end{figure}

We assume the range of liquid options in the $(\tau, m)$-coordinates is fixed over time, and develop arbitrage-free models on the the liquid range. We denote by $\underline{\tau}$ and $\overline{\tau}$ the minimal and maximal liquid time-to-expiries. For fixed $\tau$, we denote by $\underline{m}(\tau)$ and $\overline{m}(\tau)$ the minimal and maximal liquid moneynesses. We then define the constant set $\mathcal{R}_\text{liq} = \{(\tau, m) \in \mathbb{R}^2:  \tau \in [\underline{\tau}, \overline{\tau}], m\in [\underline{m}(\tau), \overline{m}(\tau)] \}$ as the range of moneynesses and time-to-expiries where options are liquid. A typical market's range exhibits that $0 < \underline{\tau} < \overline{\tau} < T^*$; $\underline{m}(\tau) < 0 < \overline{m} (\tau)$ for all $\tau \in [\underline{\tau}, \overline{\tau}]$; $\underline{m}(\tau_1) \geq \underline{m}(\tau_2)$ and $\overline{m}(\tau_1) \leq \overline{m}(\tau_2)$ for any $\underline{\tau} \leq \tau_1 < \tau_2 \leq \overline{\tau}$. An example can be seen in Figure \ref{fig:liquidRange}.

\begin{figure}[!ht]
\centering
\includegraphics[scale=0.66]{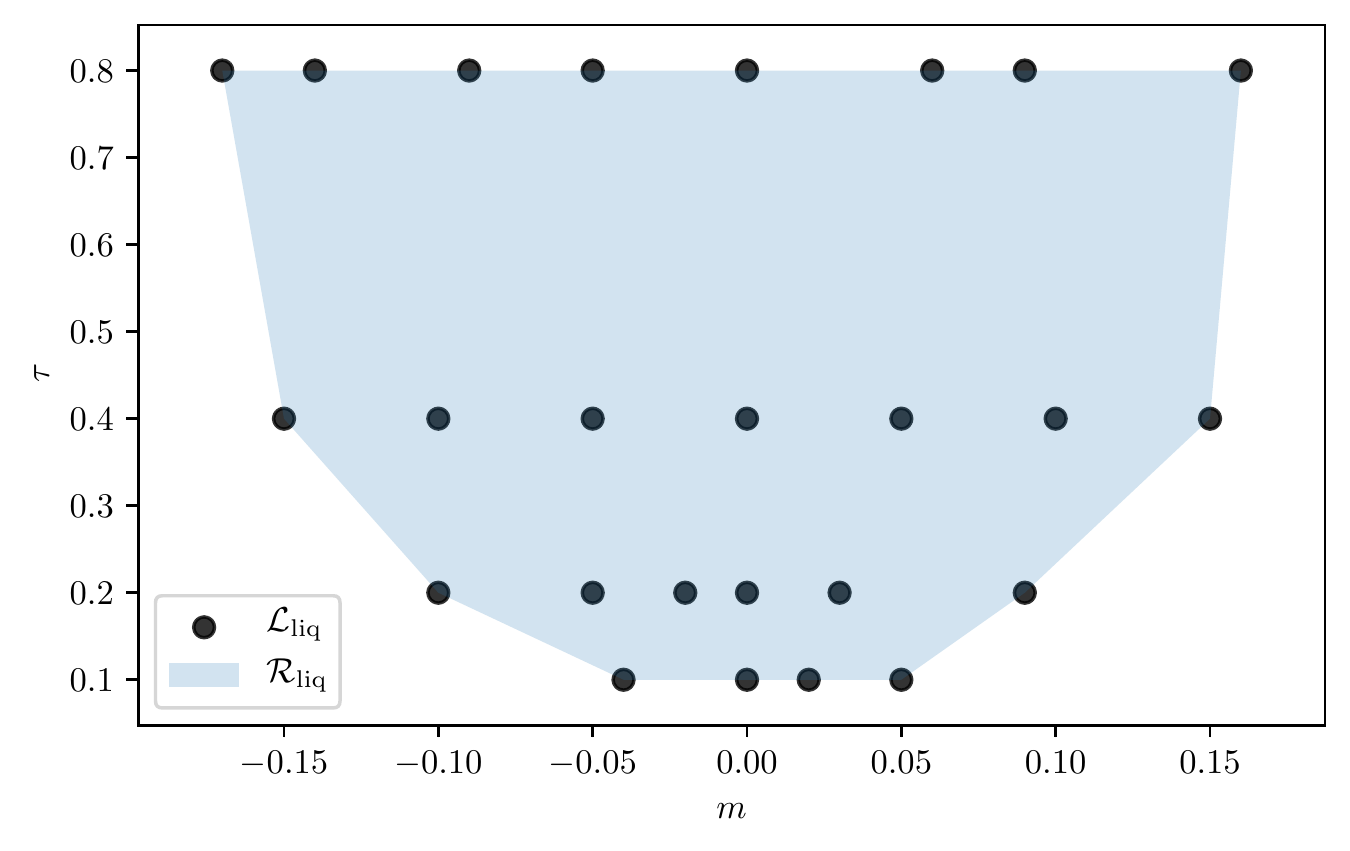}
\caption{An example of liquid range and lattice.}
\label{fig:liquidRange}
\end{figure}

Our models construct option price surfaces on $\mathcal{R}_\text{liq}$ through smooth interpolation of a finite collection of option prices, with parameters given by the set $\mathcal{L}_\text{liq} = \{ (\tau_i, m^i_j) \in \mathcal{R}_\text{liq}: \underline{m}(\tau_i) = m^i_1 < \cdots < m^i_{n_i} = \overline{m}(\tau_i), \underline{\tau} = \tau_1 < \cdots < \tau_p = \overline{\tau}, 1 \leq j \leq n_i, 1 \leq i \leq p \}$. We call $\mathcal{L}_\text{liq}$ an \textit{option lattice}. There are two reasons for us to have this model design. First, given the static arbitrage requirements concern monotonicity and convexity of the price surface, our models are free of static arbitrage on $\mathcal{R}_\text{liq}$ if options on $\mathcal{L}_\text{liq}$ are statically arbitrage-free and these prices are extended to $\mathcal{R}_\text{liq}$ using a shape-preserving interpolation method (in particular, a method which preserves monotonicity in time to maturity and monotonicity and convexity in moneyness). Second, options with very close strike and expiry will have strongly dependent price processes, and so contain little new statistical information. We use options on $\mathcal{L}_\text{liq}$ to represent local dynamics of different sub-areas of the surface on range $\mathcal{R}_\text{liq}$; this is analogous to a common practice for quoting OTC data. For example, in the OTC FX market, traders quote options on the lattice [1-day, 1-week, 1-month, ...] $\times$ [10-delta, 15-delta, ...] and smoothly interpolate them for pricing a specific contract.

\subsection{Factor representation of call option prices}

Even after reducing our model to the finite set $\mathcal{L}_\text{liq}$ of liquid options, one usually finds that these prices contain significant redundancies. For this reason, it is efficient to express them through a low-dimensional factor representation, for which a reliable model can be built. These factors will be built from data, but are simple enough to be intuitively understood and plotted, as they have a linear relationship with option prices.

We will build a model based on the randomness provided by $d+1$ standard independent Brownian motions $W_{0}, W_{1}, \dots, W_{d}$ under the objective measure $\mathbb{P}$. Writing $W =[W_{1} ~\cdots ~W_{d}]^\top$ and $\overline{W} = [W_{0} ~W_{1} ~\cdots ~W_{d}]^\top$, we describe the information available at time $t$ by the $\sigma$-algebra $\mathscr{F}_t$, which is the completion of $\cap_{s>t} \mathscr{F}_s^0$ with $\mathbb{P}$-null sets, where $\mathscr{F}_t^0 := \Sigma(\{\overline{W}_s: 0 \leq s \leq t\}) \vee \mathscr{F}_0$ (here we use $\Sigma(\cdot)$ to denote the sigma-algebra generator). Taking $\Omega$ to be the space of paths of $\overline W$, we have a filtered probability space $(\Omega, \{\mathscr{F}_t\}_{0\leq t\leq T^*}, \mathbb{P})$.

Suppose that the market is described by the underlying price $S \in \mathbb{R}$ and some latent \textit{market factor} $\xi \in \mathbb{R}^d$. We model the joint dynamics of the underlying price and the market factor by a $(d+1)$-dimensional diffusion process solving the time-homogeneous SDE:
\begin{equation}
\begin{dcases}
    \diff S_t = \left( \alpha(\tilde{\xi}_t)-q_t \right) S_t \diff t + \gamma(\tilde{\xi}_t) S_t \diff W_{0,t}, \quad & S_0 = s_0 \in \mathbb{R};\\
    \diff \xi_t = \mu (\tilde{\xi}_t) \diff t + \sigma (\tilde{\xi}_t) \diff W_t,  \quad & \xi_0 = \zeta_0 \in \mathbb{R}^d,
\end{dcases}
\label{eq:market_model}
\end{equation}
where $\tilde{\xi}_t = (S_t, \xi_t)$. We use $\alpha_t$, $\gamma_t$, $\mu_t$ and $\sigma_t$ to denote $\alpha(\tilde{\xi}_t)$, $\gamma(\tilde{\xi}_t)$, $\mu(\tilde{\xi}_t)$ and $\sigma(\tilde{\xi}_t)$, respectively. We denote by $L^p_\text{loc}(\mathbb{R}^d)$ the space of all $\mathbb{R}^d$-valued, progressively measurable, and locally $p$-integrable (in $t$, $\mathbb{P}$-a.s.) processes on $[0,T^*]$, and assume that $\alpha S \in L^1_\text{loc}(\mathbb{R})$, $\mu \in L^1_\text{loc}(\mathbb{R}^{d})$, $\gamma S \in L^2_\text{loc}(\mathbb{R})$ and $\sigma \in L^2_\text{loc}(\mathbb{R}^{d})$.

We now associate the latent market factor with call option prices. We denote by $C_t(T,K)$ the time $t$ market price of an European call option with expiry $T$ and strike $K$. We call the map $(T,K) \mapsto C_t(T,K)$ a \textit{call price surface}. After casting the surface into relative coordinates $(\tau, m)$, the surface is parametrised over the range $\{(\tau, m): \tau \in (0, T^*-t], m \in \mathbb{R} \}$. For some fixed $T$ and $K$, we define the following transformations of call option prices for all $t \in [0, T^*)$:
\begin{equation}
    \hat{c}_t(T,K) = \frac{C_t(T,K)}{D_t(T)F_t(T)}, \quad
    \tilde{c}_t(T-t, M(K; F_t(T))) = \hat{c}_t(T,K),
\label{eq:call_price_transformation}
\end{equation}
where $\hat{c}_t$ and $\tilde{c}_t$ are the \textit{normalised} prices but with different arguments. The parametric surface we shall model is the normalised call price surface $(\tau, m) \mapsto \tilde{c}_t(\tau, m)$ rather than the original call price surface $(T,K) \mapsto C_t(T,K)$, though $C(T,K)$ can be implied from a family of $\tilde{c}(\tau, m)$. A key advantage of $\tilde{c}(\tau, m)$ over $C(T,K)$ is that it is more reasonable to suppose $t\mapsto\tilde{c}_t(\tau, m)$ to be a time-homogeneous process, whereas $C_t(T,K) \rightarrow (S_T-K)^+$ as $t \uparrow T$, making a statistical model more difficult.

For some fixed $\tau$ and $m$, we assume that the normalised option price $\tilde{c}_t(\tau, m)$ is determined by the latent factor $\xi$ through a time-independent linear transformation
\begin{equation}
    \tilde{c}_t(\tau,m) = g(t, \tau, m, \xi_t) := G_0(\tau, m) + \sum_{i=1}^d G_i(\tau, m) \xi_{it},
\label{eq:factr_rep_linear}
\end{equation}
where $G_i \in C^{1,2}(\mathcal{R}_\text{liq})$ is called a \textit{price basis function}, for $i=0,\dots,d$. The smoothness requirements on $G_i(\cdot)$ are necessary for deriving the dynamics of $C(T,K)$, for a given $T$ and $K$, from the dynamics of a family of $\tilde{c}(\tau, m)$. Though this representation does not show explicit dependence between the normalised call price $\tilde{c}(\tau, m)$ and the underlying price $S$, it does relate the two variables through the moneyness $m$.

Combining (\ref{eq:market_model}) (\ref{eq:call_price_transformation}) and (\ref{eq:factr_rep_linear}), we have a model for the joint dynamics of $S$ and $C(T,K)$ for all relevant strikes and expiries. We observe that the drift and diffusion coefficients of $\xi$ are functions of $S$. In addition, given the potential nonlinearity in (\ref{eq:market_model}), we argue that a linear transformation in (\ref{eq:factr_rep_linear}) is not a restrictive assumption, as a smooth nonlinear relation with respect to one factor can be approximated by a linear combination of polynomials of that factor.

\begin{remark}
The use of a Markovian model for $\xi$ may seem a little restrictive. In theory, any observed data could be included in the model for $\xi$. In practice, we shall build factors $\xi$ to minimize reconstruction errors in \eqref{eq:factr_rep_linear}. As this involves all liquid option prices, it determines the law of $S$ under $\mathbb{Q}$ in the future (through the Breeden--Litzenberger formula \cite{breeden1978}); in other words, the joint process given by $S$ and all call option prices is guaranteed to be Markovian under the pricing measure. This is closely related to the endogenous completeness of market models with options, see Davis and Ob{\l}{\'o}j \cite{DavisObloj2008}, Schwarz \cite{Schwarz2017} and Wissel \cite{Wissel2008}.
\end{remark}

\subsubsection*{Convergence to terminal payoff}

To ensure prices of call options converge to correct payoffs at expiry, independently of the behavior of $\xi$, by examining \eqref{eq:factr_rep_linear} it is clearly sufficent (and generally necessary) that
\begin{equation}
    G_0(0, m) = (1 - e^m)^+, \text{~and~} G_i(0, m) = 0, ~\forall i=1,\dots,d.
    \label{eq:initial_G}
\end{equation}
Under this assumption, $\tilde{c}_t(0,m) = (1 - e^m)^+$ and, for any $T$ and $K$,
\begin{equation*}
\begin{aligned}
    \lim_{t \uparrow T} C_t(T,K) = & \lim_{t \uparrow T} D_t(T) F_t(T) \tilde{c}_t \left( T-t, \ln \left( \frac{K}{F_t(T)} \right) \right) \\ 
     = & S_T \tilde{c}_T \left(0, \ln \left( \frac{K}{S_T} \right) \right) = \left( S_T -K \right)^+.
\end{aligned}
\end{equation*}

\subsubsection*{Connections with the Black--Scholes model}
We can also see that the Black--Scholes model is within the class of factor models we consider (in particular, it is a trivial model where no additional factors $\xi$ are needed). The Black--Scholes model assumes deterministic drift and constant volatility for the underlying asset price, i.e.\ $\gamma_t \equiv \gamma$. In fact, with 
$d_{1,2} (\tau,m) = -\frac{m}{\gamma \sqrt{\tau}} \pm \frac{1}{2}\gamma \sqrt{\tau}$,
\begin{equation}
    \tilde{c}_t(\tau,m) = G(\tau,m) := \Phi \left( d_1 (\tau,m) \right) - e^m \Phi \left( d_2 (\tau,m) \right).
\label{eq:G_BS}
\end{equation}

More generally, we can see that $\tilde{c}(\tau, m)$ has a deterministic bijective relation with the Black--Scholes implied volatility, for arbitrary but fixed $\tau$ and $m$. With implied volatility $\sigma^\text{imp}_t(\tau, m)$, the transformation (\ref{eq:call_price_transformation}) indicates that
\begin{equation}
\begin{aligned}
    \tilde{c}_t(\tau,m) & = \frac{C_t(t+\tau, e^m F_t(t+\tau))}{D_t(t+\tau) F_t(t+\tau)} \\
    & = \Phi \left(-\frac{m}{\sigma^\text{imp}_t \sqrt{\tau}} + \frac{1}{2}\sigma^\text{imp}_t \sqrt{\tau} \right) - e^m \Phi \left(-\frac{m}{\sigma^\text{imp}_t \sqrt{\tau}} - \frac{1}{2}\sigma^\text{imp}_t \sqrt{\tau} \right),
\end{aligned}
\end{equation}
where $\Phi(\cdot)$ is the cumulative density of the standard normal distribution. In addition,
\begin{equation*}
	\frac{\diff \tilde{c}_t}{\diff \sigma^\text{imp}_t} (\sigma^\text{imp}_t; \tau,m)  = \sqrt{\tau} \phi \left(-\frac{m}{\sigma^\text{imp}_t \sqrt{\tau}} + \frac{1}{2}\sigma^\text{imp}_t \sqrt{\tau} \right) > 0,
\end{equation*}
where $\phi(\cdot)$ is the density of the standard normal distribution. Hence, there is a time-independent and deterministic bijective relation between $\tilde{c}_t$ and $\sigma^\text{imp}_t$ for arbitrary but fixed $\tau$ and $m$. Consequently, by modelling the normalised call prices, we are essentially modelling the implied volatility surface as a function of $\tau$ and $m$. This aligns with the conventional ``sticky delta'' rule (see Daglish, Hull and Suo \cite{Daglish2007}), where the process for implied volatility depends on $T$, $K$, $F$ and $t$ only through its dependence on $\tau$ and $m$.

\subsection{Absence of arbitrage}
\label{sec:absence_dynamic_arbitrage}

Arbitrage refers to a costless trading strategy that has zero risk and a positive probability (under $\mathbb{P}$) of profit. The First Fundamental Theorem of Asset Pricing (FFTAP) establishes an equivalence relation between no-arbitrage and the existence of an equivalent martingale measure (EMM). After the landmark work of Harrison and Kreps \cite{harrison1979}, there are various versions of the FFTAP and extensions of the no-arbitrage concept (e.g. no free lunch, Kreps \cite{Kreps1981}, no free lunch with vanishing risk, Delbaen and Schachermayer \cite{Delbaen1994}, no unbounded profit with bounded risk, Karatzas and Kardaras \cite{Karatzas2007}). In this article, we work with a simplified version of FFTAP as follows. Given the model, there is no arbitrage if and only if $\exists \mathbb{Q} \sim \mathbb{P}$, such that the discounted price processes for all tradable assets are martingales. These tradable assets include the stock and all forwards and options written on it. 

If the model is arbitrage-free, then, for \textit{any} fixed $T$ and $K$, the discounted option price $D_0(t)C_t(T,K)$ should be a $\mathbb{Q}$-martingale.
We derive an HJM-type drift restriction for this no-arbitrage condition to hold in Appendix \ref{sec:dynamic_arbitrage},
by analyzing
the existence of a \emph{market price of risk} process $\psi_t = [\psi_{1,t} ~\cdots ~\psi_{d,t}]^\top$. Specifically, suppose we model $N = |\mathcal{L}_\text{liq}|$ options,
with prices given in $(\tau,m)$-coordinates.
Then we look for a solution $\psi_t$ to
\begin{equation}
    \begin{cases}
    &\mathbf{G}^\top \sigma_t \psi_t=\mathbf{G}^\top \mu_t - z_t, \quad \text{~where~} \\
    & z_t = \left( -\pderiv{}{\tau} - \frac{1}{2} \gamma_t^2 \pderiv{}{m} + \frac{1}{2} \gamma_t^2 \pderiv[2]{}{m} \right) \begin{bmatrix}
    \tilde{c}_t(\tau_1, m_1) \\
    \vdots \\
    \tilde{c}_t(\tau_N, m_N)
    \end{bmatrix}.
    \end{cases}
    \label{eq:mpor_linear_system}
\end{equation}
Here, we define $\mathbf{G}$ as the $d\times N$ matrix with $i$-th row $\mathbf{G}_i$, where $\mathbf{G}_i = (G_{ij})_{j=1}^N \in \mathbb{R}^N$ for $i=0,\dots, d$, and $G_{ij} = G_i(\tau_j, m_j)$. We call $\mathbf{G}^\top$ a \textit{price basis} and each of its columns a \textit{price basis vector}. 

If $z_t$ lives in the column space of the price basis $\mathbf{G}^\top$, and given that $\sigma_t$ is invertible, then (\ref{eq:mpor_linear_system}) always admits a solution for $\psi_t$, regardless of what the drift and diffusion functions are. In other words, the HJM conditions are principally related to the choice of price basis $\mathbf{G}$, rather than to the specific choice of calibrated model (i.e.\ to $\mu, \sigma$).

This encourages us to construct factors from principal components of $z_t$, as will be explored in detail in Algorithm \ref{alg:decode_factor}. Unless $z_t \in \mathbb{R}^N$ lives in an $\mathbb{R}^d$ submanifold, equality (\ref{eq:mpor_linear_system}) always yields an over-determined linear system when the number of options $N$ is greater than the number of latent market factors $d$, which is usually the case in practice. Given a price basis $\mathbf{G}$, we will then use (\ref{eq:mpor_linear_system}) to encourage choices of $\mu$ and $\sigma$ that do not suggest unreasonable values for the market price of risk $\psi$, evaluated on our liquid lattice $\mathcal{L}_\text{liq}$.

\begin{remark}
If we modelled \textit{all} $\tau$ and $m$ and enforced the drift restrictions at \textit{all} times, then this would be enough to guarantee no-arbitrage. However, our models will be built on finitely many options on a fixed liquid lattice in the $(\tau, m)$-coordinates, as specified by $\mathcal{L}_\text{liq}$. These options do not correspond to fixed contracts (i.e.\ fixed $T$ and $K$) over time. Therefore, the HJM-type drift restrictions cannot be implemented in practice, and further consideration of no-arbitrage conditions is needed.
\end{remark}

\subsection{Constrained state space for static arbitrage}
\label{sec:static_arbitrage}

To ensure that our model does not generate static/model-free arbitrage, we introduce static arbitrage constraints on the latent market factors.

\subsubsection*{Static arbitrage-free relations}

We study the static arbitrage-free relations at time $t=0$ without loss of generality. For some $T$ and $K$, we define $M_T = S_T / F_0(T)$, $k = K / F_0(T)$, and $\breve{c} (\tau, k) = \tilde{c}_0 (\tau, \ln k)$. We assume that $\breve{c}(T,k) \in C^{1,2}(\mathbb{R}_{\geq 0}^2)$.

Carr, G\'{e}man, Madan and Yor \cite{CGMY2003} and Carr and Madan \cite{Carr2005} establish the equivalence of no static arbitrage and the existence of a Markov Martingale Measure (MMM) $\mathbb{Q}$, under which
\begin{equation}
\breve{c} \left( T, k \right) = \mathbb{E}^\mathbb{Q} \left[ \left. \left( M_T - k \right)^+ \right| \mathscr{F}_0 \right].
\label{eq:call_function_arbfree}
\end{equation}
 For an arbitrary but fixed $T$, using Breeden and Litzenberger's \cite{breeden1978} analysis, the marginal measure $\mathbb{Q}_T = \mathbb{Q}(\cdot|\mathscr{F}_T)$ exists if $-1 \leq \partial \breve{c} / \partial k \leq 0$ and $\partial^2 \breve{c} / \partial k^2 \geq 0$. Furthermore, if a family of marginal measures $\{\mathbb{Q}_T\}_{T \in [0, T^*)}$ on $(\mathbb{R}, \mathcal{B}(\mathbb{R}))$ exist, which  are non-decreasing in convex order\footnote{The convex order is given by $\mathbb{Q}_{T_1} \geq_\text{cx} \mathbb{Q}_{T_2}$ if  $\int_\mathbb{R} \phi \diff \mathbb{Q}_{T_1} \geq \int_\mathbb{R} \phi \diff \mathbb{Q}_{T_2}$ for each convex function $\phi:\mathbb{R} \rightarrow \mathbb{R}$. We say $\{\mathbb{Q}_T\}_{T \in [0, T^*)}$ is non-decreasing in convex order if $\mathbb{Q}_{T_1} \geq_\text{cx} \mathbb{Q}_{T_2}$ whenever $T_1 \ge T_2$.} (NCDO), then there exists a Markov martingale measure with these marginals, by Kellerer's theorem \cite{kellerer1972}.  The convex order can be characterised in terms of call price functions:
\begin{equation*}
\mathbb{Q}_{T_1} \geq_\text{cx} \mathbb{Q}_{T_2} \Longleftrightarrow
\begin{dcases}
& \mathbb{Q}_{T_i} \text{ and } \mathbb{Q}_{T_j} \text{ have equal means}; \\
& \int_\mathbb{R} (x-k)^+ ~\text{d}\mathbb{Q}_{T_1} \geq \int_\mathbb{R} (x-k)^+ ~\text{d}\mathbb{Q}_{T_2} \quad \forall x\in \mathbb{R}.
\end{dcases}%
\end{equation*}
Hence, $\{\mathbb{Q}_T\}_{T \in [0, T^*)}$ is NDCO if $\partial \breve{c} / \partial T \geq 0$. 

Therefore, if we define the set of surface functions $s(x,y): D \rightarrow \mathbb{R} \cap [0,1]$, where $D \subseteq \mathbb{R}_{\geq 0}^2$ are compact sets, by
\begin{equation}
    \mathcal{S}(D) = \left\{s(x, y) \in C^{1,2}(D): 0 \leq s \leq 1, \pderiv{s}{x} \geq 0, -1 \leq \pderiv{s}{y} \leq 0, \pderiv[2]{s}{y} \geq 0 \right\},
\label{eq:surface_no_arbitrage}
\end{equation}
then no arbitrage can be constructed on the surface $(T,k) \mapsto \breve{c}(T,k)$ if $\breve{c} \in \mathcal{S}(\mathbb{R}_{\geq 0}^2)$.

Our models aim to ensure absence of static arbitrage for the liquid range $\mathcal{R}_\text{liq}$ defined on the $(\tau, m)$-coordinates, or equivalently the liquid range $\breve{\mathcal{R}}_\text{liq}$ defined on the $(\tau, e^m)$-coordinate, where $\breve{\mathcal{R}}_\text{liq} = \{(\tau, k): (\tau, \ln k) \in \mathcal{R}_\text{liq} \}$. Accordingly, the corresponding discrete lattice $\mathcal{L}_\text{liq}$ in $(\tau, e^m)$-coordinates is denoted by $\breve{\mathcal{L}}_\text{liq} = \{(\tau_i, k^i_j): (\tau_i, \ln k^i_j) \in \mathcal{L}_\text{liq} \}$. We will construct models to ensure that $\breve{c} \in \mathcal{S}(\breve{\mathcal{R}}_\text{liq})$. We achieve this in two steps:
\begin{enumerate}[leftmargin=*, label=(\roman*)]
\setlength\itemsep{1pt}
\item Ensure that no static arbitrage can be constructed from the finitely many liquid options on $\breve{\mathcal{L}}_\text{liq}$, i.e.\
\begin{equation}
\exists s \in \mathcal{S}(\breve{\mathcal{R}}_\text{liq}), \text{ s.t. } \forall (\tau, k) \in \breve{\mathcal{L}}_\text{liq}, ~s(\tau, k) = \breve{c}(\tau, k). 
\label{eq:call_shape_constraints}
\end{equation}

\item Given statically arbitrage-free prices on $\breve{\mathcal{L}}_\text{liq}$, apply some shape-preserving smooth interpolation method which ensures monotonicity (on $\tau$ and $k$ axes) and convexity (on $k$ axis) of the surface. The interpolant function then belongs to $\mathcal{S}(\breve{\mathcal{R}}_\text{liq})$. For example, Fengler and Hin \cite{Fengler2015} use bivariate tensor-product B-spines to estimate a statically arbitrage-free call price surface from observed prices of options on an arbitrarily shaped lattice.
\end{enumerate}

\subsubsection*{Static arbitrage constraints on the liquid lattice}

For notational simplicity, we reindex our liquid lattice $\mathcal{L}_\text{liq}=\{(\tau_i, m^i_j)\}_{i,j}$ to a single sequence $\{(\tau_j, m_j)\}_{j=1}^N$, and denote $c_j = \breve{c}(\tau_j, k_j) = \tilde{c} (\tau_j, m_j = \ln k_j)$.

Cousot \cite{cousot2007} finds practically verifiable conditions for (\ref{eq:call_shape_constraints}), which we have simplified in \cite{Cohen2020}. We can write these constraints in the form $\mathbf{A} \mathbf{c} \geq \widehat{\mathbf{b}}$, where $\mathbf{c} = [c_1 ~\cdots ~c_{N}]^\top \in \mathbb{R}^N$ is the vector of normalised call prices, $\mathbf{A} = (A_{ij}) \in \mathbb{R}^{R \times N}$ and $\widehat{\mathbf{b}} = (\hat{b}_j) \in \mathbb{R}^{R}$ are a known constant matrix and vector. Here $R$ is the number of static arbitrage constraints. The construction of $\mathbf{A}$ and $\widehat{\mathbf{b}}$ depends on the choice of $\mathcal{L}_\text{liq}$, and is given in \cite{Cohen2020}. We assume for simplicity that $\mathcal{L}_\text{liq}$ is fixed over time.

Using our factor representation \eqref{eq:factr_rep_linear} of option prices, we have $\mathbf{c}_t = \mathbf{G}_0 + \mathbf{G}^\top \xi_t$, with $\mathbf{G}$ as in \eqref{eq:mpor_linear_system} and $\mathbf{G}_0$ the vector with entries $G_0(\tau_j, m_j)$. Consequently, the market model allows no static arbitrage among options on the liquid lattice $\mathcal{L}_\text{liq}$ if $\xi_t$ satisfies
\begin{equation}
\mathbf{A} \mathbf{G}^\top \xi_t \geq \mathbf{b} := \widehat{\mathbf{b}} - \mathbf{A G}_0^\top, \text{ for all } t \in [0, T^* - \overline{\tau}].
\label{eq:factor_static_arbitrage}
\end{equation}

The dynamics of the diffusion process $\xi$ are characterised by its drift and diffusion coefficients, which are assumed to be functions of $\tilde{\xi} = (S, \xi)$. We will address how to enforce this constraint on $\xi_t$ when training a model in Section \ref{sec:neural_sde_polytope}.

\subsection{Model inference}

Given an observable normalised option price surface, we estimate the model following two steps:
\begin{enumerate}[leftmargin=*, label=(\roman*)]
\setlength\itemsep{1pt}
    \item we decode the factor representation $\tilde{c}_t(\tau, m) = G_0(\tau, m) + \sum_{i=1}^d G_i(\tau, m) \xi_{it}$ by calibrating the price basis functions $\{G_i\}_{i=0}^d$; 
    \item we infer the drift and diffusion functions of the decoded factor process $\xi$.
\end{enumerate}

While these models are written in a continuous-time setting, inference is made from discretely observed data. Let us assume that we observe a discrete time series of normalised call option prices $\mathbf{c}_t = [\tilde{c}_t(\tau_1, m_1) ~\cdots ~\tilde{c}_t(\tau_N, m_N)]^\top$ at times $0 = t_0 < t_1 < \cdots < t_L = T$, where $\{(\tau_j, m_j)\}_{1 \leq j \leq N} = \mathcal{L}_\text{liq}$, assuming $d \ll N \ll L$. 

\subsubsection*{Decoding the factor representation}
\label{sec:decode_factor_rep}

Projecting the data from $\tilde{c}$ to $\xi$ reduces dimensions by a linear transformation. The matrix representation of this linear projection gives an estimate of the price basis vectors $\{\mathbf{G}_i\}_{i=0,\dots,d}$. Define matrices $\mathbf{C} = (C_{lj}) \in \mathbb{R}^{(L+1) \times N}$ and $\bm\Xi = (\Xi_{li}) \in \mathbb{R}^{(L+1) \times d}$, where $C_{lj} = \tilde{c}_{t_{l-1}}(\tau_j, m_j)$ and $\Xi_{li} = \xi_{i, t_{l-1}}$, respectively. We represent the observed data $\mathbf{C}$, with residuals $\bm\Upsilon \in \mathbb{R}^{(L+1) \times N}$, by
\begin{equation}
	\mathbf{C} = \mathbf{1}_{L+1} \otimes \mathbf{G}_0 + \bm\Xi \mathbf{G} + \bm\Upsilon,
\label{eq:represent_C}
\end{equation}
where $\mathbf{1}_{L+1}$ is an $(L+1)$-vector of ones, and $\otimes$ denotes the outer product of two vectors. We call $\mathbf{C} - \bm\Upsilon$ the \textit{reconstructed prices}, and the Frobenius norm $||\bm\Upsilon||_F$ the \textit{reconstruction error}. For some fixed $d < N$, an ideal factor representation should achieve the following three objectives:
\begin{enumerate}[leftmargin=*, label=(\roman*)]
\setlength\itemsep{1pt}
\item \textbf{Statistical accuracy}: The reconstruction error $||\bm\Upsilon||_F$ should be as small as possible. Among linear dimension reduction techniques, PCA  applied to the observed prices $\mathbf{C}$ gives the optimal reconstruction error. 
\item \textbf{Minimal dynamic arbitrage}: There should be as few violations of the HJM-type drift restrictions (\ref{eq:mpor_linear_system}) as possible.  For this purpose, we estimate $z_t$ from our observed prices $\mathbf{C}$, then construct price basis from principal components of $z_t$, to minimize the discrepancy between $z_t$ and the space spanned by the factors.
\item \textbf{No static arbitrage}:  The reconstructed prices should violate the static arbitrage constraints (\ref{eq:factor_static_arbitrage})  as seldom as possible.
\end{enumerate}

Suppose we construct $d^\text{st}$, $d^\text{da}$ and $d^\text{sa}$ factors for the above-mentioned objectives, respectively, where $d^\text{st} + d^\text{da} + d^\text{sa} = d$. Algorithm \ref{alg:decode_factor} gives the details of the factor decoding process.

\begin{algorithm}[!ht]
\footnotesize
\SetAlgoLined
\SetKwInOut{Input}{Input}\SetKwInOut{Output}{Output}
\Input{Matrix of price data $\mathbf{C} \in \mathbb{R}^{(L+1) \times N}$; number of factors $d$; number of statistical accuracy factors $d^\text{st}$; number of dynamic arbitrage factors $d^\text{da}$; number of static arbitrage factors $d^\text{sa}$; data for $z_t$ in \eqref{eq:mpor_linear_system}, denoted as $\mathbf{Z} \in \mathbb{R}^{(L+1) \times N}$; constraint matrices $\mathbf{A}$, $\mathbf{b}$ from (\ref{eq:factor_static_arbitrage}).}
\Output{Factor data $\bm\Xi \in \mathbb{R}^{(L+1) \times d}$; price basis vectors $\mathbf{G}_i \in \mathbb{R}^N$ for $i=0,\dots, d$.}
\BlankLine

Let $G_{0j} = \frac{1}{L+1} \sum_{l=1}^{L+1} C_{lj}$ for $j=1,\dots,N$ \;

\tcc{Construct factors to minimize dynamic arbitrage.}

Compute the residual data $\mathbf{R}_0 = \mathbf{C} - \mathbf{1}_{L+1} \otimes \mathbf{G}_0$\;

Compute the principal component decomposition of $\mathbf{Z}$ and assign, $\forall i=1,\dots, d^\text{da}$ and $\forall l=1,\dots,L+1$,
\begin{equation*}
    \mathbf{G}_i = i\text{-th principal component of } \mathbf{Z}, ~\Xi_{li} = \langle (\mathbf{R}_0)_l, \mathbf{G}_i\rangle,
\end{equation*}
where $(\mathbf{R}_0)_l$ is the $l$-th row of $\mathbf{R}_0$\;

\tcc{Construct factors to maximize statistical accuracy.}

Compute the residual data $\mathbf{R}_{d^\text{da}} = \mathbf{R}_0 - \sum_{k=1}^{d^\text{da}} \bm\Xi_k \otimes \mathbf{G}_k$, where $\bm\Xi_k = [\Xi_{1k} \cdots \Xi_{L+1,k}]$\;

Compute the principal component decomposition of $\mathbf{R}_{d^\text{da}}$ and assign, $\forall i=d^\text{da} + 1,\dots, d^\text{da} + d^\text{st}$ and $\forall l=1,\dots,L+1$, 
\begin{equation*}
    \mathbf{G}_i = (i-d^\text{da})\text{-th principal component of } \mathbf{R}_{d^\text{da}}, ~\Xi_{li} = \langle (\mathbf{R}_{d^\text{da}})_l, \mathbf{G}_i\rangle,
\end{equation*}
where $(\mathbf{R}_{d^\text{da}})_l$ is the $l$-th row of $\mathbf{R}_{d^\text{da}}$\;

\tcc{Construct factors to minimize static arbitrage.}
\ForEach{$i = d^\text{da} + d^\text{st}, \dots, d-1$}{
	Compute the residual data $\mathbf{R}_{i} = \mathbf{R}_0 - \sum_{k=1}^{i} \bm\Xi_k \otimes \mathbf{G}_k$, where $\bm\Xi_k = [\Xi_{1k} \cdots \Xi_{L+1,k}]$\;
	
	Compute the covariance of the residual data  $\mathbf{\Sigma}_i = \mathbf{R}_i^\top \mathbf{R}_i$ \;
	
	Eigen-decompose $\mathbf{\Sigma}_i = \mathbf{Q}_i \bm\Lambda_i \mathbf{Q}_i^{-1}$ (where $\mathbf{Q}_i$ is a unitary matrix) \;
	
	For some $\mathbf{w} \in \mathbb{R}^{d}$ such that $\|\mathbf{w}\|_2=1$, define
	\begin{equation*}
	    \mathbf{q}_i(\mathbf{w}) = \mathbf{Q}_i \mathbf{w}, ~ \mathbf{s}_i(\mathbf{w}) = \mathbf{R}_{i} \mathbf{q}_i(\mathbf{w}).
	\end{equation*}
	Then solve the following minimization problem:
	\begin{equation}
	    \mathbf{w}^* = \argmin_{\|\mathbf{w}\|_2=1}\Big( - \sum_{l=1}^{L+1} \mathbbm{1}_{\left\{ \langle \mathbf{A} \mathbf{q}_i(\mathbf{w}), \mathbf{s}_i(\mathbf{w})\rangle \geq \mathbf{b} - \mathbf{A} \left(\sum_{k=1}^{i-1} \Xi_{lk} \mathbf{G}_k \right)
	    \right\}} + \lambda \left\| \mathbf{R}_i - \mathbf{s}_i(\mathbf{w}) \otimes \mathbf{q}_i (\mathbf{w}) \right\|_2\Big),
	    \label{eq:factor_opt}
	\end{equation}
	where we require $\lambda < \min \left(1, \frac{1}{\|\mathbf{R}_i \|}\right)$ so that the penalty term is strictly less than 1\;
	
	Let $\mathbf{G}_{i+1} = \mathbf{q}_i(\mathbf{w}^*)$ and $\mathbf{\Xi}_{i+1} = \mathbf{s}_i (\mathbf{w}^*)$ \;
	}
\caption{Decoding the factor representation}
\label{alg:decode_factor}
\end{algorithm}

\begin{remark}
We discuss here some practical aspects of Algorithm \ref{alg:decode_factor}:
\begin{enumerate}[leftmargin=*, label=(\roman*)]
\setlength\itemsep{1pt}
\item The order of fulfilling the three objectives could be altered. However, if the factor construction for no static arbitrage was not placed in the last step, it is possible that any factors defined afterwards could introduce more static arbitrage.
\item The objective function (\ref{eq:factor_opt}) calculates the number of arbitrage inequalities which are violated in the reconstructed data, with an additional penalty term concerning the reconstruction errors. The penalty term plays a role in ensuring uniqueness of the preferred weights $\mathbf{w}$ in the case that multiple weights produce the same number of arbitrage-free samples.
\item The decoding algorithm depends on the choice of the parameters $d^\text{st}$, $d^\text{da}$ and $d^\text{sa}$, which should be determined through some preliminary analysis of the given data. We show an example in the numerical experiments in Section \ref{sec:numerics}.
\item For numerical convenience, it is useful to ensure that the factors constructed are orthonormal (i.e.\ $\mathbf{G}$ is a unitary matrix). This can be done through a simple Gram--Schmidt process once an initial set of factors has been constructed.
\end{enumerate}
\end{remark}

\subsubsection*{Estimating the constrained diffusion process}

Our goal is to estimate the functions $\alpha: \mathbb{R}^{d+1} \rightarrow \mathbb{R}$, $\mu: \mathbb{R}^{d+1} \mapsto \mathbb{R}^{d}$, $\gamma: \mathbb{R}^{d+1} \rightarrow \mathbb{R}$, and $\sigma: \mathbb{R}^{d+1} \mapsto \mathbb{R}^{d \times d}$ of the diffusion process $\tilde{\xi} = (S, \xi)$, given observations of the solution of the SDE (under the objective measure $\mathbb{P}$)
\begin{subequations}
\begin{align}
    \label{eq:factor_S}
    \diff S_t = & \left( \alpha(\tilde{\xi}_t)-q_t \right) S_t \diff t + \gamma(\tilde{\xi}_t) S_t \diff W_{0,t}, \\
    \label{eq:factor_diffusion} \diff \xi_t = & \mu (\tilde{\xi}_t) \diff t + \sigma (\tilde{\xi}_t) \diff W_t.
\end{align}
\label{eq:diffusion_estimation_overall}%
\end{subequations}
To ensure reconstructed prices do not generate static arbitrage, $\mu, \sigma$ in \eqref{eq:factor_diffusion} should be chosen to guarantee that $\xi$ takes values in the (deterministic) polytope given by \eqref{eq:factor_static_arbitrage}.

The SDE system (\ref{eq:diffusion_estimation_overall}) can be divided into two subsystems and estimated independently, as the constraints are only needed in  (\ref{eq:factor_diffusion}), and do not depend on $S$. In the next section, we will discuss an estimation method for the model of $\xi$ (\ref{eq:factor_diffusion}), which is generic enough for estimating the model of $S$ (\ref{eq:factor_S}).

\section{Neural SDE constrained by a polytope}
\label{sec:neural_sde_polytope}

Consider a stochastic process $Y \in \mathbb{R}^d$ which is a solution of the  SDE
\begin{equation}
    \diff Y_t = \mu (Y_t) \diff t + \sigma (Y_t) \diff W_t,
    \label{eq:sde_polytope}
\end{equation}
where $W$ is a standard $d$-dimensional Brownian motion under the measure $\mathbb{P}$. In addition, $\mu, \sigma$ satisfy the usual Lipschitz regularity conditions such that a unique strong solution $Y$ of the SDE exists. Given discrete observations of $Y$, we study the problem of estimating the drift $\mu: \mathbb{R}^d \rightarrow \mathbb{R}^d$ and the diffusion $\sigma: \mathbb{R}^d \rightarrow \mathbb{R}^{d \times d}$ in the model (\ref{eq:sde_polytope}), such that paths of $Y$ are guaranteed to lie within the deterministic convex polytope generated by a finite matrix-vector pair $(\mathbf{V}, \mathbf{b})$, that is,
\[Y_t \in \mathcal{P}:= \{y \in \mathbb{R}^d: \mathbf{V} y \geq \mathbf{b}\}\qquad\text{for all }t\ge 0\quad \mathbb{P}\text{-a.s.}\] 
Without loss of generality (by rescaling), we assume the rows of $\mathbf{V}$ are unit vectors.

\subsection{Nonattainability of a diffusion process}
\label{sec:nonattainability_diffusion}

Let $\mathbf{v}_k$ denote the $k$-th row of $\mathbf{V}$. For each $k$, the diffusion process $Y$ should not enter the domain $\mathcal{P}_k^c := \{y \in \mathbb{R}^d: \mathbf{v}_k^\top y < b_k \}$ at any time $t > 0$, $\mathbb{P}$-a.s. Writing $\mathcal{P}^c := \cup_{k} \mathcal{P}_k^c = \mathbb{R}^d \setminus \mathcal{P}$, we formulate the desired nonattainability of $\mathcal{P}^c$ from its exterior by the solution $Y$ of the SDE (\ref{eq:sde_polytope}) as
\begin{equation}
1-\mathbb{P}\big(Y_t \in \mathcal{P}\quad \text{for all } t>0\big| Y_0 \in \mathcal{P} \big) = \mathbb{P} \big(  \exists t >0 \text{ s.t. } Y_t \in \mathcal{P}^c ~\big| Y_0 \in \mathcal{P} \big) = 0.
\label{eq:nonattainability_formulation}
\end{equation}

Friedman and Pinsky \cite{friedman1973} give a set of sufficient conditions for (\ref{eq:nonattainability_formulation}) in terms of $\mu$ and $\sigma$. It suffices that the normal components of the diffusion and the drift vanish on any boundaries $\partial \mathcal{P}_k : = \{ y \in \mathbb{R}^{d}: \mathbf{v}_k^\top y = \mathbf{b} \}$, and that an additional ``convexity'' relation between the drift and diffusion coefficients on $\partial \mathcal{P}_k$ is imposed. Conveniently, this results in separate conditions for each of our $k$ constraints, which are easily analyzed.

For each $k$, let $\rho_k(y):= \text{dist}(y, \mathcal{P}_k^c)$ be the distance function defined for $y \not\in \mathcal{P}_k^c$, and let $\nu:\mathbb{R}^{d} \rightarrow \mathbb{R}^{d}$ be the outward normal to $\partial \mathcal{P}_k$ (relative to $\mathcal{P}^c_k$, that is, a vector towards the interior of $\mathcal{P}$). Then, restricting to the closure of $\mathcal{P}$, we have $\nu_i = \partial \rho_k / \partial y_i$ on $\partial \mathcal{P}_k$. Defining $a:= \sigma \sigma^\top$, to ensure that $\mathcal{P}_k^c$ is nonattainable we verify
\begin{subequations}
\begin{align}
\label{eq:friedman1973_diffusion} \sum_{i,j=1}^{d} a_{ij} \nu_i \nu_j = 0 & \quad \text{on } \partial \mathcal{P}_k, \\ 
\label{eq:friedman1973_drift} \sum_{i=1}^{d} \mu_i \nu_i + \frac{1}{2} \sum_{i,j=1}^{d} a_{ij} \pderiv{\rho_k}{y_i, y_j} \geq 0 & \quad \text{on }\partial \mathcal{P}_k,
\end{align}
\label{eq:friedman1973}%
\end{subequations}
assuming that $\rho_k$ has second-order derivatives in a neighbourhood of $\partial \mathcal{P}_k$.  This is formally written as Theorem \ref{thm:friedman1973}. 

\begin{theorem}[Friedman and Pinsky (1973) \cite{friedman1973}]
\label{thm:friedman1973}
Suppose $\mu$ and $\sigma$ satisfy Lipschitz and linear growth conditions, and let $Y_t$ be any solution to (\ref{eq:sde_polytope}) with initial data $Y_0 \in \mathcal{P}$. If (\ref{eq:friedman1973}) holds for all $k$, then $\mathcal{P}^c$ is nonattainable by $Y_t$, i.e.\ $\mathbb{P} \left( \left. \exists t >0 \text{ s.t. } Y_t \in \mathcal{P}^c ~\right| Y_0 \in \mathcal{P} \right) = 0$.
\end{theorem}


In our case of interest, the distance function $\rho_k(y)$ measures the (Euclidean) distance between a point $y \in \mathbb{R}^{d}$ and a hyperplane $\partial \mathcal{P}_k $, that is, $\rho_k(y) = \mathbf{v}_k^\top y - b_k$. (Note that $\rho_k(y)$ is guaranteed to be non-negative for all $y$  within the interior of $\mathcal{P}$.) Therefore, $\partial \rho_k (y) / \partial y = \mathbf{v}_k$ and $\partial^2 \rho_k (y) / \partial y_i \partial y_j = 0$ on $\partial \mathcal{P}_k$. Therefore, we are able to re-write the conditions (\ref{eq:friedman1973}) as, for all $k$,
\begin{subequations}
\begin{align}
\mathbf{v}_k^\top a(y) \mathbf{v}_k = 0 & \quad \text{if } \mathbf{v}_k^\top y = b_k,\quad (\text{\emph{no diffusion over the boundary}}); \label{eq:friedman1973_1_1} \\ 
\mathbf{v}_k^\top \mu(y) \geq 0 & \quad \text{if } \mathbf{v}_k^\top y = b_k,\quad  (\text{\emph{inward-pointing drift at the boundary}}).  \label{eq:friedman1973_1_2}
\end{align}
\label{eq:friedman1973_1}%
\end{subequations}
By enforcing these simple linear (in)equality constraints on $a= \sigma\sigma^\top$ and $\mu$, we ensure that our process cannot leave the desired region $\mathcal{P}$.

\subsection{The likelihood function approximation}
\label{sec:likelihood_function}

Denote by $y_{0:T}$ the observed discrete sample path of $Y_t$ at times $0 = t_0 < t_1 < \cdots < t_L = T$. The problem of estimating the model (\ref{eq:sde_polytope}) is formulated as the constrained optimization problem:
\begin{equation}
\max_{\mu, \sigma} ~p(y_{0:T}; \mu, \sigma),
\text{ subject to } (\ref{eq:friedman1973_1_1}) \text{ and } (\ref{eq:friedman1973_1_2}), \text{ for all } k.
\label{eq:inference}
\end{equation}
Here, $p(y_{0:T}; \mu, \sigma)$ is the likelihood (under measure $\mathbb{P}$) of observing the time series $y_{0:T}$ given $\mu$ and $\sigma$. The formulation (\ref{eq:inference}) is an infinite-dimensional optimization problem, where the optimization variables are real-valued Lipschitz functions that are subject to boundary conditions.

We can compute the likelihood $p(y_{0:T}; \mu, \sigma)$ using the marginal density of the finite dimensional observations $y_{0:T}$. However, the transition density is generally not available in closed form. A simple approximation is to use the Euler--Maruyama scheme for the SDE of $Y_t$. Suppose we have a uniform time mesh and let $\Delta t = t_i - t_{i-1} = T / L$ for all $1 \leq i \leq L$. We have the single step approximation
\begin{equation*}
Y_{t_{i+1}} - Y_{t_{i}} \approx \mu(Y_{t_{i}}) \Delta t + \sigma(Y_{t_{i}}) (W_{t_{i+1}}-W_{t_{i}}),
\end{equation*}
where $W_{t_{i+1}}-W_{t_{i}} \sim \mathcal{N} (\mathbf{0}, \mathbf{I}_{d})$ and $0\le i<L$. 
For notional simplicity, let $\mu(i)= \mu(y_{t_{i}})$, $\sigma (i) = \sigma(y_{t_{i}})$ and $a (i) = \sigma (i) \sigma (i)^\top$. The likelihood of observing $y_{t_{i+1}}$ given $y_{t_{i}}$ is then approximately
\begin{equation*}
p(y_{t_{i+1}} | y_{t_{i}}; \mu, \sigma)
\propto
\frac{1}{\sqrt{|a(i)|}}
\exp \left\{ -\frac{\| y_{t_{i+1}} - y_{t_{i}} - \mu(i) \Delta t \|_{a(i)}}{2 \Delta t} \right\},
\end{equation*}
where $|a(i)| = \det(a(i))$, and for $u, ~\nu \in \mathbb{R}^{d}$ and a positive definite matrix $a \in \mathbb{R}^{d \times d}$,
\begin{equation*}
(u, \nu)_a := u^\top a^{-1} \nu, ~~ \|u\|^2_a := (u, u)_a.
\end{equation*}
Due to the Markov property, we have
\begin{equation*}
\begin{split}
p(y_{0:T} | y_0; \mu, \sigma) 
 & = \prod_{i=0}^{L-1} p(y_{t_{i+1}} | y_{t_{i}}; \mu, \sigma) \\
 & \propto \bigg(\prod_{i=0}^{L-1} |a(i)|^{-\frac{1}{2}}\bigg)  \times \exp \bigg\{ -\frac{1}{2 \Delta t} \sum_{i=0}^{L-1} \| y_{t_{i+1}} - y_{t_{i}} \|^2_{a(i)} \bigg\} \\
 & ~~~~ \times \exp \bigg\{ -\frac{1}{2} \sum_{i=0}^{L-1} \Big[ \| \mu(i)\|^2_{a(i)} \Delta t - 2 \big(\mu(i), y_{t_{i+1}} - y_{t_{i}}\big)_{a(i)} \Big] \bigg\}.
\end{split}
\end{equation*}
Taking the log of the approximate likelihood $p$, we have
\begin{equation}
\begin{split}
 & \ln p(y_{0:T} | y_{t_{0}}0; \mu, \sigma) \\ 
\propto & -\frac{1}{2} \sum_{i=0}^{L-1}
\left[
\ln |a(i)| + 
\frac{1}{\Delta t} \| y_{t_{i+1}} - y_{t_{i}} \|^2_{a(i)} + \| \mu(i) \|^2_{a(i)} \Delta t - 2 \left(\mu(i), y_{t_{i+1}} - y_{t_{i}}\right)_{a(i)}
\right].
\end{split}
\label{eq:approx_likelihood}
\end{equation}

\subsection{Deep learning as function approximation algorithm}

Given the approximated likelihood function (\ref{eq:approx_likelihood}), the optimization (\ref{eq:inference}) can be viewed as a \textit{supervised learning} problem, where the input variable is $y_t$ and the output variable is the increment $y_{t+\Delta t} - y_t$.

We will use a neural network to represent functions $\mu(\cdot)$ and $\sigma(\cdot)$ in a non-parametric way, limited only by properties such as continuity and boundedness.
We design a neural network such that $\mu(\cdot)$ and $\sigma(\cdot)$ always satisfy the constraints listed in \eqref{eq:friedman1973_1}. This is achieved by viewing the \textit{target} functions $\mu(\cdot)$ and $\sigma(\cdot)$ as transformations of some \textit{underlying} functions $\hat{\mu}(\cdot)$ and $\hat{\sigma}(\cdot)$ through some operators $\mathcal{G}_\mu$ and $\mathcal{G}_\sigma$, respectively:
\begin{equation*}
    \mu = \mathcal{G}_\mu [\hat{\mu}], \quad \sigma = \mathcal{G}_\sigma [\hat{\sigma}].
    \label{eq:underlying_transform}
\end{equation*}
We construct the operators $\mathcal{G}_\mu$ and $\mathcal{G}_\sigma$ such that \eqref{eq:friedman1973_1} holds for any functions $\hat{\mu}(\cdot)$ and $\hat{\sigma}(\cdot)$, but impose essentially no further restrictions. We then proceed by modelling the functions $\hat{\mu}(\cdot)$ and $\hat{\sigma}(\cdot)$ using a standard neural network. In Figure \ref{fig:nn}, we show an illustration of the proposed constrained neural network model.

\begin{figure}[!ht]
    \centering
    \resizebox{\textwidth}{!}{
    \input{figure/tikz_nn}
    }
    \caption{Constrained neural network.}
    \label{fig:nn}
\end{figure}
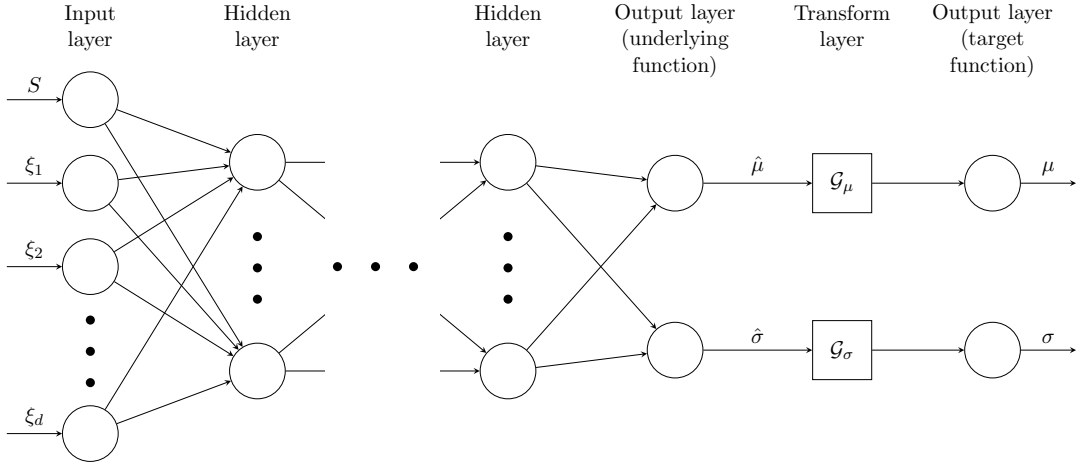

\subsubsection*{Operator for diffusion constraints}

The diffusion function $\sigma$ is unconstrained except when $y$ approaches a boundary of $\mathcal{P}$. More specifically, if $y$ is on the $k$-th boundary, i.e.\ $\rho_k(y) = \mathbf{v}_k^\top y - b_k = 0$, the diffusion matrix $\sigma (y)$ will live in the kernel of $\mathbf{v}_k^\top$, as $(\mathbf{v}_k^\top \sigma)^2 = \mathbf{v}_k^\top a \mathbf{v}_k = 0$. We want to construct a Lipschitz operator $\mathcal{G}_\sigma$ such that, for any non-degenerate locally (Lipschitz) continuous $\hat{\sigma} : \mathbb{R}^{d} \rightarrow \mathbb{R}^{d \times d}$,
\begin{equation*}
    \sigma = \mathcal{G}_\sigma [\hat{\sigma}] \in \ker(\mathbf{v}_k^\top), ~\text{whenever}~ \rho_k(y) = 0.
\end{equation*}

We propose to achieve this by \textit{shrinking} $\mathbf{v}_k^\top \sigma $ to zero as $\rho_k \downarrow 0$. We accordingly construct an operator which is able to act when near multiple boundaries, and should be unaffected by boundaries distant from $y$. Specifically, we define
\begin{equation}
    \sigma(y) = \mathcal{G}_\sigma [\hat{\sigma}](y) := (\mathbf{P}(y))^\top \hat{\sigma} (y),
    \label{eq:constrained_sigma}
\end{equation}
where $\mathbf{P}(y)$ is a linear shrinking matrix constructed using Algorithm \ref{alg:diffusion_shrinking}.

\begin{algorithm}[!ht]
\footnotesize
\SetAlgoLined
\SetKwInOut{Input}{Input}\SetKwInOut{Output}{Output}
\Input{A non-empty $d$-polytope $\mathcal{P} = \{x \in \mathbb{R}^d : \mathbf{V} x \geq \mathbf{b}, \mathbf{V} \in \mathbb{R}^{R \times d}, \mathbf{b} \in \mathbb{R}^d \}$;  a  Lipschitz continuous and monotonically increasing function $h^\sigma: \mathbb{R}_{\geq 0} \rightarrow \mathbb{R} \cap [0,1]$ satisfying $h^\sigma(0)=0$ and $h^\sigma(x) \rightarrow 1$ for $x\rightarrow\infty$; some fixed $y \in \mathbb{R}^d$. }
\Output{The shrinking transformation $\mathbf{P} (y) \in \mathbb{R}^{d \times d}$.}
\BlankLine
\tcc{In our experiments, we use the simple example $h^\sigma(x) = 1 - 1/(1+x)$.}
\ForEach{$k = 1, \dots, R$}{
    Construct the normalised distances  $\varepsilon^\sigma_k = h^\sigma \circ \rho_k(y)$, where $\rho_k$ is the Euclidean distance to the $k$-th boundary of $\mathcal{P}$, as defined in Section \ref{sec:nonattainability_diffusion}.
}

Sort $\{\varepsilon^\sigma_k\}_{k}$ in ascending order, and denote the corresponding permutation of the indices of $\varepsilon^\sigma$ as $\pi$. That is, $\varepsilon^\sigma_{k} = \varepsilon^\sigma_{(\pi(k))}$\;

Keep the first $d$ minimal $\varepsilon^\sigma$s, denoted by $\{\varepsilon^\sigma_{(k)}\}_{1 \leq k \leq d}$, and the associated hyperplane boundary coefficients $\{ \mathbf{v}_{(k)} \}_{1 \leq k \leq d}$. Let $\mathbf{U} = [\mathbf{v}_{(1)} ~\cdots ~\mathbf{v}_{(d)}]$ and assume it has full rank. (To ensure $\mathbf{U}$ is of full rank, we can exclude from the ordering any $k$ such that $\mathbf{v}_{(k)}\in \mathrm{span}\{\mathbf{v}_{(1)}, \dots, \mathbf{v}_{(k-1)}\}$)\;

\tcc{The set $\{ \mathbf{v}_{(k)} \}_k$ represents the $d$ directions on which the projected diffusion components most need to vanish. Therefore, we consider shrinking the diffusion along the orthogonal basis that is formed by these directions according to the order and scale of the corresponding distances $\{\varepsilon^\sigma_{(k)}\}_{k}$.}

Let $\mathbf{q}_1 = \mathbf{v}_{(1)}$\;

\ForEach{$k = 2, \dots, d$}{
    Following the Gram--Schmidt process, let
    \begin{equation}
        \mathbf{q}_k = \frac{1}{\sqrt{1 - \sum_{j=1}^{k-1} \langle \mathbf{q}_j, \mathbf{v}_{(k)} \rangle^2}} \left( \mathbf{v}_{(k)} - \sum_{j=1}^{k-1} \langle \mathbf{q}_j, \mathbf{v}_{(k)} \rangle \mathbf{q}_j \right),
        \label{eq:alg2GramSchmidt}
    \end{equation}
    then $\mathbf{q}_k$ is the unit component of $\mathbf{v}_{(k)}$ that is orthogonal to $\mathrm{span} \{\mathbf{v}_{(1)}, \dots, \mathbf{v}_{(k-1)}\}$\;
}

Let $\mathbf{Q} = [\mathbf{q}_{1} ~\cdots ~\mathbf{q}_{d}]$, be the resulting change of basis matrix 

\tcc{A geometric illustration of the construction of $\mathbf{Q}$ is shown in Figure \ref{fig:diffusion_operator}}

Define
\begin{equation}
    \mathbf{P}(y) = \text{diag} \left( \sqrt{\varepsilon_{(1)}}, \cdots, \sqrt{\varepsilon_{(d)}} \right) \times \mathbf{Q};
    \label{eq:construct_P}%
\end{equation}

\tcc{Unless the directions $\{ \mathbf{v}_{(k)} \}_k$ are orthogonal, shrinking the diffusion along $\mathbf{q}_j$ will shrink it along all directions in $\{ \mathbf{v}_{(k)} \}_{k\geq j}$. This may cause the diffusion along $\mathbf{v}_{(k)}$, for any $k \geq 2$, to be shrunk by a much smaller factor than $\sqrt{\varepsilon_{(k)}}$. To mitigate this compounded shrinking issue, we propose an alternative method to compute adaptive shrinking scales in Appendix \ref{sec:adaptive_shrinking}.}

\caption{Constructing diffusion shrinking transformation matrix}
\label{alg:diffusion_shrinking}
\end{algorithm}

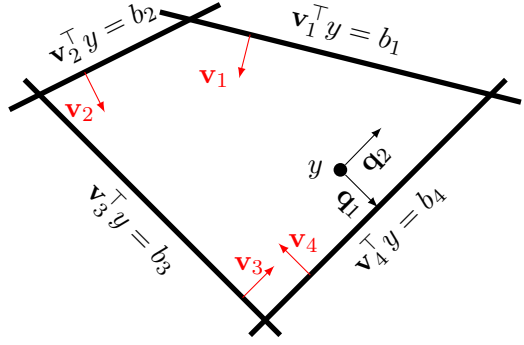
\begin{figure}[!ht]
    \centering
    \input{figure/tikz_diffusion_operator}
    \caption{Orthogonal directions for shrinking the diffusion function. As $y$ is closer to boundary $4$ (bottom right) than the others, direction $\mathbf{q}_1$ points to this boundary, and will be shrunk the most by the linear shrinking matrix $\mathbf{P}(y)$. }
    \label{fig:diffusion_operator}
\end{figure}

\begin{proposition}
Given the operator $\mathcal{G}_\sigma$ defined by (\ref{eq:constrained_sigma}) and (\ref{eq:construct_P}), we have $\sigma = \mathcal{G}_\sigma [\hat{\sigma}] \in \ker(\mathbf{v}_k^\top)$ whenever $\rho_k(y) = 0$, for each $k$ and any function $\hat{\sigma}: \mathbb{R}^d \rightarrow \mathbb{R}^{d \times d}$.
\end{proposition}
\begin{proof}
Suppose that $y$ is on the $k$-th boundary. Recall that in Algorithm \ref{alg:diffusion_shrinking} we define $\pi$ as the permutation of the indices of $\varepsilon^\sigma$ when sorting $\{\varepsilon^\sigma_j\}_{j}$ in ascending order. If $\pi(k) \geq d$, then $y$ is on a vertex of $\mathcal{P}$ with $\varepsilon_{(j)}=0$ for all $j=1,\dots,d$. Hence,
\begin{equation*}
    \mathbf{v}_k^\top a \mathbf{v}_k = \sum_{i,j=1}^d a_{ij} (\mathbf{v}_k^\top \mathbf{q}_i) (\mathbf{v}_k^\top \mathbf{q}_j) \sqrt{\varepsilon_{(i)} \varepsilon_{(j)}} = 0.
\end{equation*}

Otherwise, when $1 \leq \pi(k) < d$, the Gram--Schmidt process ensures that $\mathbf{v}_k = \mathbf{v}_{(\pi(k))} \in \textrm{span}\{\mathbf{q}_1, \dots, \mathbf{q}_{\pi(k)}\} \perp \mathbf{q}_j$ for all $j > \pi(k)$. Hence,
\begin{equation*}
    \mathbf{v}_k^\top a \mathbf{v}_k = \sum_{i,j=1}^{\pi(k)} a_{ij} (\mathbf{v}_k^\top \mathbf{q}_i) (\mathbf{v}_k^\top \mathbf{q}_j) \sqrt{\varepsilon_{(i)} \varepsilon_{(j)}} = 0,
\end{equation*}
as $\varepsilon_{(j)}=0$ for all $j=1,\dots,\pi(k)$. Therefore, $\sigma \in \ker(\mathbf{v}_k^\top)$ as $(\mathbf{v}_k^\top \sigma)^2 = \mathbf{v}_k^\top a \mathbf{v}_k = 0$.
\end{proof}

\subsubsection*{Operator for drift constraints}

As with the diffusion function, the drift function $\mu$ should be unconstrained except when $y$ approaches a boundary of $\mathcal{P}$. The geometric interpretation of the constraint $\mathbf{v}_k^\top \mu \geq 0$ is that the drift is pointing
inwards whenever $y$ is on the boundary $\mathcal{P}_k$. 

As $\mathcal{P}$ is convex, we know that the vector from any point on the boundary $y$ towards any interior point $\zeta_k$ will satisfy $\mathbf{v}_{k}^\top (\zeta_k - y) >0$. Therefore, we construct $\mu$ as a unconstrained locally (Lipschitz) continuous function $\hat{\mu}$, \textit{corrected} by some $\mathcal{P}$-inward-pointed $d$-vectors, i.e.\
\begin{equation}
    \mu (y) = \mathcal{G}_\mu [\hat{\mu}] (y) := \hat{\mu}(y) + \sum_{k} \lambda_k (y) (\zeta_k - y),
    \label{eq:constrained_mu}
\end{equation}
where $\{\zeta_k\}_{k}$ is a collection of interior points in $\mathcal{P}$, and $\lambda_k$ is a scalar weight chosen such that $\zeta_k - y$ will dominate $\hat{\mu}$ when $y$ moves close to the $k$-th boundary. 

In order to ensure our correction acts smoothly as we approach a boundary, for some small $\rho^*>0$, we select $\{\zeta_k\}_k$ from the $\rho^*$-interior points, defined as $\mathcal{P}^{\rho^*} := \{x \in \mathbb{R}^d: \mathbf{V} x \geq \mathbf{b} - \rho^* \mathbf{1}\}$, provided that $\mathcal{P}^{\rho^*} \neq \emptyset$ (which is guaranteed for $\rho^*$ sufficiently small, as long as our inequality constraints are not degenerate). We pre-compute one $\rho^*$-interior point of $\mathcal{P}$ for each boundary (see Algorithm \ref{alg:interior_points} in Appendix \ref{sec:interior_point} for a method for constructing interior points). Ideally, an interior point should be as far from its corresponding boundary as possible, so that the correction vector $\zeta_k - y$ aligns more with the inward orthogonal vector to the corresponding boundary, where the correction is more efficient. An example is shown in Figure \ref{fig:drift_interior}.

\begin{figure}[!ht]
    \centering
    \input{figure/tikz_drift_interior}
    \caption{An example of $\rho^*$-interior points for calculating drift correction directions}
    \label{fig:drift_interior}
\end{figure}
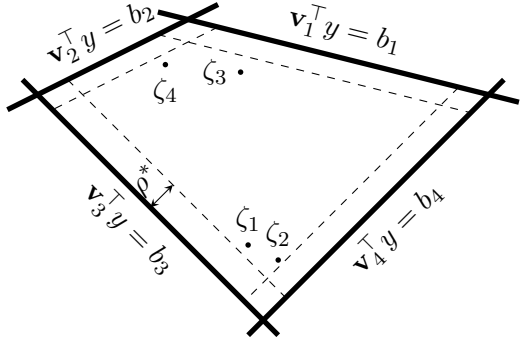

Since the drift constraint is only required on the boundary, we require the correction vectors $\{\zeta_k - y\}_k$ to vanish quickly when $y$ moves away from the boundary. To achieve this, it is sufficient to define
\begin{equation}
    \lambda_k(y) = \left( \frac{\mathbf{v}_k^\top \hat{\mu} - \varepsilon^\mu_k(y)}{\mathbf{v}_k^\top (\zeta_k - y)} \right)^+,
    \label{eq:constrained_mu_lambda}
\end{equation}
where $\varepsilon^\mu_k = h^\mu \circ \rho_k$, for some continuous and monotonically increasing function $h^\mu:\mathbb{R}_{\geq 0} \rightarrow \mathbb{R}_{\geq 0}$ with $h^\mu(0)=0$ and $h^\mu(x) \geq \varepsilon^{\mu*}$ for $x\geq \rho^*$ and some sufficiently large real number $\varepsilon^{\mu*}$. For example, we define $h^\mu(x) = \varepsilon^{\mu*} (e^{x} -1) / (e^{\rho^*} -1)$. 

\begin{proposition}
Given the operator $\mathcal{G}_\mu$ defined by (\ref{eq:constrained_mu}) and (\ref{eq:constrained_mu_lambda}), we have $\mathbf{v}_k^\top \mu \geq 0$ whenever $\rho_k(y)=0$, for each $k$ and any function $\hat{\mu}: \mathbb{R}^d \rightarrow \mathbb{R}^d$.
\end{proposition}

\begin{proof}
For each $k$ and any $\rho^*$-interior point $\zeta$, whenever $\rho_k(y) < \rho^*$, we have
\begin{equation}
    \mathbf{v}_k^\top (\zeta - y) = (\mathbf{v}_k^\top \zeta - b_k) - (\mathbf{v}_k^\top y - b_k) \geq \rho^* - \rho^* = 0.
    \label{eq:correction_property}
\end{equation}
Consequently, by the definition of $\lambda_k$ in (\ref{eq:constrained_mu_lambda}),
\begin{equation*}
    \mu = \hat{\mu} + \sum_{\rho_{k}(y) < \rho^*} \lambda_k (\zeta_k - y) + \sum_{\rho_{k}(y) \geq \rho^*} 0 \cdot (\zeta_k - y).
\end{equation*}

If $y \in \mathcal{P}^{\rho^*}$, we have $\rho_i(y) \geq \rho^*$ for all $k$, which leads to $\mu = \hat{\mu}$. Therefore, no direction correction vectors are imposed on the drift when $y$ is in the $\rho^*$-interior region of $\mathcal{P}$. For any $k$ such that $\rho_{k}(y) < \rho^*$,
\begin{equation*}
    \mathbf{v}_k^\top \mu = \mathbf{v}_k^\top \hat{\mu} + \sum_{\rho_{i}(y) < \rho^*} \lambda_i \mathbf{v}_k^\top (\zeta_i - y) \geq \mathbf{v}_k^\top \hat{\mu} + \lambda_k \mathbf{v}_k^\top (\zeta_k - y) \geq - \varepsilon^\mu_k(y),
\end{equation*}
where we use the property of $\rho^*$-interior points as shown in (\ref{eq:correction_property}) and the definition of $\lambda_k$ in (\ref{eq:constrained_mu_lambda}). It follows that, when $\rho_k(y)=0$, we have $\mathbf{v}_k^\top \mu \geq - \varepsilon^\mu_k(y) = 0$.

\end{proof}

\subsubsection*{Loss functional}

With the approximated likelihood function and the transformed drift and diffusion functions, we are able to re-define the model estimation problem from the constrained optimization problem (\ref{eq:inference}) to the following unconstrained optimization problem:
\begin{equation}
\begin{aligned}
    \min_{\hat{\mu}, \hat{\sigma}} J[\hat{\mu}, \hat{\sigma}] = & \sum_{i=0}^{L-1}
\left[
\ln |a(i)| + 
\frac{1}{\Delta t} \| y_{t_{i+1}} - y_{t_{i}} \|^2_{a(i)} + \| \mu(i) \|^2_{a(i)} \Delta t - 2 \left(\mu(i), y_{t_{i+1}} - y_{t_{i}}\right)_{a(i)} 
\right] \\
 & + \lambda \mathcal{R}(\hat{\mu}, \hat{\sigma}),
\end{aligned}
\label{eq:inference_cost_functional}
\end{equation}
where $\mathcal{R}(\hat{\mu}, \hat{\sigma})$ is a regularization term added to the loss function to encourage useful properties such as sparsity; $\lambda > 0$ is a tuning parameter which controls the importance of the regularization term. Here, $\mu$ and $a$ are constructed from $\hat{\mu}$ and $\hat{\sigma}$ by $\mu = \mathcal{G}_\mu [\hat{\mu}]$ and $a = \mathcal{G}_\sigma[\hat{\sigma}] (\mathcal{G}_\sigma[\hat{\sigma}])^\top$.

Finally, to prevent the likelihood becoming degenerate, it is convenient to ensure $\hat{\sigma}$ (and hence $a$) is of full rank everywhere within $\mathcal{P}$. We do this by parametrising $\hat{\mu}$, $\hat{\sigma}$ using $\phi^\theta : \mathbb{R}^{d} \rightarrow \mathbb{R}^{\frac{1}{2}d(d+3)}$, a neural network with weight parameters $\theta$, and let $\phi^\theta = \left( \phi^\theta_1, \dots, \phi^\theta_{\frac{1}{2}d(d+3)} \right)$. We construct
\begin{equation}
    \hat{\sigma} = \begin{bmatrix}
        \exp(\phi^\theta_1) & 0 & \cdots & 0 \\
        \phi^\theta_2 & \exp(\phi^\theta_3) & \cdots & 0 \\
        \vdots & \vdots & \ddots & \vdots \\
        \phi^\theta_{\frac{1}{2}d(d-1)+1} & \phi^\theta_{\frac{1}{2}d(d-1)+2} & \dots & \exp \left(\phi^\theta_{\frac{1}{2}d(d+1)} \right)
    \end{bmatrix},~
    \hat{\mu} = \begin{bmatrix}
        \phi^\theta_{\frac{1}{2}d(d+1)+1} \\
        \phi^\theta_{\frac{1}{2}d(d+1)+2} \\
        \vdots \\
        \phi^\theta_{\frac{1}{2}d(d+3)} \\
    \end{bmatrix}.
\label{eg:nn_dirft_diffusion}
\end{equation}
Here we choose a lower triangular form for $\hat{\sigma}$ without loss of generality, and its diagonal terms have been exponentiated to ensure positivity. In addition, the lower triangular structure of $\hat{\sigma}$ greatly simplifies the computation of the likelihood part of the loss functional $J$. Specifically,
\begin{equation*}
    \ln |a| = \ln \left|\mathbf{P}^\top \hat{\sigma} \hat{\sigma}^\top \mathbf{P}\right| = \ln \left( |\mathbf{P}|^2 \prod_{i=1}^d \hat{\sigma}_{ii}^2 \right) = 2 \ln |\mathbf{P}| + 2 \sum_{i=1}^d \phi^\theta_{\frac{1}{2}i(i+1)},
\end{equation*}
and for any $u,v\in\mathbb{R}^d$,
\begin{equation*}
    (u, v)_{a} = u^\top \left( \mathbf{P}^\top \hat{\sigma} \hat{\sigma}^\top \mathbf{P} \right)^{-1} v = \left[ \hat{\sigma}^{-1} \left(\mathbf{P}^\top\right)^{-1} u  \right]^\top \left[ \hat{\sigma}^{-1} \left(\mathbf{P}^\top\right)^{-1} v \right]
\end{equation*}
can be computed efficiently, as $\hat \sigma$ is a triangular matrix.

\section{Numerical results -- estimating the market model}
\label{sec:numerics}

We validate our modelling and inference approach with numerical experiments using data generated from a stochastic local volatility model (Jex, Henderson and Wang \cite{Jex1999}), which is a state-of-the-art model frequently used on equity and FX derivatives desks in investment banks. We will here use simulated data, as our aim is to study the effectiveness of these methods in a situation with a well understood ground truth model, where we are not limited by any lack of data. Application to market data has further market-specific complexities, which we shall explore in future work.

\subsection{Input simulation data}
\label{sec:inputsimulation}
The Heston-type Stochastic Local Volatility (Heston-SLV) model we use is given by:

\begin{equation}
    \begin{aligned}
        \diff S_u & = \left( r_u - q_u \right) S_u \diff u + \mathcal{L}_{t}(u, S_u) \sqrt{\nu_t} S_u \diff W_u^S, \\
        \diff \nu_u & = \kappa (\theta - \nu_u) \diff u + \sigma \sqrt{\nu_u} \diff W_u^\nu, \\
        \diff \langle W_u^S, W_u^\nu \rangle & = \rho \diff u,
        \qquad \qquad \qquad \qquad \qquad \qquad \qquad u \in (t,T^*).
    \end{aligned}
    \label{eq:model_hestonslv}
\end{equation}

This model is popular in industry for its realistic dynamic properties through the stochastic volatility component $\sqrt{\nu_t}$, and its ability to exactly reprice vanilla options through fitting a local volatility component, the so-called \emph{leverage function} $\mathcal{L}$. Here, the subscript $t$ of $\mathcal{L}$ indicates that calibration of $\mathcal{L}$ is performed from derivative prices observed at $t$. The calibration typically involves an initial optimisation of the fit over the Heston parameters, followed by a calibration of the leverage function using Markovian projections (see Piterbarg \cite{Piterbarg2006} for a review of early approaches and the references in Bain, Mariapragassam and Reisinger \cite{Bain2021} for more recent contributions).

We obtain the parameters and the leverage function used for generating synthetic input data in our tests by calibrating a Heston-SLV model to Bloomberg OTC USDBRL option price data on 28th October, 2008, using the QuantLib library \cite{quantlib}.
The Heston parameters are listed in Table \ref{tab:params_values},
and we plot the leverage function $\mathcal{L}$ in Figure \ref{fig:leverage_surface}. 

We simulate discrete time series data for $S$, $\nu$ and call options with $N$ fixed values of moneyness and maturity, over a time period of $(L+1)\Delta t$ (see Table \ref{tab:params_values} for the values used). These values are chosen to broadly emulate the observation, at moderate frequency  (of the order of every 10 minutes), of a single year's prices for the underlying asset and liquid European options book.
Here, a single path of $(S,\nu)$ is sampled by an Euler--Maruyama approximation to \eqref{eq:model_hestonslv}
on $[0,T^*]$\footnote{Since the leverage function is only defined on a bounded interval of $S$, we choose a path where the trajectory of $S$ is always within the interval in order to avoid extrapolating the input leverage function. The trajectory of $S$ can be found on the top right graph of Figure \ref{fig:simulation_trajectory}.}, and, for each $t$, option values at $(S_t,\nu_t)$
are produced using the calibrated Heston parameters and the initial portion, i.e. $\mathcal{L}_t(u,\cdot) = \mathcal{L}_0(u-t,\cdot)$, $u \in (t,T^*)$.\footnote{This is consistent with the re-calibration approach in financial practice, which lends justification to our assumption of a stationary model for each fixed maturity.}

This step is only required to generate the simulated input data, but not when real market data are used.
The SDE \eqref{eq:model_hestonslv} is formulated in the risk-neutral measure used for pricing options for each $t$, started at $(S_t,\nu_t)$ along the chosen trajectory. We also use \eqref{eq:model_hestonslv} for generating the $(S_t,\nu_t)$ trajectory on $[0,T]$. Here, any drift could be used, and we choose 0 for simplicity, but note that this does not imply simulation under the risk-neutral measure.\footnote{In fact, selection of a path within a certain range implicitly introduces a drift, and the use of the leverage function translated forward at different $t$, means that the simulation measure differs from the pricing measure.}

\begin{table}[!ht]
\centering
\footnotesize
\begin{tabular}{cccccccccccc}
\toprule
\multicolumn{1}{l}{} & \multicolumn{8}{c}{Heston-SLV model}         & \multicolumn{3}{c}{Simulation}       \\ \cmidrule(lr){1-1} \cmidrule(lr){2-9} \cmidrule(lr){10-12}
Parameter            & $r_t$ & $q_t$ & $S_0$ & $\nu_0$ & $\theta$ & $\kappa$  & $\sigma$ & $\rho$ & $L$ &  $\Delta t$ & $N$  \\
Value                & 0 & 0 & 100 & 0.0083  & 0.0085    & 8.3 & 0.32     & -0.42   & 10000 & 0.0001 & 46  \\
\bottomrule
\end{tabular}
\caption{Parameter values}
\label{tab:params_values}
\end{table}

\begin{figure}[!ht]
    \centering
    \includegraphics[scale=.66]{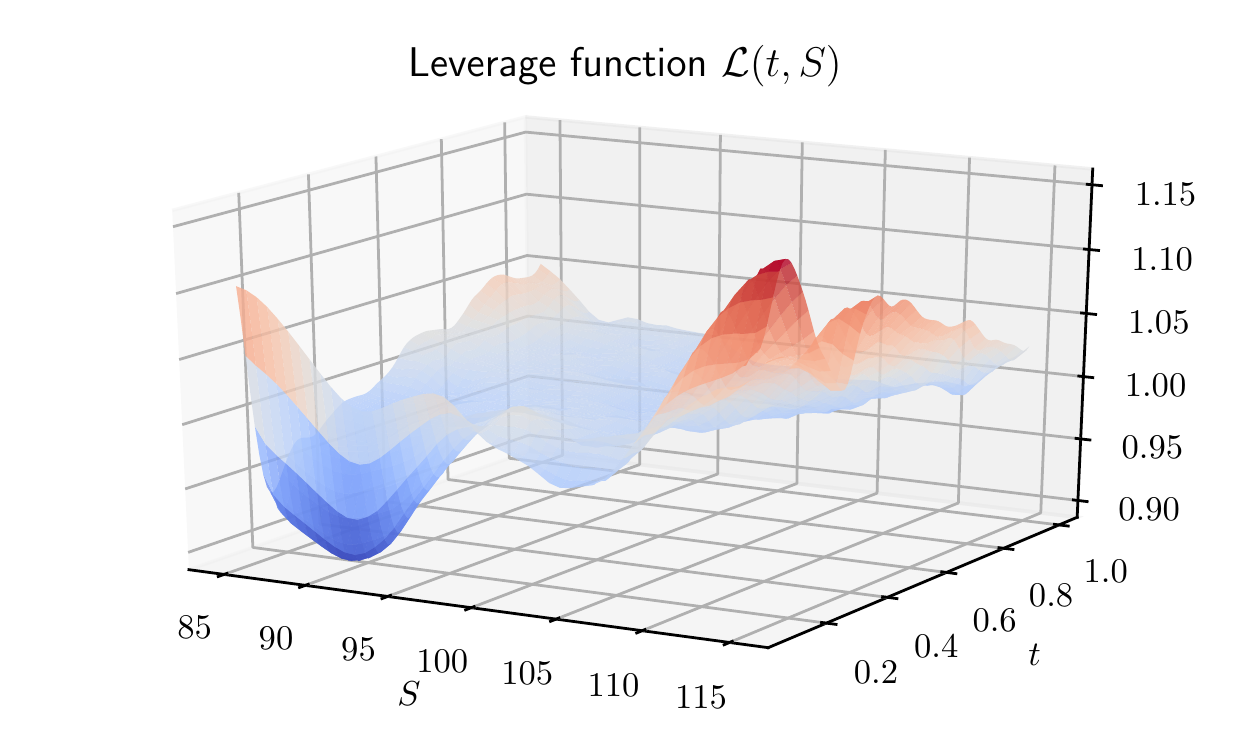}
    \caption{The leverage function $\mathcal{L}(t,S)$ of the Heston-SLV model used for simulation.}
    \label{fig:leverage_surface}
\end{figure}

\subsection{Implementation details of the model estimation strategy}
\label{sec:estimation_strategy}

\paragraph*{Approximation of $z_t$ in \eqref{eq:mpor_linear_system}.}

We need $\mathbf{Z} = (z_t)_{t}$, as inputs for Algorithm \ref{alg:decode_factor} to decode factors. However, neither $\gamma_t$ (see \eqref{eq:market_model} and thereafter) nor the partial derivatives of call option prices, as seen in (\ref{eq:mpor_linear_system}), are readily available for computing $z_t$.

To estimate $\gamma_t$, we fit an \textit{initial} model\footnote{We use the neural SDE estimation method described in Section \ref{sec:neural_sde_polytope}, where the neural network architecture is defined in \eqref{eq:nn_architecture_S}.} for $S$ with state variable $\tilde{\xi} = (S, \xi^{(0)})$, where $\xi^{(0)}$ is the factor projected from the first $d^{(0)}$ principal components of the input call prices. We use a large $d^{(0)}$, such that the reconstruction error is small. This initial calibration gives us a candidate volatility $\hat{\gamma}_t^{(0)}$. 

Next, we need to evaluate the partial derivatives $\partial \tilde{c} / \partial \tau$, $\partial \tilde{c} / \partial m$ and $\partial^2 \tilde{c} / \partial m^2$ on the option lattice $\mathcal{L}_{\mathrm{liq}} =\{(\tau_j, m_j)\}_{j=1,\dots,N}$. We interpolate these values using $C^{1,2}$ basis functions\footnote{One could also use shape-preserving interpolation or convex regression to ensure monotonicity and convexity of the surface. This is likely to induce significant additional computational costs, as the interpolation method would depend on the state of the risk factors $\xi$.}, from which we can approximate the required derivatives. In Figure \ref{fig:call_derivatives}, we show the approximated partial derivatives at $\{(\tau_j, m_j)\}_{j=1,\dots,N}$ for some fixed $t$. With $\hat{\gamma}_t^{(0)}$ and the approximated partial derivatives, we can compute $z_t$.

\begin{figure}[!ht]
    \centering
    \includegraphics[scale=.76]{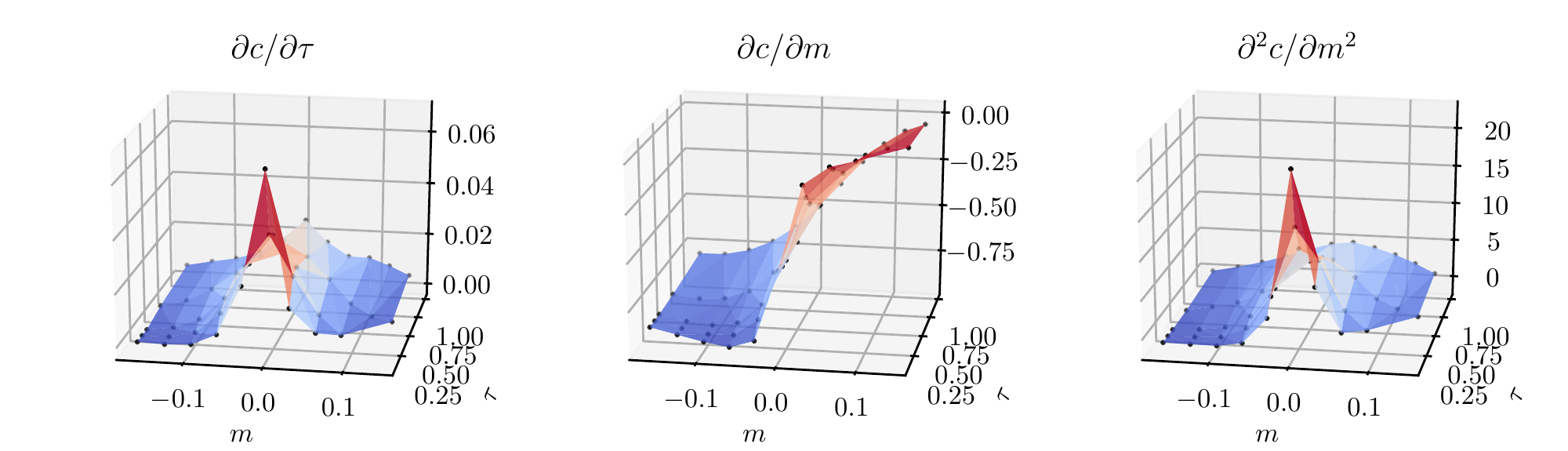}
    \caption{Partial derivatives of the normalised call price surface approximated by interpolation at $t=0$.}
    \label{fig:call_derivatives}
\end{figure}

\paragraph*{Decorrelation and normalisation of factors.}

Though the price basis decoded from Algorithm \ref{alg:decode_factor} is orthogonal by construction, i.e.\ $\mathbf{G} \mathbf{G}^\top = \mathbf{I}_d$, the corresponding factors $\bm\Xi$ could still have high correlations. To enhance the explanatory capability of factors, we decorrelate the factor data by considering the full principal component decomposition $\bm\Xi^{(1)} = \bm\Xi \mathbf{H}$, where columns of $\mathbf{H}$ are eigenvectors of $\bm\Xi^\top \bm\Xi$.

In addition, data of different factors could distribute over very different scales. To improve numerical stability of model training, we further normalise the decorrelated factor data $\bm\Xi^{(1)}$ as $\bm\Xi^{(2)} = \bm\Xi^{(1)} \bm\Lambda^{-1}$, where
\begin{equation*}
    \bm\Lambda = \textrm{diag} \left( \lambda_1, \dots, \lambda_d \right), ~ \text{with } \lambda_j = \lambda_0 \left[\max_l \left(\bm\Xi^{(1)}_{lj}\right) - \min_l\left(\bm\Xi^{(1)}_{lj}\right) \right].
\end{equation*}
Here $\lambda_0$ is a postive constant; we choose $\lambda_0 = 10$ in our numerical example so that the min-max ranges of all normalised factor data are $1/\lambda_0 = 0.1$. Consequently, the factor representation of call price data as described in (\ref{eq:represent_C}) can be re-written as
\begin{equation*}
    \mathbf{C} = \mathbf{1}_{L+1} \otimes \mathbf{G}_0 + \bm\Xi^{(2)} \mathbf{G}^{(2)} + \bm\Upsilon, ~\text{where } \mathbf{G}^{(2)} = \bm\Lambda \mathbf{H}^\top \mathbf{G}.
\end{equation*}
We will build models for the decorrelated and normalised factors. In the following sections, for notational simplicity, we let $\bm\Xi \leftarrow \bm\Xi^{(2)}$ and $\mathbf{G} \leftarrow \mathbf{G}^{(2)}$. 

\paragraph*{Parametrisation of the drift $\mu_t$.}

To improve the efficiency of the maximum likelihood estimator of the drift $\mu_t$ for the model of $\xi$, we provide a baseline estimator of the drift and use neural network only as a modification to the baseline model. Specifically, rather than constructing $\hat{\mu}$ with the neural network $\phi^\theta$ directly as in (\ref{eg:nn_dirft_diffusion}), we define, for some baseline drift function $\bar{\mu} (\tilde{\xi}) : \mathbb{R}^{d+1} \rightarrow \mathbb{R}^d$,
\begin{equation*}
    \hat{\mu} = \textrm{diag} \left( \phi^\theta_{\frac{1}{2}d(d+1) + 1}, \dots,  \phi^\theta_{\frac{1}{2}d(d+3)} \right)
    \times \bar{\mu}.
\end{equation*}

A sensible baseline drift model can be built based on the assumption that the input option data do not permit dynamic arbitrage, i.e. the HJM drift restriction (\ref{eq:mpor_linear_system}) holds, and that the market price of risk $\psi_t$ is small. In particular, assuming zero market price of risk $\psi_t \equiv 0$, the drift restriction (\ref{eq:mpor_linear_system}) implies\footnote{The normalisation of factors results in $\mathbf{G} \mathbf{G}^\top = \bm\Lambda^2$, rather than an identity matrix.} that $\mu_t = \bm\Lambda^{-2} \mathbf{G} z_t$. In the following tests, we first perform a regression on the data $\{(\tilde{\xi}_{t_l}, \bm\Lambda^{-2} \mathbf{G} z_{t_l})\}_l$ to obtain a baseline drift model $\tilde{\xi} \mapsto \bar{\mu} (\tilde{\xi})$.

\subsection{Factor construction and no-arbitrage boundaries}

We apply Algorithm \ref{alg:decode_factor} to decode factors from the normalised call prices that are generated from the Heston-SLV model. To see how well the decoded factors reconstruct the input data, we examine the following metrics:
\begin{itemize}
\setlength\itemsep{1pt}
\item Mean absolute percentage error (MAPE):
\begin{equation*}
    \text{MAPE} = \frac{1}{(L+1)N} \sum_{l=0}^L \sum_{j=1}^N \frac{\left| \tilde{c}_{t_l}(\tau_j, m_j) - G_{0j} - \sum_{i=1}^d G_{ij} \xi_{it_l} \right|}{\tilde{c}_{t_l}(\tau_j, m_j)}.
\end{equation*}
\item Proportion of dynamic arbitrage (PDA):
\begin{equation*}
    \text{PDA} = 1 - \frac{\tr \left(\bm\Lambda^{-1} \mathbf{G}\widetilde{\mathbf{Z}}^\top \widetilde{\mathbf{Z}}\mathbf{G}^\top \bm\Lambda^{-1} \right)}{\tr \left(\widetilde{\mathbf{Z}}^\top \widetilde{\mathbf{Z}} \right)},
\end{equation*}
where $\widetilde{\mathbf{Z}} = (\tilde{z}_{lj}) \in \mathbb{R}^{(L+1) \times N}$, with $\tilde{z}_{lj} = z_{t_l j} - \frac{1}{L+1} \sum_{k} z_{t_k j}$. This metric gives the fraction of variance of $z_t$ that is unexplained by the constructed factors.
\item Proportion of statically arbitrageable samples (PSAS):
\begin{equation*}
    \text{PSAS} = 1 - \frac{ \sum_{l=0}^L \bm{1}_{\{\mathbf{A} \mathbf{G}^\top \xi_{t_l} \geq \mathbf{b} \}}}{L+1}.
\end{equation*}
\end{itemize}

We show the three metrics for a few combinations of factors in Table \ref{tab:factor_metrics}. Using only two factors, i.e.\ one dynamic arbitrage factor and one static arbitrage factor\footnote{The minimization problem \eqref{eq:factor_opt} in the factor decoding algorithm has a discrete objective function, where gradient-based optimisation methods do not apply. We use Py-BOBYQA (Cartis, Fiala, Marteau and Roberts \cite{cartis2019}), a derivative-free optimization solver, to find its global optimum heuristically.} (last row in the table), can represent the whole collection of call prices with reasonable accuracy. Note that the use of a static arbitrage factor is significant in reducing the number of violations of the static arbitrage constraints, as evidenced by the corresponding reduction in PSAS. 

\begin{table}[!ht]
    \centering
    \footnotesize
    \begin{tabular}{lccc}
        \toprule
        Factors & MAPE & PDA & PSAS  \\
        \cmidrule(lr){1-1} \cmidrule(lr){2-4}
        Dynamic arb. & $24.37\%$ & $3.51\%$ & $60.67\%$  \\
        Dynamic arb. + Statistical acc. & $5.11\%$ & $3.21\%$ & $28.11\%$ \\
        Dynamic arb. + Static arb. & $3.85\%$ & $3.04\%$ & $0.37\%$ \\
        \bottomrule
    \end{tabular}
    \caption{MAPE, PDA and PSAS metrics when including different combinations of factors.}
    \label{tab:factor_metrics}
\end{table}

As a proof-of-concept, we will restrict our attention to this simple three-factor model (i.e.\ two vectors $\xi$, in addition to the stock price), as it also allows us to demonstrate qualitative features of the model easily. We plot the price basis functions of these two factors, denoted as $G_1$ and $G_2$, as well as $G_0$, the constant term of $\tilde{c}$, in Figure \ref{fig:Gs}. The points in the liquid lattice (in $(\tau, m)$ coordinates)\footnote{These lattice points are chosen to mimic a typical liquid range of EURUSD options traded at CME, as seen in Figure \ref{fig:optionLattice}.} are also shown on this plot. Here the real-valued functions $G_1$ and $G_2$ are obtained by interpolating the price basis vectors $\mathbf{G}_1$ and $\mathbf{G}_2$, and $G_0$ is obtained by interpolating the normalised call prices averaged over time. Since the minimal option expiry in the input data is nonzero, there is no guarantee that the terminal payoff convergence condition (\ref{eq:initial_G}) will be satisfied by simply extrapolating the input data along expiry. Hence, we introduce a few artificial data points at $\tau=0$ for $G_i$, where $i=0,1,2$, that satisfy (\ref{eq:initial_G}) before interpolation.

\begin{figure}[!ht]
    \centering
    \includegraphics[scale=.76]{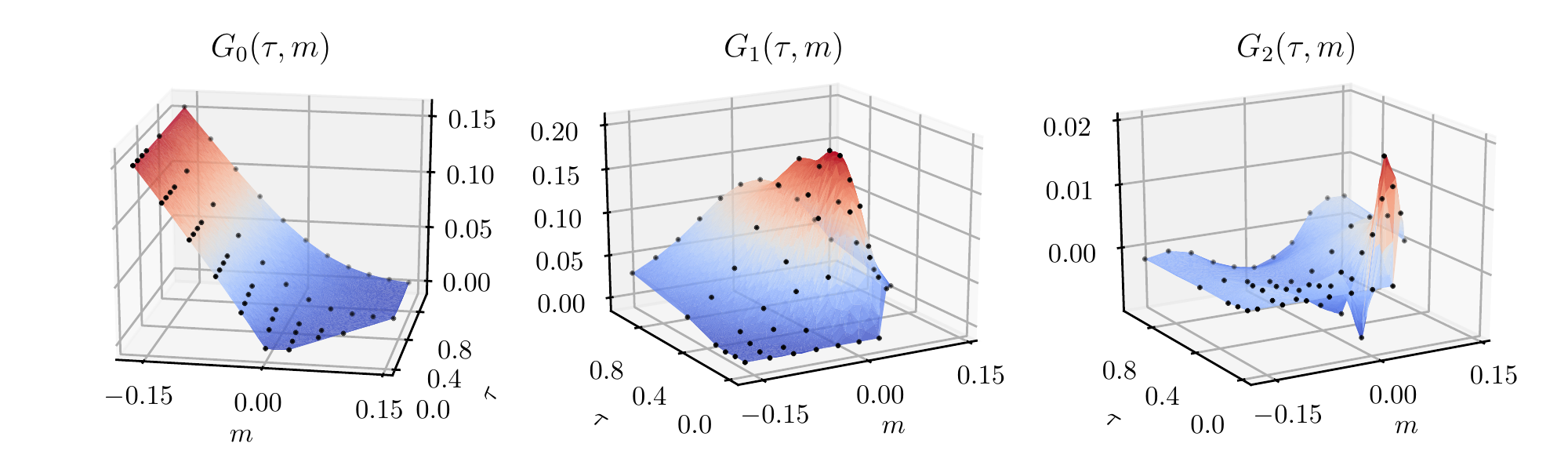}
    \caption{Price basis functions of the normalised call price surface.}
    \label{fig:Gs}
\end{figure}

Given these factors, we use the linear programming method (Caron, McDonald and Ponic \cite{caron1989}) to eliminate redundant constraints in the system $\mathbf{A} \mathbf{G}^\top \xi \geq \mathbf{b}$, which is the projection of the original no-arbitrage constraints, constructed in price space, to the $\mathbb{R}^2$ factor space. This results in only 7 constraints, which we indicate as red dashed lines in Figure \ref{fig:factors}. The convex polygonal domain surrounded by these constraints (light green area) is the statically arbitrage-free zone for the factors, that is, provided the factor process remains in this region, we are guaranteed to have no static arbitrage in the reconstructed call prices on our liquid lattice $\mathcal{L}_{\text{lip}}$. 

\begin{figure}[!ht]
    \centering
    \includegraphics[scale=.66]{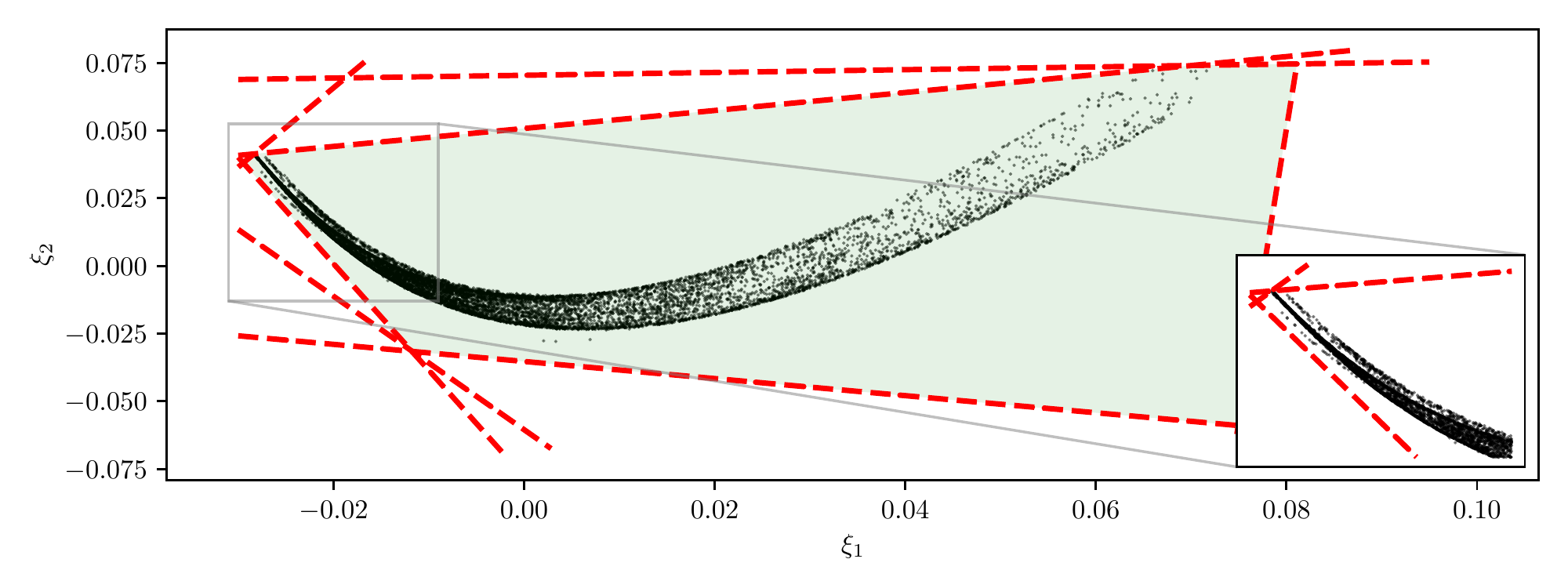}
    \caption{Trajectory (black dots) of the $\mathbb{R}^2$ factors and the corresponding static arbitrage constraints (red dashed lines) projected to the $\mathbb{R}^2$ factor space.}
    \label{fig:factors}
\end{figure}

The data we use for training the market model should be in the interior of the arbitrage-free factor space, otherwise we either have data exhibiting model-free arbitrage, or (if points are on the boundary) the covariance matrix $a$ will be singular, causing a failure in computing the likelihood part of the loss function $J$. Therefore, we truncate the factor data by removing those observations outside or on the boundary of the no-arbitrage region\footnote{We are solving a supervised learning problem with the input-label data as $(x_l, y_l) = (\tilde{\xi}_{t_l}, \xi_{t_{l+1}} - \xi_{t_{l}})$ for $l=1,\dots,L$. Hence, if $\xi_{t_l}$ is an arbitrageable data point, we need to remove two samples $(\tilde{\xi}_{t_{l-1}}, \xi_{t_{l}} - \xi_{t_{l-1}})$ and $(\tilde{\xi}_{t_l}, \xi_{t_{l+1}} - \xi_{t_{l}})$ from the training data. For our training data, this leads to removing 49 samples given that we have observed 37 arbitrageable data points (after factor reconstruction) in our simulation.}.

\subsection{Neural network training results}
\label{sec:training_results}

We apply the method in Section \ref{sec:neural_sde_polytope} to estimate the dynamics of the decoded factors. Specifically, we let $y = \xi$ and $\mathbf{V} = \mathbf{AG}^\top$ and, by slight abuse of notation, $\mu$ and $\sigma$ take $\tilde{\xi} = (S, \xi)$ as their argument rather than $\xi$. This does not alter the estimation method, except that the input layer of the neural network shall consist of $d+1$, instead of $d$, neurons. 

For the neural-SDE model of $\xi$, we use a simple architecture\footnote{The robustness of the estimated neural network is assessed via sensitivity analysis in Appendix \ref{sec:nn_sensitivity}.} that is a composition of fully-connected layers and activation functions in the following orders:
\begin{equation}
\phi^\theta = \mathcal{F}_{d+1} \circ \mathcal{A}_\text{ReLU} \circ \mathcal{F}_{256} \circ \mathcal{A}_\text{ReLU} \circ \mathcal{F}_{256} \circ \mathcal{A}_\text{ReLU} \circ \mathcal{F}_{256},
\label{eq:nn_architecture_xi}
\end{equation}
where $\mathcal{F}_x$ is a fully connected layer, or affine transformation, with $x$ units, and $\mathcal{A}_\text{xxx}$ is an activation function. Each layer $\mathcal{F}$ is parametric, but we omit the parameters for notational simplicity. We use a smaller network $\phi^{S,\theta}: \mathbb{R}^{d+1} \rightarrow \mathbb{R}^2$ for the neural-SDE model of $S$:
\begin{equation}
    \phi^{S,\theta} = \mathcal{F}_{d+1} \circ \mathcal{A}_\text{ReLU} \circ \mathcal{F}_{128} \circ \mathcal{A}_\text{ReLU} \circ \mathcal{F}_{128} \circ \mathcal{A}_\text{ReLU} \circ \mathcal{F}_{128}.
    \label{eq:nn_architecture_S}
\end{equation}
We train the model for $S$ separately from that for $\xi$. In addition, to mitigate over-fitting problems, we train both networks with $50\%$ sparsity, meaning that the $50\%$ smallest weights are pruned to zero. We implement and train our model using the standard tools within the Tensorflow \cite{tensorflow2015-whitepaper} environment.

In Figure \ref{fig:loss_history}, we show the evolution of training losses and validation losses\footnote{The first $90\%$ of the dataset is used for training and the last $10\%$ is reserved as validation data.} over epochs during the training of the models for $S$ and $\xi$, respectively. For the model of $S$, the loss value quickly drops during the first 10 epochs, and slowly declines and converges within 100 epochs. The loss function for the model of $\xi$ is much more complex than that for the model of $S$ (due to the higher dimensionality and transformations near boundaries), so we train it for significantly more epochs. Similar to the loss history for the model of $S$, the loss value for the model of $\xi$ has a rapid decline for the first a few epochs, and then gradually converges.

\begin{figure}[!ht]
    \centering
    \includegraphics[scale=.66]{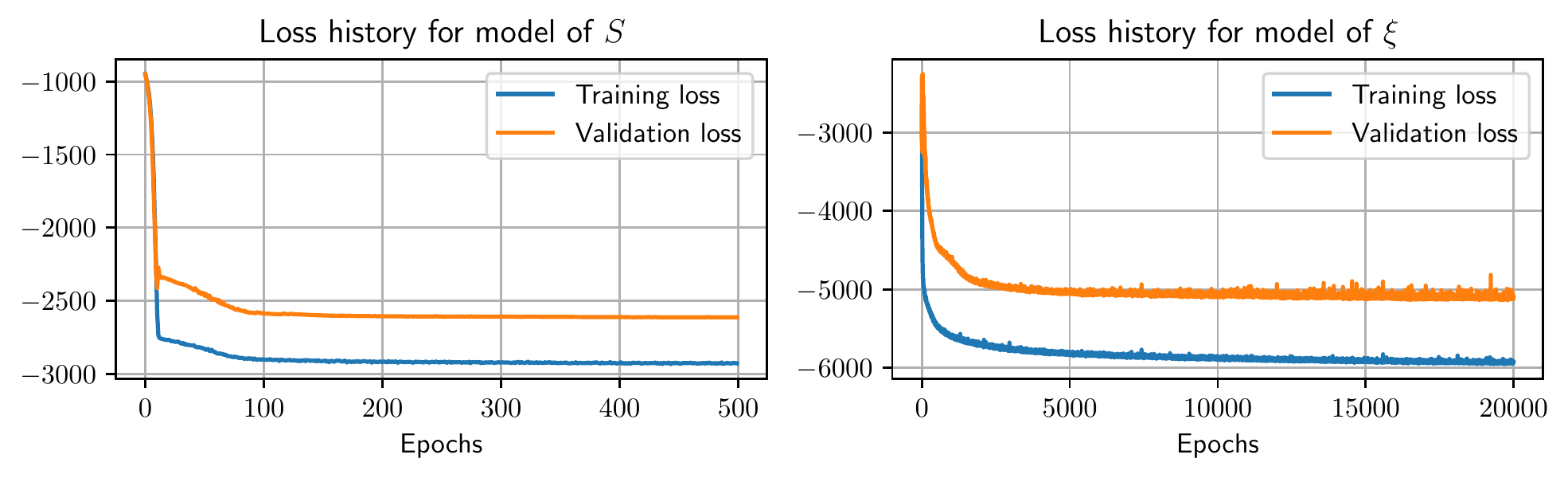}
    \caption{Evolution of training losses and validation losses.}
    \label{fig:loss_history}
\end{figure}

In addition to observing the convergence of the loss functions, we need to demonstrate that the model has been trained sensibly compared with the ground truth Heston-SLV model, and that the learnt model is capable of generating data similar to the input data. We consider the following in-sample and out-of-sample tests.

\subsection{In-sample test}

We assess how well the learnt model recovers the ground truth Heston-SLV model (\ref{eq:model_hestonslv}). Specifically, we study $\sigma_t^S = \mathcal{L}(t,S_t) \sqrt{\nu_t} S_t$ (the diffusion coefficient of $S$), $\mu_t = \kappa(\theta - \nu_t)$ (the drift coefficient of $\nu$) and $\sigma_t = \eta \sigma \sqrt{\nu_t}$ (the diffusion coefficient of $\nu$), and compare values deduced from the input data against those generated by the learned model.

\paragraph*{Model for $S$.}

The neural network for the model of $S$ is written as $\phi^{S,\theta} = (\phi^{S, \theta}_\mu, \phi^{S, \theta}_\sigma)$, where $\phi^{S, \theta}_\mu, \phi^{S, \theta}_\sigma: \mathbb{R}^{d+1} \rightarrow \mathbb{R}$ give the approximation to the drift and diffusion coefficients of $S$, respectively. Using in-sample data $\{\tilde{\xi}_{t_l} = (S_{t_l}, \xi_{t_l})\}_{l=0,\dots,L}$, we can compute $\{\phi_\mu^{S, \theta} (\tilde{\xi}_{t_l})\}_{l}$ and $\{\phi_\sigma^{S, \theta} (\tilde{\xi}_{t_l})\}_{l}$.

It is well known that the estimation of drift terms in an SDE, from observation of a single time series over a short period, is prone to significant error. We hope that the drift estimates will be close to $0$, the ground truth drift of $S$. However, as seen in Figure \ref{fig:test_model_S_mu}, there are some deviations from $0$, as well as some polynomial relations with the factors. Nevertheless, these deviations from $0$ are small when compared to the size of the estimated diffusion, as will be seen later, and the polynomial relations with regards to the factors are possibly also artifacts introduced by selecting a path where $S$ does not leave the domain of the leverage function.

The diffusion estimates $\{\phi_\sigma^{S, \theta} (\tilde{\xi}_{t_l})\}_{l}$ approximate $\{ \sigma^S_{t_l} \}_{l}$, which is generated from the ground truth Heston-SLV model. As we see in Figure \ref{fig:test_model_S}, the volatility of $S$ has been essentially correctly captured, with some noise and a slight upward bias for very low volatilities. We observe a $4.96\%$ mean absolute percentage error in this volatility estimate, computed as $\frac{1}{L+1} \sum_{l} | \phi_\sigma^{S, \theta} (\tilde{\xi}_{t_l}) -  \sigma^S_{t_l}| / \sigma^S_{t_l}$.

We could also use these methods to perform an out-of-sample calibration check, provided that we can simulate $\tilde{\xi}_t$ realistically -- this is indeed verified in Section \ref{sec:outsample_test}; here we present in-sample performance to avoid any potential biases introduced by simulation methods. 

\begin{figure}[!ht]
    \centering
    \begin{subfigure}[b]{.68\textwidth}
    \centering
        \includegraphics[scale=.66]{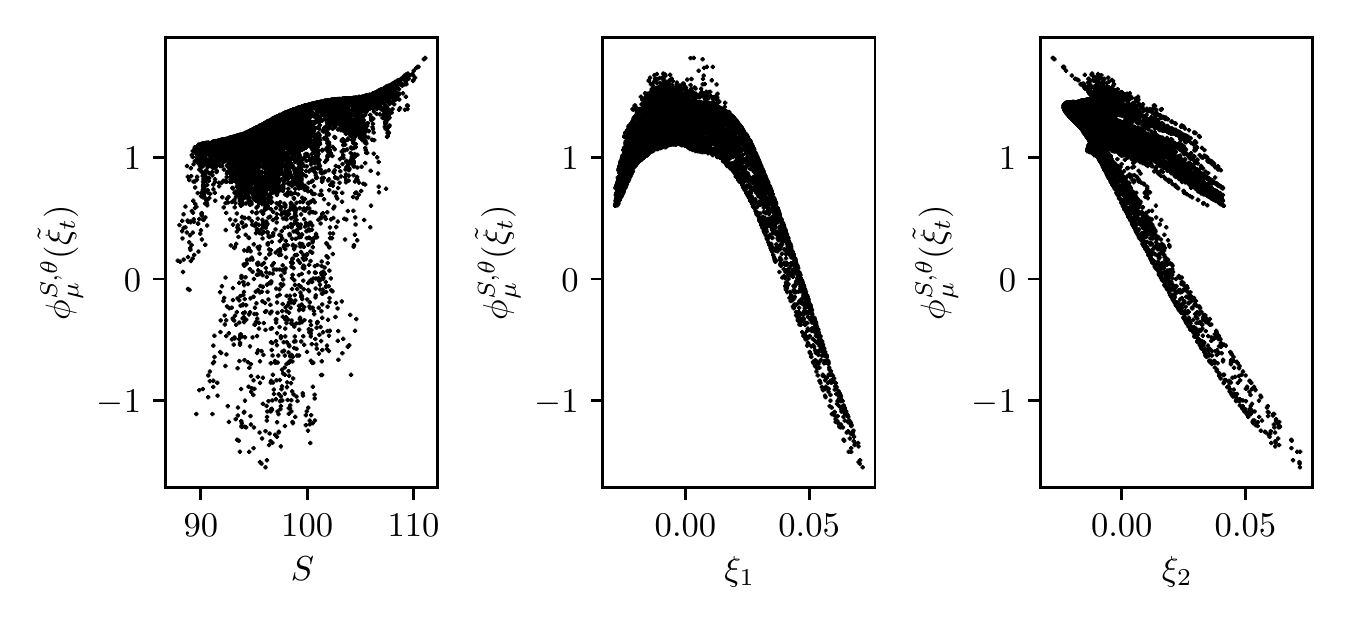}
        \caption{Scattergrams of $\phi_\sigma^{S, \mu} (\tilde{\xi}_{t_l})$ against $S$, $\xi_1$ and $\xi_2$. The estimated drift $\phi_\sigma^{S, \mu} (\tilde{\xi}_{t_l})$ has values around $0$, the ground-truth drift.}
        \label{fig:test_model_S_mu}
    \end{subfigure}
    \hfill
    \begin{subfigure}[b]{.31\textwidth}
    \centering
        \includegraphics[scale=.66]{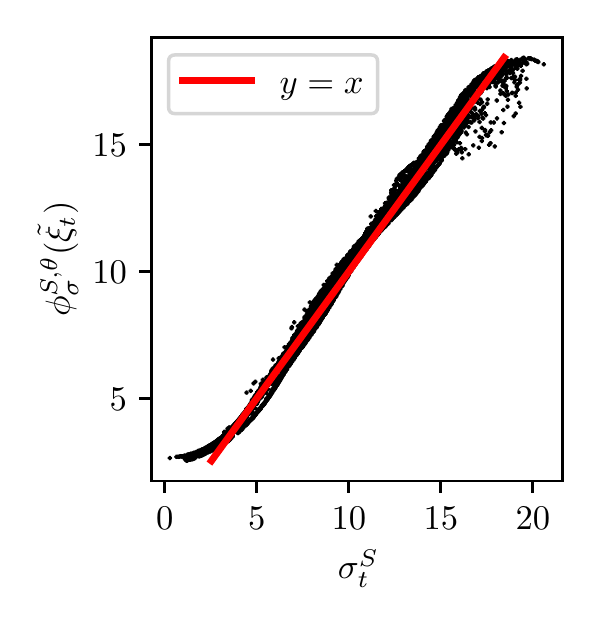}
        \caption{Scattergram of $\phi_\sigma^{S, \theta} (\tilde{\xi}_{t_l})$ against $\sigma^S_{t}$.}
        \label{fig:test_model_S}
    \end{subfigure}
    \caption{Estimated drift and diffusion coefficients for $S$.}
\end{figure}

\paragraph*{Model for $\xi$.}

It is non-trivial to derive the ground-truth model for the factor $\xi$, given that the input data are generated from a Heston-SLV model. Nevertheless, the first calibrated factor $\{\xi_{1t_l}\}_{l=0,\dots,L}$ has a very strong linear relationship with the Heston-SLV simulated variance process $\{\nu_{t_l}\}_{l=0,\dots,L}$, as shown in Figure \ref{fig:xi1nu}.  Using this, an approximated ground-truth model for $\xi_1$ can be used to benchmark the learnt neural network model.

\begin{figure}[!ht]
    \centering
    \centering
        \includegraphics[scale=.66]{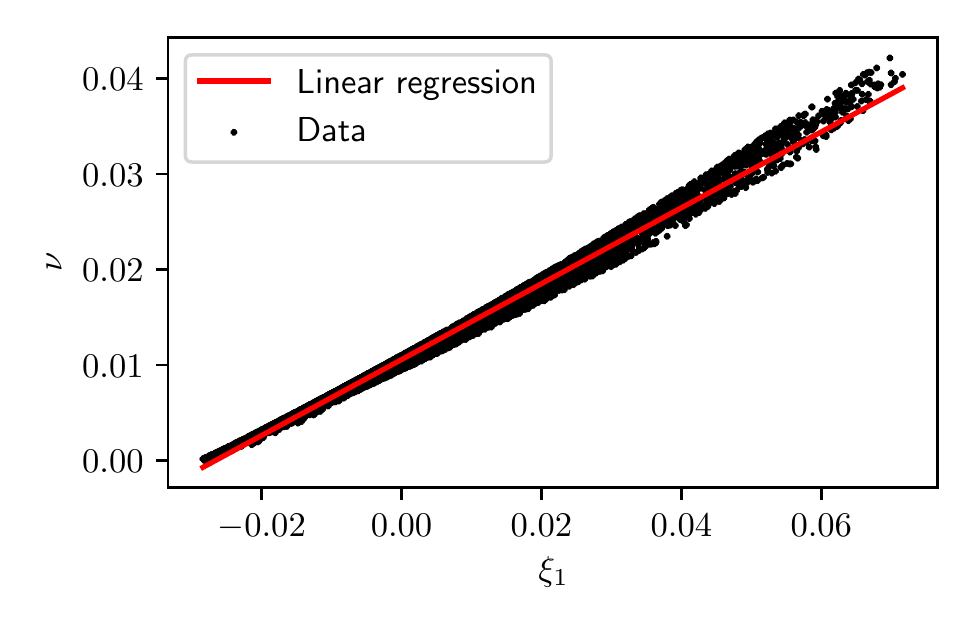}
        \caption{The linear relationship between $\xi_1$ and $\nu$.}
        \label{fig:xi1nu}
\end{figure}

Assuming a simple linear model $\nu = \beta_0 + \beta_1 \xi_1$ with $\beta_1 \neq 0$, we apply It\^{o}'s lemma to the SDE for $\nu$ in \eqref{eq:model_hestonslv} to get
\begin{equation}
    \diff \xi_{1t} = \mu^{\xi_1}(\xi_{1t}) \diff t + \sigma^{\xi_1}(\xi_{1t}) \diff W_t^\nu,
    ~\mu^{\xi_1}(x) = \frac{\kappa (\theta - \beta_0 -\beta_1 x)}{\beta_1}, ~
    \sigma^{\xi_1}(x) = \frac{\eta \sigma \sqrt{\beta_0 + \beta_1 x}}{\beta_1} .
    \label{ref:model_xi1}
\end{equation}

Let $\phi^{\xi_1,\theta}_{\mu}: \mathbb{R}^{d+1} \rightarrow \mathbb{R}$ and $\phi^{\xi_1,\theta}_{\sigma}: \mathbb{R}^{d+1} \rightarrow \mathbb{R}$ be the components of the neural network model that approximate the drift and diffusion functions of $\xi_1$, respectively. Using in-sample data $\{\tilde{\xi}_{t_l} = (S_{t_l}, \xi_{t_l})\}_{l=0,\dots,L}$, we can compute $\{\phi_\mu^{\xi_1, \theta} (\tilde{\xi}_{t_l})\}_{l}$ and $\{\phi_\sigma^{\xi_1, \theta} (\tilde{\xi}_{t_l})\}_{l}$, which are supposed to approximate $\{ \mu^{\xi_1}(\xi_{1t_l}) \}_{l}$ and $\{ \sigma^{\xi_1}(\xi_{1t_l}) \}_{l}$, which were generated from the ground truth model (\ref{ref:model_xi1}). We plot and compare these data against $\{\xi_{1t_l}\}_{l}$ in Figure \ref{fig:test_model_xi1}. We observe that the neural network model has captured the ground truth model well for $\xi_1 < 0.04$, where there are rich data for training. The square-root behavior in the diffusion of $\nu$ has been particularly well captured, as has the decreasing drift, despite the presence of noise.

\begin{figure}[!ht]
    \centering
    \includegraphics[scale=.66]{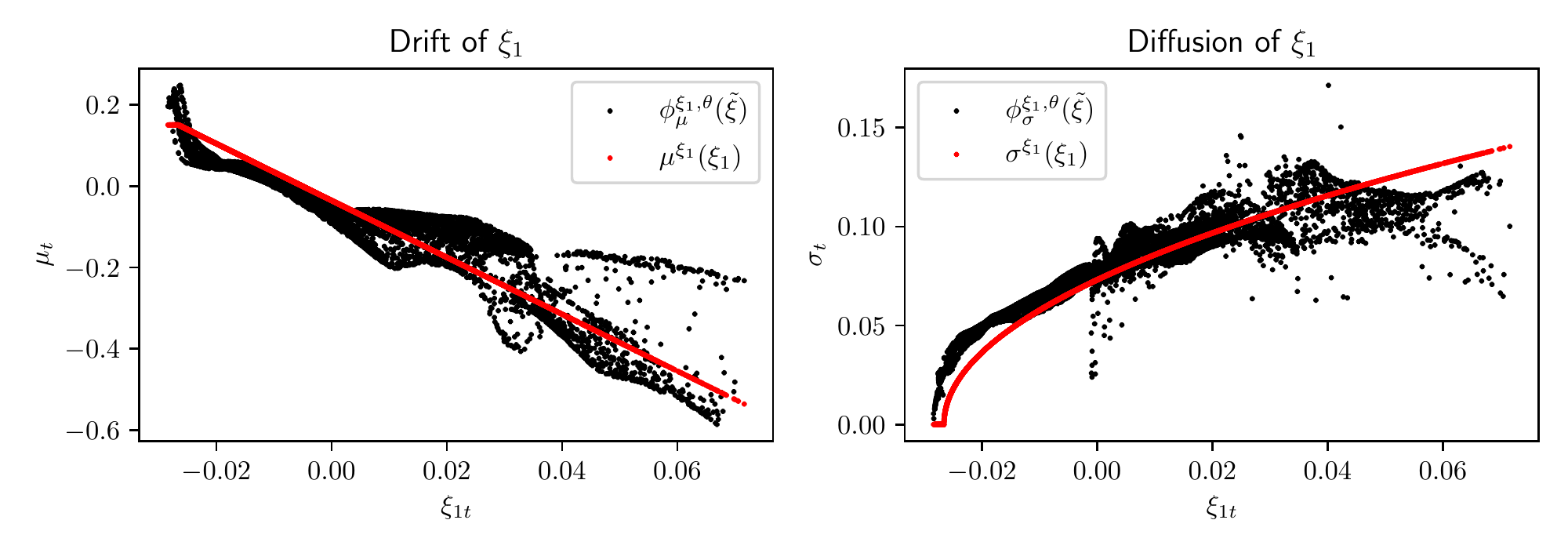}
    \caption{Comparison of the estimated drift and diffusion functions and the (approximated) ground-truth drift and diffusion functions for $\xi_1$.}
    \label{fig:test_model_xi1}
\end{figure}

\paragraph*{Market price of risk}

As discussed in Section \ref{sec:absence_dynamic_arbitrage}, the no-arbitrage HJM-type drift restriction implies an over-determined linear system \eqref{eq:mpor_linear_system}. Given the already estimated stock volatility $\gamma_t$ and the approximated partial derivatives of $\tilde{c}$, we compute $z_t$ and find an approximate solution to \eqref{eq:mpor_linear_system} using the ordinary least squares method,
\begin{equation}
    \hat{\psi}_t = \left( \sigma_t^\top \mathbf{G} \mathbf{G}^\top \sigma_t \right)^{-1} \sigma_t^\top \mathbf{G} \left( \mathbf{G}^\top \mu_t - z_t \right).
\label{eq:mpr_calculation}
\end{equation}
The violation to the drift restriction can be measured by
\begin{equation}
    \chi_t = \left\|  \mathbf{G}^\top \mu_t - z_t - \mathbf{G}^\top \sigma_t \hat{\psi}_t \right\|_2.
    \label{eq:dynamic_arbitrage_term}
\end{equation}

We estimate the market price of risk $\psi_t$ from the learnt drift and diffusion coefficients using (\ref{eq:mpr_calculation}) for all in-sample data points. In Figure \ref{fig:mpr}, we show the heatmaps of the violation of the HJM-type drift restriction $\chi$, as defined in (\ref{eq:dynamic_arbitrage_term}), and the size of the market price of risk $\|\hat{\psi}\|_2$. The calibrated market price of risk $\hat\psi$ initially appears large, given we have simulated under a risk-neutral measure; however, the use of a factor representation, forward translation of the leverage function, interpolation of prices, and the selection of a training path where $S$ does not leave the region of the leverage function (as discussed in Section \ref{sec:inputsimulation}) may lead to a non-vanishing market price of risk in the underlying ground truth model. We also note, comparing with Figures \ref{fig:factors}, \ref{fig:simulation_trajectory} and \ref{fig:simulation_dist}, that the most common values of $\xi_1$ are slightly negative, coinciding with the region where the market price of risk is low.

\begin{figure}[!ht]
    \centering
    \includegraphics[scale=.66]{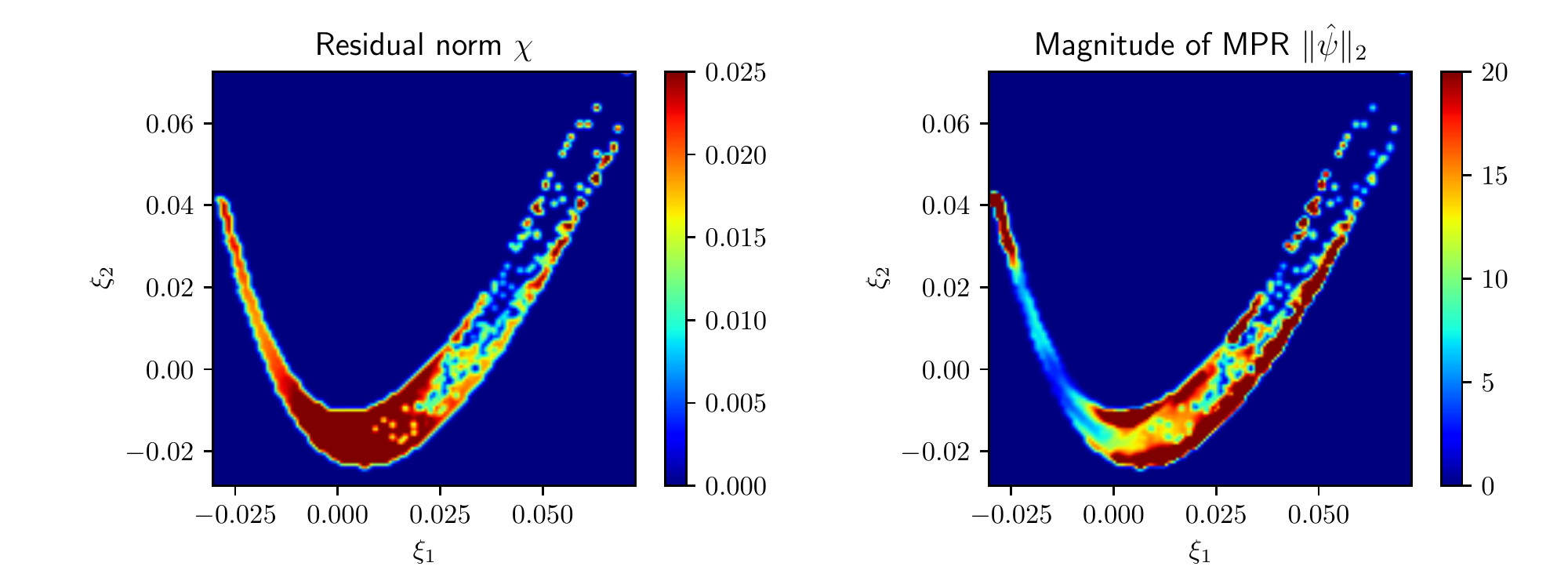}
    \caption{Statistics of the estimated market price of risk.}
    \label{fig:mpr}
\end{figure}

\subsection{Out-of-sample simulation test}
\label{sec:outsample_test}

Many potential applications of our models involve the simulation of $S$ and $\xi$, which can then be used for pricing or risk management. We take the trained model and simulate sample paths using a tamed\footnote{We include the taming method simply to ensure stability of simulations if our neural networks were to produce unusually large values for drifts or volatilities. In our examples, the tamed Euler scheme has very similar performance to the classical Euler--Maruyama scheme. } Euler scheme (see Hutzenthaler, Jentzen and Kloeden \cite{Hutzenthaler2012} or Szpruch and Zhang \cite{Szpruch2018}), which is given by
\begin{equation}
    \begin{aligned}
        S_{t + \Delta t} & = S_t + \frac{\mu^S (\tilde{\xi}_t)}{1 + |\mu^S(\tilde{\xi}_t)| \sqrt{\Delta t}} \Delta t + \frac{\sigma^S(\tilde{\xi}_t)}{1 + |\sigma^S(\tilde{\xi}_t) | \sqrt{\Delta t}} \left( W_{0, t+\Delta t} - W_{0,t} \right), \\
        \xi_{t + \Delta t} & = \xi_t + \frac{\mu (\tilde{\xi}_t)}{1 + |\mu(\tilde{\xi}_t)| \sqrt{\Delta t}} \Delta t + \frac{\sigma(\tilde{\xi}_t)}{1 + \|\sigma(\tilde{\xi}_t)\| \sqrt{\Delta t}} \left( W_{t+\Delta t} - W_t \right).
    \end{aligned}
\label{eq:tamed_euler}
\end{equation}
The values of $\mu^S$, $\sigma^S$, $\mu$ and $\sigma$ are approximated by the trained neural networks (along with the diffusion scaling $\mathcal{G}_\sigma$ and drift corrections $\mathcal{G}_\mu$).

To demonstrate the learnt model's ability to simulate time series data that are alike the input data, we show an independent sample path of $10000$ steps for $\tilde{\xi} = (S, \xi)$ in Figure \ref{fig:simulation_trajectory}, right.\footnote{There is no pathwise similarity here, as our learned factor model is based on a larger number of Brownian motions than the ground truth model.}
In the scatter plot on the left of Figure \ref{fig:simulation_trajectory}, we see that the dependence structure between $\xi_1$ and $\xi_2$ is well captured. In addition, the simulated factors remain within the no-arbitrage region, due to the hard constraints imposed on the drift and diffusion functions.

As expected, when the factor process is close to any no-arbitrage boundary, its diffusion component that is normal to the boundary will tend to vanish, while its drift will point inwards the no-arbitrage region. To illustrate this, we sample a few simulated factor data, and visualize their drift and diffusion coefficients in Figure \ref{fig:simulation_drift}.

\begin{figure}[!ht]
    \centering
    \includegraphics[scale=.66]{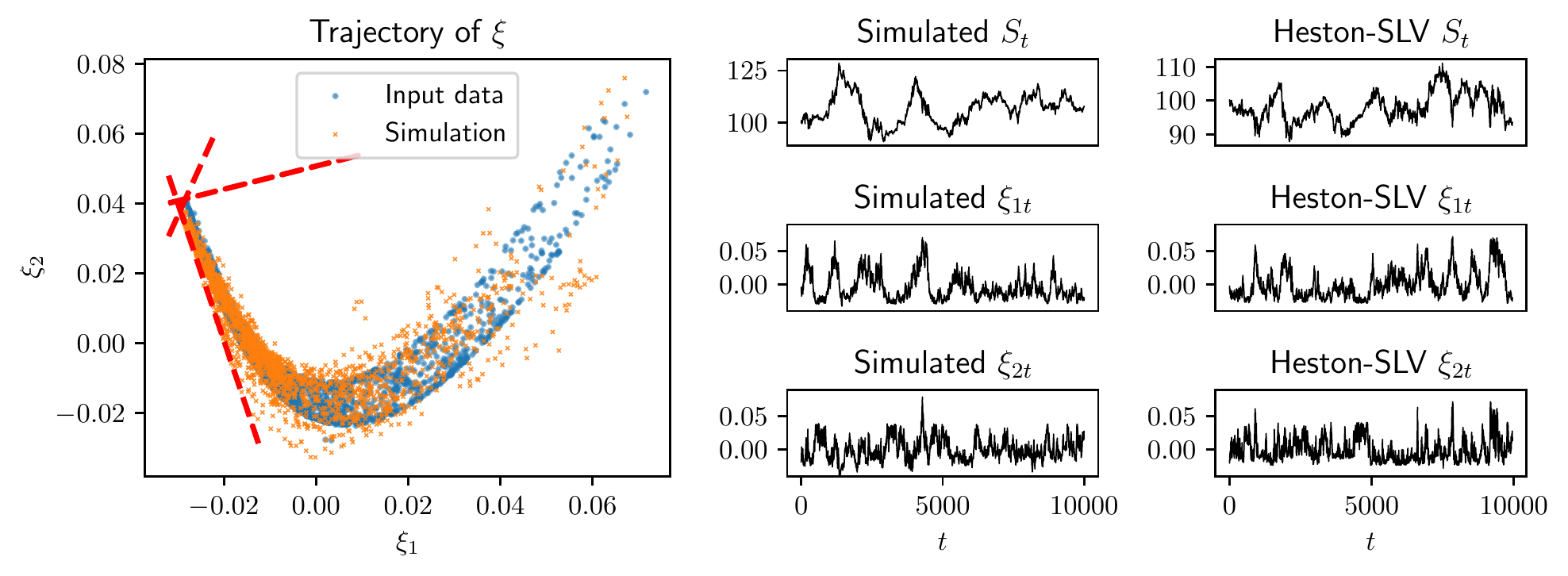}
    \caption{Simulation of $S$ and $\xi$ from the learnt neural network model, compared with the Heston-SLV model generated data.}
    \label{fig:simulation_trajectory}
\end{figure}

\begin{figure}[!ht]
    \centering
    \includegraphics[scale=.66]{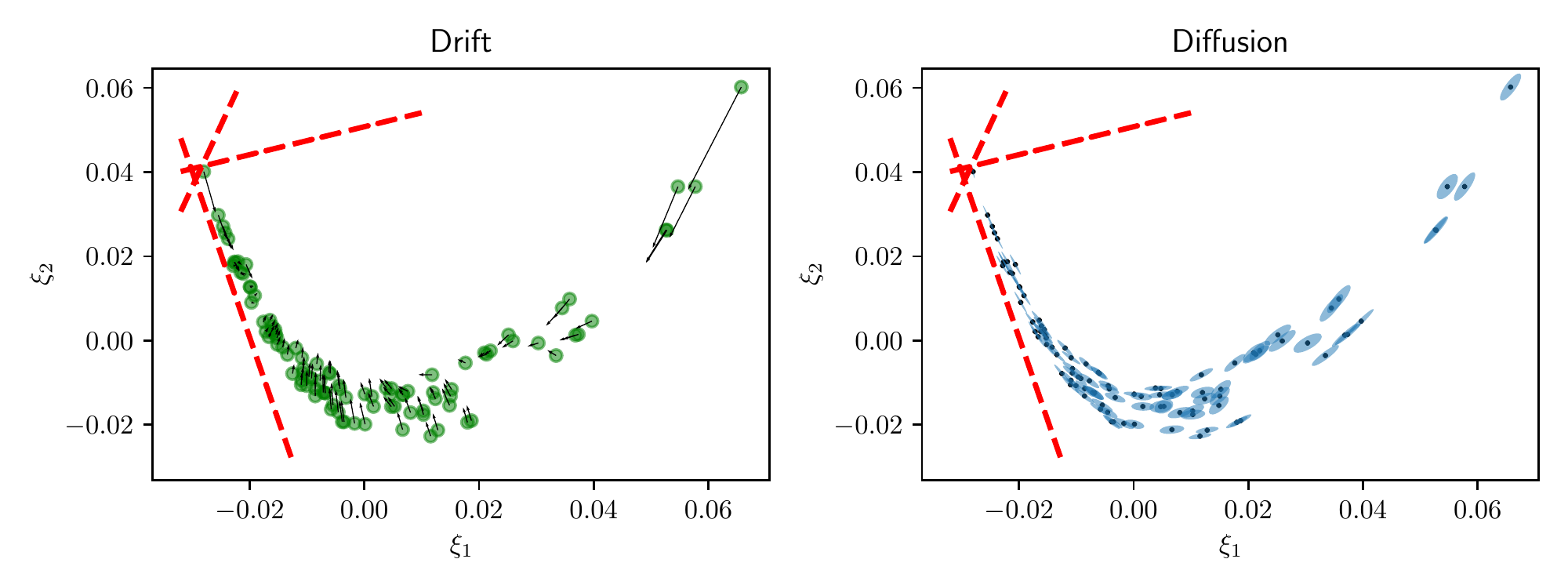}
    \caption{Drift vectors (arrows on the left plot) and diffusion matrices (ellipses representing the principal components of the diffusion on the right plot) for some randomly selected factor data points.}
    \label{fig:simulation_drift}
\end{figure}

In Figure \ref{fig:simulation_dist}, we compare the empirical distributions of the simulated log-return of $S$, $\xi_1$ and $\xi_2$ with those of the input data. We see that the learnt model is capable of generating realistic long time series data that is similar to the input data. More simulation results for implied volatilities can be found in Appendix \ref{sec:iv_simulation}, including the simulated paths of implied volatilities for a variety of option specifications.

\begin{figure}[!ht]
    \centering
    \includegraphics[scale=.66]{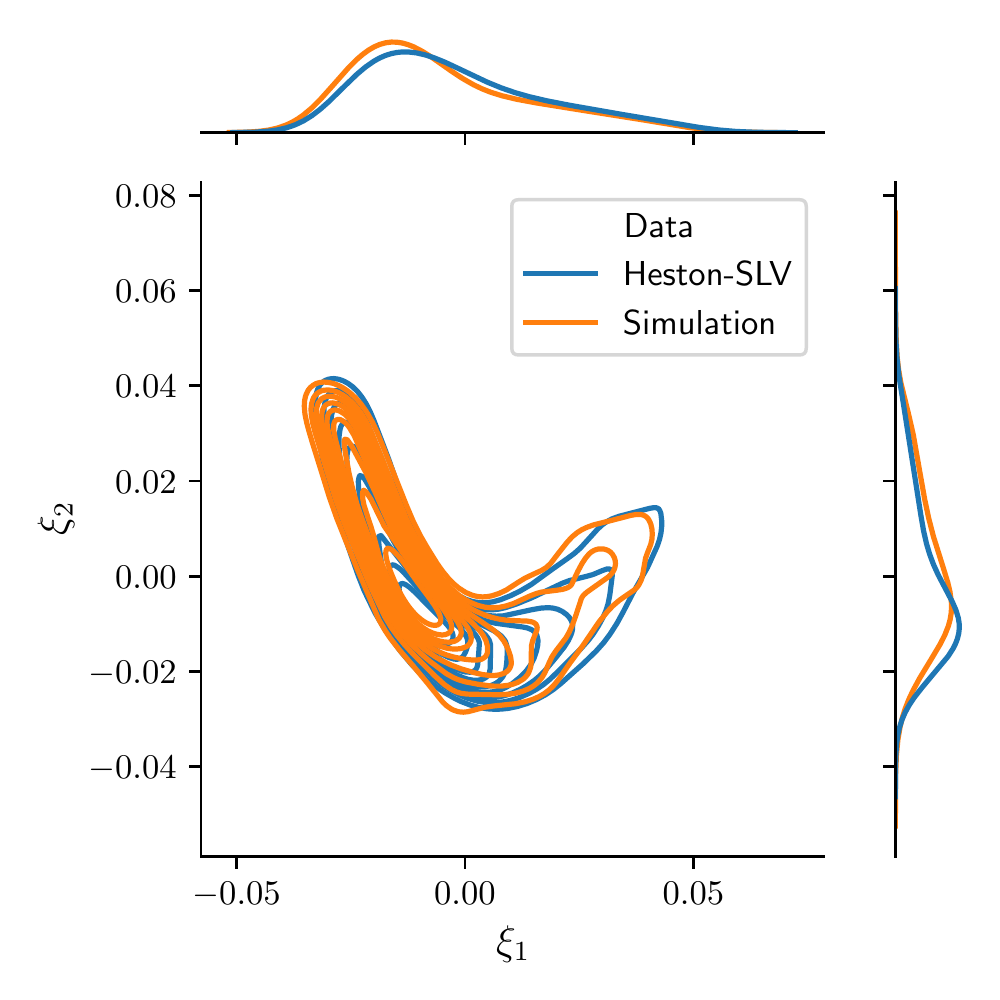}
    \caption{Comparison of the marginal and joint distributions of the simulated $\xi_1$ and $\xi_2$ with the real distributions (generated from the Heston-SLV model).}
    \label{fig:simulation_dist}
\end{figure}

\paragraph*{VIX simulation}

In addition, we follow the CBOE VIX calculation methodology \cite{vix} to compute a \textit{volatility index} from the simulated option prices. Specifically, VIX is a linear combination of OTM call and put option prices, which can be further written as a linear combination of call prices only, provided that put-call parity holds under no-arbitrage. 

Suppose the time $t$ VIX index can be written as $\textrm{VIX}_t := \mathbf{h}^\top \mathbf{c}_t$ for some constant vector $\mathbf{h} \in \mathbb{R}^N$. Given simulated factors $\xi_t^s$, we first reconstruct prices $\mathbf{c}_t^s$ and then compute VIX as
\begin{equation}
	\textrm{VIX}_t^s = \mathbf{h}^\top \mathbf{c}_t^s = \mathbf{h}^\top  \mathbf{G}_0 +  \mathbf{h}^\top \mathbf{G}^\top \xi_t^s.
	\label{eq:vix_simulation}
\end{equation}
In Figure \ref{fig:s_vix1}, we plot the marginal and joint distributions of the log-return of $S$ and the VIX index, for both the input data and the simulation data. While the distributions of the log-return of $S$ are very close, our simulation gives a VIX of much lower kurtosis. This error is principally due to our factor representation, which does not align with the weights involved in the VIX calculation. 

To improve the  performance of our existing model for  VIX simulation, we seek to minimise the impact of the factor reconstruction error. Consider fitting the linear regression model
\begin{equation}
	\textrm{VIX}_t = \beta_0 + \bm\beta^\top\xi_t + \varepsilon_t,
	\label{eq:vix_lr}
\end{equation}
Using this relationship, we can compute VIX directly from the simulated factors $\xi_t^s$ using
\begin{equation}
	\textrm{VIX}_t^r = \hat{\beta}_0 + \hat{\bm\beta} \xi_t^s.
	\label{eq:vix_simulation_regress}
\end{equation}
We plot the distributions for this factor-regression based VIX in Figure \ref{fig:s_vix2}. Both its marginal distribution and the joint distribution with the log-return of $S$ look reasonably similar to those of the Heston-SLV input data. This demonstrates that our model is capturing the dependence structure between the volatility index and the underlying $S$. In addition, the simulated time series of VIX and log-return of $S$ are plotted in Figure \ref{fig:simulation_vix}. We see several occurrences of volatility clustering in the return series, which always coincide with high VIX values.

\begin{figure}[!ht]
    \centering
    \begin{subfigure}[b]{.49\textwidth}
    \centering
        \includegraphics[scale=.65]{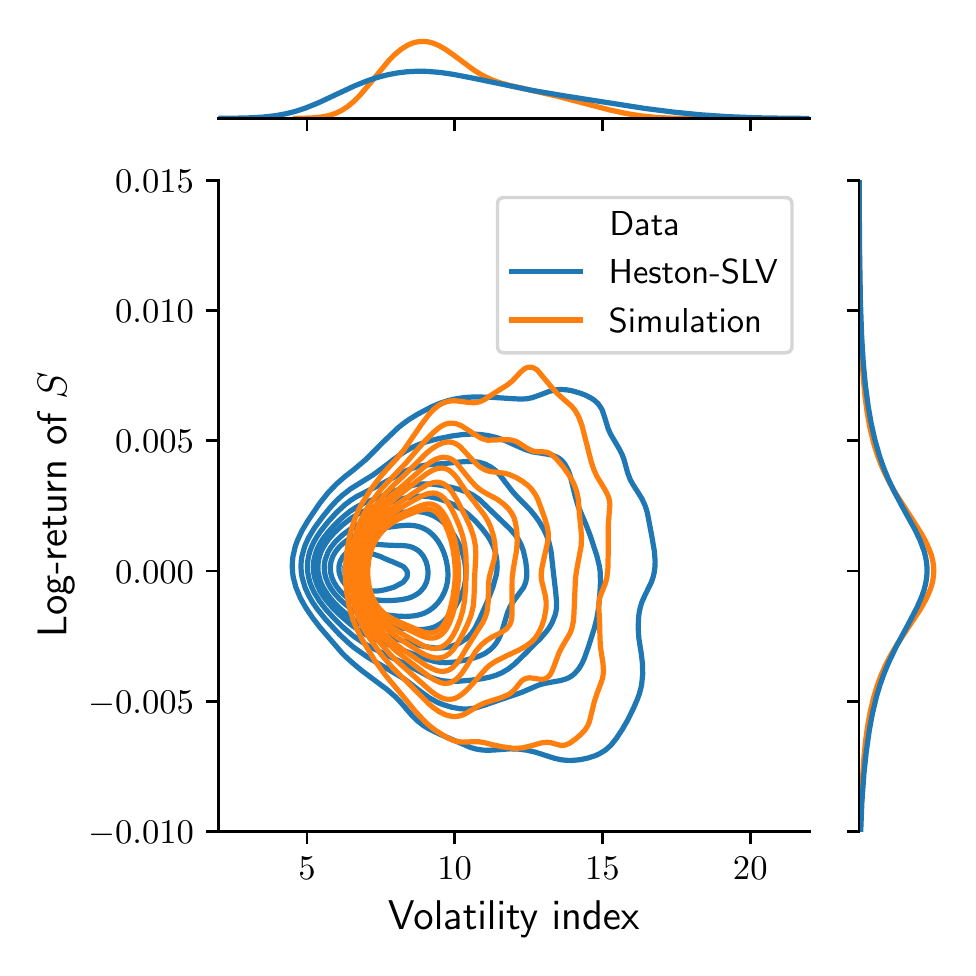}
        \caption{Heston-SLV VIX v.s. VIX calculated from the simulated option prices.}
        \label{fig:s_vix1}
    \end{subfigure}
    \hfill
    \begin{subfigure}[b]{.49\textwidth}
    \centering
        \includegraphics[scale=.65]{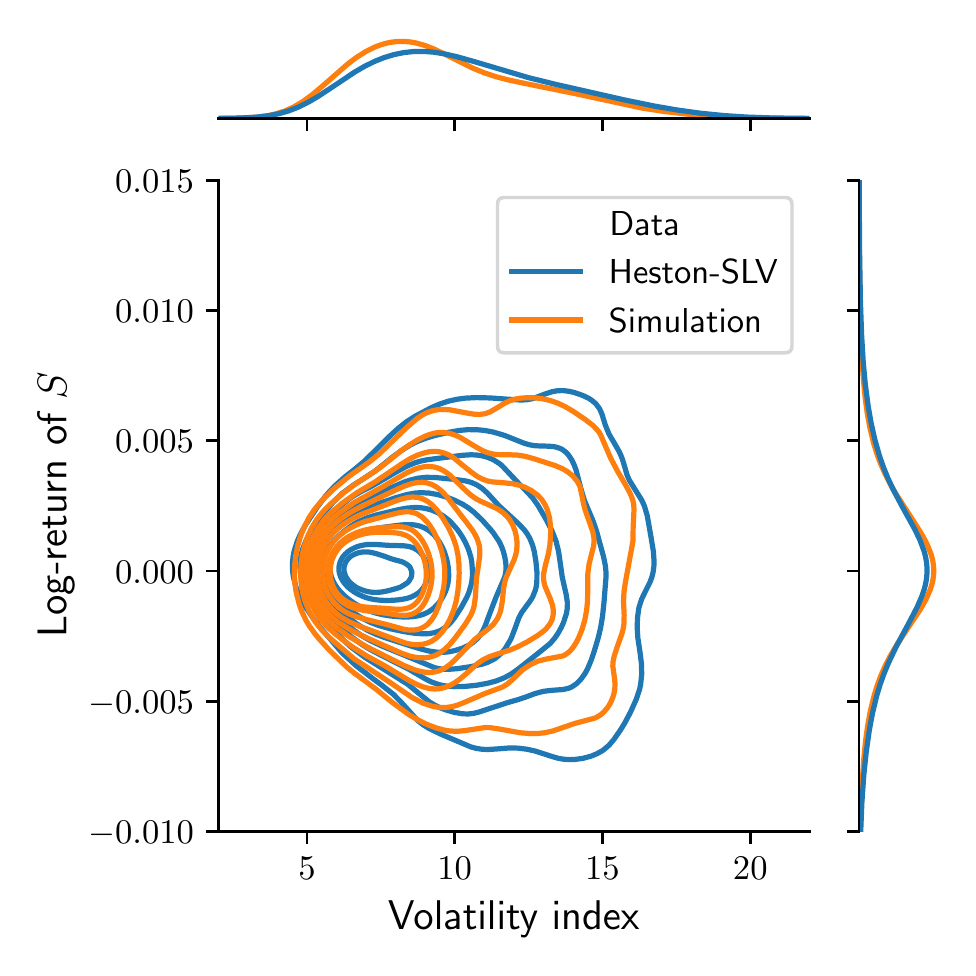}
        \caption{Heston-SLV VIX v.s. $\textrm{VIX}^r$ calculated from the regression model on the simulated factors.}
        \label{fig:s_vix2}
    \end{subfigure}
    \caption{Joint distribution of the log-return of $S$ and the VIX-like volatility index.}
    \label{fig:s_vix}
\end{figure}

\begin{figure}[!ht]
    \centering
    \includegraphics[scale=.66]{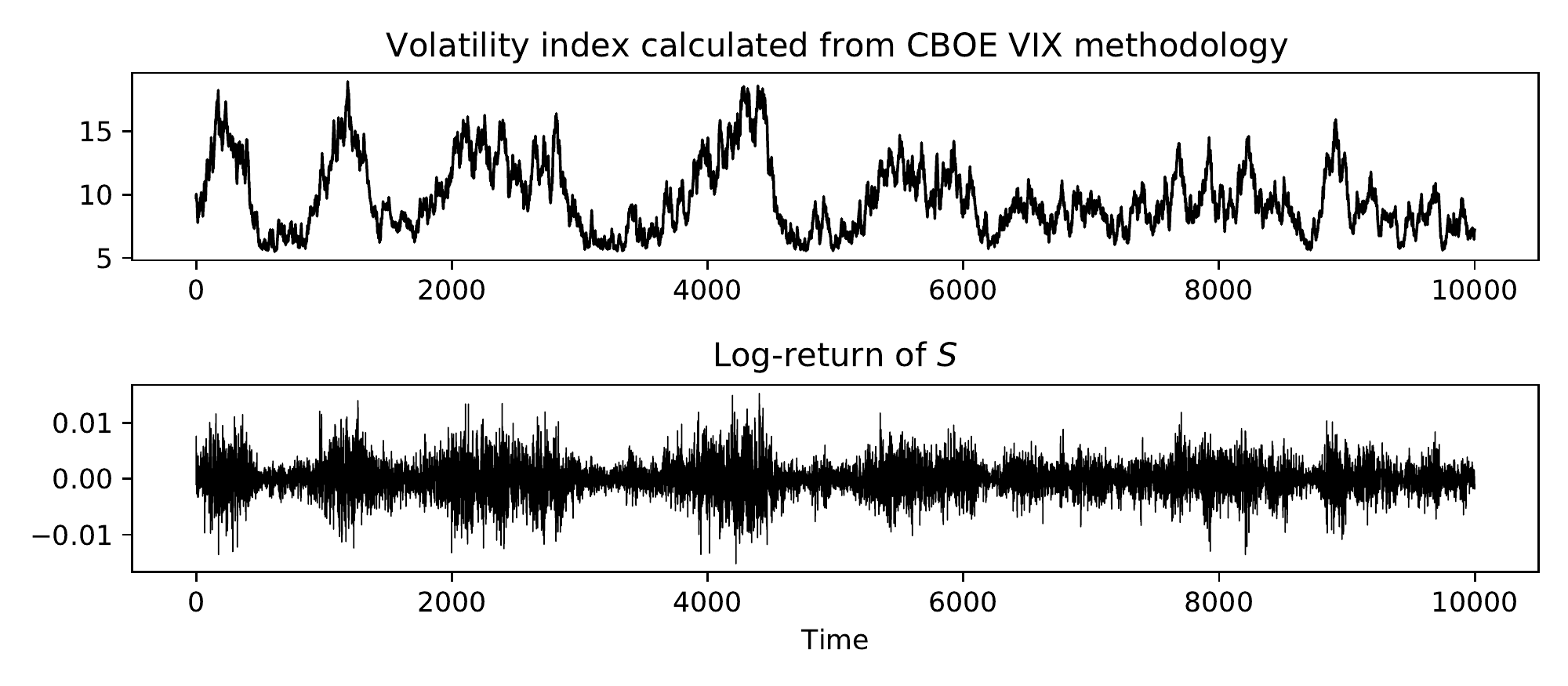}
    \caption{Simulated time series of $\textrm{VIX}^r$, together with log-returns of the underlying stock.}
    \label{fig:simulation_vix}
\end{figure}

\begin{remark}
The nontrivial improvement in VIX simulation from (\ref{eq:vix_simulation}) to (\ref{eq:vix_simulation_regress}) indicates that the decoded factors and their corresponding price basis do not optimise the representation of VIX. If one values a more accurate replication of VIX, an additional ``VIX factor'' could be included as a price basis vector in $\mathbf{G}$, with the objective of minimising the VIX reconstruction error and included in the model. (As VIX is here given by a linear combination of option prices, it can be directly included as a basis vector.)
\end{remark}

\subsection{Concluding remarks on the numerical example}

We have chosen to use three factors $(S, \xi_1, \xi_2)$ to represent a whole collection of call prices, while the ground-truth Heston-SLV model is a two-factor model with states $(S, \nu)$. Our model adopts an extra factor to compensate for the restrictive linear representation (\ref{eq:factr_rep_linear}) of normalised call prices in terms of latent factors. In fact, we see that $\xi_1$ has a strong linear relationship with $\nu$ (Figure \ref{fig:xi1nu}), where most of the residuals from this linearity has been captured by $\xi_2$. In addition, as seen in Figure \ref{fig:simulation_trajectory}, there seems to be a dominating polynomial relationship between $\xi_2$ and $\xi_1$. For higher dimensional models, this polynomial relationship (which is forced by the use of linear factors) seems likely to cause the main difficulty in calibration, and appropriate methods to address this may be valuable. We develop an $\mathbb{R}^{13}$ model and assess its simulation performance in Appendix \ref{sec:highdim_mdl}.

While our focus here is on understanding the presence of true arbitrage in our model, in practice it may also be of concern if the model exhibits statistical arbitrage, that is, where a combination of options moves with very high probability in a given direction. This would correspond to the estimated market price of risk process $\hat\psi$  providing a perfect solution (so \eqref{eq:dynamic_arbitrage_term} is zero) but taking very large values. In practice, therefore, an additional regularization term to discourage large  $\hat\psi$ may also be useful (for example the addition of $\|\hat\psi_t\|_2$ to \eqref{eq:dynamic_arbitrage_term}).

\appendix
\small
\section{HJM drift restriction for dynamic arbitrage}

\label{sec:dynamic_arbitrage}

In order to understand and avoid dynamic arbitrage, we will find a condition on the drifts $\alpha$, $\mu$ in (\ref{eq:market_model}) such that the resulting discounted prices $\bar c_t(T,K) := C_t(T,K)/D_t(T) = D_0(t) C_t(T,K)/D_0(T)$ are martingales under some equivalent measure.

From the price transformation (\ref{eq:call_price_transformation}) and the latent factor representation (\ref{eq:factr_rep_linear}), we can write $\bar{c}_t(T,K) = F_t(T) g (t, T-t, M(K;F_t(T)), \xi_t)$. For some fixed $\tau$ and $m$, we write $g(t,y)=g(t,\tau,m,y)$ for notational simplicity. By It\^{o}'s lemma, we have
\begin{equation}
\begin{aligned}
    \diff \tilde{c}_t = \diff g(t, \xi_t) & = \left[ \pderiv{g}{t}(t,\xi_t) + \sum_{i=1}^d \mu_{i,t} \pderiv{g}{y_i}(t,\xi_t) + \frac{1}{2} \sum_{i=1}^d \sum_{k=1}^d \sum_{j=1}^d \sigma_{ij,t} \sigma_{kj,t} \pderiv{g}{y_i, y_k} (t, \xi_t) \right] \diff t \\
     & ~~~~ + \sum_{i=1}^d \sum_{j=1}^d \sigma_{ij,t} \pderiv{g}{y_i}(t,\xi_t) \diff W_{j,t} \\
     & =: \tilde{\eta}_t \diff t + \tilde{\nu}_t \diff W_t.
\end{aligned}
\end{equation}

Let $F_t(T) = F(t, S_t; r_t, q_t)$. The dynamics of the $T$-maturity forward price are then
\begin{equation*}
    \diff F_t = \diff F(t,S_t) = (\alpha_t-r_t) F_t \diff t + \gamma_t F_t d W_{0,t}.
\end{equation*}
Since moneyness is a function of forward price, we use It\^{o}'s lemma again to derive the dynamics of moneyness. With $M(f) = M(K;f)$ for some fixed strike $K$, we have
\begin{equation}
    \diff m_t = \diff M(F_t) = \left[(\alpha_t - r_t)F_t\pderiv{M}{f}(F_t) + \frac{1}{2} \gamma_t^2 F_t^2 \pderiv[2]{M}{f} (F_t) \right] \diff t + \gamma_t F_t\pderiv{M}{f}(F_t) \diff W_{0,t}.
\label{eq:dynamics_moneyness}
\end{equation}

Now, for some fixed expiry $T$ and strike $K$, $\hat{c}_t(T,K) = \tilde{c}_t(\tau_t, m_t)$ where $m_t$ has the dynamics given by (\ref{eq:dynamics_moneyness}) and $\tau_t = T-t$. Using a generalised It\^{o}--Wentzell formula \cite{Venttsel1965},
\begin{equation}
\begin{aligned}
    \diff \hat{c}_t(T,K) & = \diff \tilde{c}_t(T-t, m_t) \\
    & = \Bigg\{ \tilde{\eta}_t(T-t, m_t) - \pderiv{g}{\tau} (t, T-t, m_t) + \frac{1}{2} \left[\gamma_t F_t \pderiv{M}{f} (F_t)\right]^2 \pderiv[2]{g}{m}(t, T-t, m_t) \\
    & ~~~~~~ + \left[ (\alpha_t - r_t)F_t\pderiv{M}{f}(F_t) + \frac{1}{2} \gamma_t^2 F_t^2 \pderiv[2]{M}{f} (F_t) \right] \pderiv{g}{m} (t,T-t, m_t)  \Bigg\} \diff t \\
    & ~~~~ + \gamma_t F_t\pderiv{M}{f}(F_t) \pderiv{g}{m}(t,T-t, m_t) \diff W_{0,t} + \tilde{\nu}_t(T-t, m_t) \diff W_t  \\
    & =: \eta_t(T,K) \diff t + \nu_t(T,K) \diff \overline{W}_t.
\end{aligned}
\end{equation}
Finally, let $\bar{c}_t = h(F_t, \hat{c}_t) = F_t \hat{c}_t$ and applying It\^{o}'s lemma to $h$, we have
\begin{equation}
    \diff \bar{c}_t = \diff h(F_t, \hat{c}_t) = \left[ \eta_t F_t + (\alpha_t - r_t) \bar{c}_t + \nu_{0,t} \gamma_t F_t \right] \diff t + ( \nu_t F_t + \gamma_t \bar{c}_t \mathbf{e}_1) \diff \overline{W}_t,
\label{eq:dynamics_cbar}
\end{equation}
where $\mathbf{e}_{i} \in \mathbb{R}^{(d+1)\times 1}$ is the unit vector with $i$-th entry one and other entries zero.

\subsubsection*{Drift restriction}

We wish to connect the dynamics of $\bar c$ with a no-arbitrage condition. Let $\overline{W} = [W_{0} ~\cdots ~W_{d}]^\top$ and $\widetilde{W}_t = [\widetilde{W}_{0,t} ~\cdots ~\widetilde{W}_{d,t}]^\top$ have dynamics given by
\begin{equation*}
    \widetilde{W}_t = \overline{W}_t + \int_0^t \varphi_t \diff t,
\end{equation*}
where $\varphi_t = [\psi_{0,t} ~\cdots ~\psi_{p,t}]^\top \in L^2_\text{loc}(\mathbb{R}^{d+1})$ is a $(d+1)$-dimensional progressively measurable process satisfying Novikov's\footnote{If Novikov's condition is not satisfied, but the process $\varphi$ is locally square-integrable in time, then some no-arbitrage results are still possible; see Karatzas and Kardaras \cite{Karatzas2007}.} condition $\mathbb{E}^\mathbb{P} [\exp (\frac{1}{2} \int_0^{T^*} \varphi_t^\top \varphi_t \diff t)] < 0$. Then Girsanov's theorem enables us to construct a measure $\mathbb{Q} \sim \mathbb{P}$ with the Radon--Nikodym derivative
\begin{equation}
    \left. \deriv{\mathbb{Q}}{\mathbb{P}} \right|_{\mathscr{F}_t} = \exp \left( -\int_0^t \varphi_s^\top \diff \overline{W}_s - \frac{1}{2} \int_0^t \varphi_s^\top \varphi_s \diff s \right).
\label{eq:girsanov}
\end{equation}
Under $\mathbb{Q}$ the process $\widetilde{W}_t$ is a $(d+1)$-dimensional standard Brownian motion. The process $\psi_i$ is interpreted as the market price of risk associated with the source of randomness $W_i$ for each $i=0,\dots,d$.

To ensure no arbitrage, we require the existence of a measure $\mathbb{Q} \sim \mathbb{P}$ under which the discounted price processes of all tradable assets are martingales. Since $\mathbb{Q}$ can be found through (\ref{eq:girsanov}) from the market price of risk process $\varphi$, we see that no-arbitrage requires the existence of an appropriate $\varphi$. Letting
\begin{equation}
    \psi_{0,t} = \frac{\alpha_t - r_t}{\gamma_t},
\label{eq:change_measure}
\end{equation}
under the measure $\mathbb{Q}$, the discounted stock price $\widetilde{S}_t = S_t \exp (\int_0^t (q_s - r_s) \diff s)$ and the forward price are $\mathbb{Q}$-local martingales, i.e.\ $\diff \widetilde{S}_t = \gamma_t \widetilde{S}_t \diff \widetilde{W}_{0,t}$ and $\diff F_t = \gamma_t F_t \diff \widetilde{W}_{0,t}$. Since we have assumed that $\gamma S \in L^2_\text{loc}(\mathbb{R}^{d+1})$, it follows that $\widetilde{S}$ and $F$ are also $\mathbb{Q}$-martingales.

We choose the rest of the $d$ components of $\varphi$ such that $\bar{c}$ is also a $\mathbb{Q}$-local martingale, i.e.\ it is driftless under $\mathbb{Q}$. Using (\ref{eq:dynamics_cbar}) and (\ref{eq:change_measure}), we have, under $\mathbb{Q}$,
\begin{equation}
    \diff \bar{c}_t = (\eta_t + \nu_{0,t} \gamma_t - \nu_t \cdot \varphi_t ) F_t \diff t + ( \nu_t F_t + \gamma_t \bar{c}_t \mathbf{e}_1) \diff \widetilde{W}_t.
\end{equation}
Hence, the no-arbitrage condition, also known as the \textit{drift restriction}, is that there must exist a process $\varphi$ satisfying
\begin{equation}
    \nu_t (T,K) \cdot \varphi_t (T,K) = \eta_t(T,K) + \nu_{0,t}(T,K) \gamma_t.
\label{eq:drift_restriction}
\end{equation}

Given the parametrisation in terms of moneyness $m_t$ and the market factor representation (\ref{eq:factr_rep_linear}), we can simplify the drift restriction (\ref{eq:drift_restriction}). For some fixed $T$ and $K$, we let $m_t=M(K;F_t(T))$ and $\tau_t=T-t$. Then $\hat{c}_t(T,K)$ has  drift and diffusion
\begin{align}
    \eta_t(T,K)  &=  \sum_{i=1}^d \left\{ \left[ \mu_{i,t} + \xi_{i,t} \left( -\pderiv{}{\tau} +  \left(r_t-\alpha_t+\frac{1}{2} \gamma_t^2 \right) \pderiv{}{m} + \frac{1}{2} \gamma_t^2 \pderiv[2]{}{m} \right) \right] G_i(\tau_t, m_t) \right\}\nonumber \\
     & \qquad + \left[ -\pderiv{}{\tau} +  \left(r_t-\alpha_t+\frac{1}{2} \gamma_t^2 \right) \pderiv{}{m} + \frac{1}{2} \gamma_t^2 \pderiv[2]{}{m} \right] G_0(\tau_t, m_t),\\
\nu_{j,t} (T,K) &=
\begin{dcases}
- \sum_{i=1}^d  \xi_{i,t} \gamma_t \pderiv{G_i(\tau_t, m_t)}{m} - \gamma_t \pderiv{G_0(\tau_t, m_t)}{m} , & j=0; \\
\sum_{i=1}^d  \sigma_{ij,t} G_i(\tau_t, m_t), & 1 \leq j \leq d.
\end{dcases}
\end{align}
Substituting $\eta_t$ and $\nu_t$ into the drift restriction (\ref{eq:drift_restriction}) yields
\begin{equation}
\begin{aligned}
     \sum_{j=1}^d \left( \sum_{i=1}^d \sigma_{ij,t} G_i(\tau_t, m_t) \right) \psi_{j,t}  
    &=  \sum_{i=1}^d \Bigg\{ \mu_{i,t} G_i(\tau_t, m_t) + \left[ -\pderiv{}{\tau} - \frac{1}{2} \gamma_t^2 \pderiv{}{m} + \frac{1}{2} \gamma_t^2 \pderiv[2]{}{m} \right] G_i(\tau_t, m_t) \xi_{i,t} \Bigg\} \\
    &=  \sum_{i=1}^d \mu_{i,t} G_i(\tau_t, m_t) + \left( -\pderiv{}{\tau} - \frac{1}{2} \gamma_t^2 \pderiv{}{m} + \frac{1}{2} \gamma_t^2 \pderiv[2]{}{m} \right) \tilde{c}_t (\tau_t, m_t).
\end{aligned}
\label{eq:drift_restriction_specific}
\end{equation}

It follows that a necessary and sufficient condition for our models to be dynamic arbitrage-free is that (\ref{eq:drift_restriction_specific}) admits solutions with sufficient integrability. 

\section{Adaptive diffusion shrinking scales}
\label{sec:adaptive_shrinking}

In Algorithm \ref{alg:diffusion_shrinking} , we would ideally like to shrink the diffusion component projected on $\mathbf{v}_{(k)}$ by $\sqrt{\varepsilon_{(k)}}$, for each $k$. However, since $\{ \mathbf{v}_{(k)} \}_k$ is, in general, not orthogonal, shrinking along one direction will also shrink along all other non-orthogonal directions.

To mitigate this compounded shrinking issue, a modification of the definition of $\mathbf{P}(y)$ (\ref{eq:construct_P}) in Algorithm \ref{alg:diffusion_shrinking} is as follows:

Similarly to \ref{alg:diffusion_shrinking}, we define
\begin{equation*}
    \mathbf{P}(y) = \text{diag} \left( \sqrt{\epsilon_{1}}, \cdots, \sqrt{\epsilon_{d}} \right) \times \mathbf{Q},
\end{equation*}
where $\{\sqrt{\epsilon_k}\}_k$ is a set of adaptive shrinking scales. The novelty here is that we compute these scales recursively, together with the Gram--Schmidt process in \eqref{eq:alg2GramSchmidt}. 

Let $\epsilon_1 = \varepsilon_{(1)}$. For each $k \geq 2$, suppose $\{\epsilon_j\}_{j\leq k-1}$ is known. By the Gram--Schimdt process, we have
\begin{equation*}
    \mathbf{v}_{(k)} = \sum_{j=1}^{k-1} \langle \mathbf{q}_j,  \mathbf{v}_{(k)} \rangle \mathbf{q}_j + \sqrt{ 1 - \sum_{j=1}^{k-1} \langle \mathbf{q}_j,  \mathbf{v}_{(k)} \rangle^2} \mathbf{q}_k.
\end{equation*}
Since we shrink along $\mathbf{q}_j$ by $\sqrt{\epsilon_j}$ for all $j$,  $\mathbf{v}_k$ will shrink to
\begin{equation*}
    \hat{\mathbf{v}}_{(k)} = \sum_{j=1}^{k-1} \sqrt{\epsilon_j} \langle \mathbf{q}_j,  \mathbf{v}_{(k)} \rangle \mathbf{q}_j + \sqrt{\epsilon_k} \sqrt{ 1 - \sum_{j=1}^{k-1} \langle \mathbf{q}_j,  \mathbf{v}_{(k)} \rangle^2} \mathbf{q}_k.
\end{equation*}
In addition, we want the rescaled vector to have norm $\| \hat{\mathbf{v}_{(k)}} \|^2 = \varepsilon_{(k)} $, which yields
\begin{equation*}
    \epsilon_k = \frac{1}{1 - \sum_{j=1}^{k-1} \langle \mathbf{q}_j,  \mathbf{v}_{(k)} \rangle^2} \left( \varepsilon_{(k)} - \sum_{j=1}^{k-1} \langle \mathbf{q}_j,  \mathbf{v}_{(k)} \rangle^2 \epsilon_j \right).
\end{equation*}
If using this method, it is then natural to re-sort the boundary directions, so that at each step we shrink the direction corresponding to the largest remaining value of $\epsilon_k$ (rather than the largest value of the original shrinking scale $\varepsilon_k$). The matrix $\mathbf{Q}$ is then built accordingly.

\section{The $\rho^*$-interior points for drift correction}
\label{sec:interior_point}

We apply Algorithm \ref{alg:interior_points} to define the interior points for drift correction. The correction is more efficient when the correction vectors formed by the interior points align with the inward orthogonal vector to the corresponding boundaries. Therefore, for each boundary, we first compute its midpoint (as a boundary of the polytope $\mathcal{P}$) and then, starting from the midpoint, we find the furthest possible $\rho^*$-interior point along the inward orthogonal direction. We can then pre-compute the set of drift correction vectors $\{\zeta_k - y_{t_i}\}_{k=1,\dots,R}$ for every observation $y_{t_i}$, where $i=0,\dots,L$.

\begin{algorithm}[H]
\footnotesize
\SetAlgoLined
\SetKwInOut{Input}{Input}\SetKwInOut{Output}{Output}
\Input{A non-empty $d$-polytope $\mathcal{P} = \{x \in \mathbb{R}^d : \mathbf{V} x \geq \mathbf{b}, \mathbf{V} \in \mathbb{R}^{R \times d}, \mathbf{b} \in \mathbb{R}^d \}$, and a small positive constant $\rho^*$.}
\Output{One $\rho^*$-interior point corresponding to each of the $R$ hyperplane boundaries of $\mathcal{P}$.}
\BlankLine

Apply the double description algorithm (Motzkin, Raiffa, Thompson and Thrall \cite{Motzkin1953}) to identify the set of vertices of $\mathcal{P}$, denoted by $\mathcal{V}$;

\ForEach{$k = 1,\dots, R$}{
	Find the set of vertices of $\mathcal{P}$ that are passed through by the $k$-th boundary, denoted by $\mathcal{V}^k = \mathcal{V} \cap \{x\in\mathbb{R}^d: \mathbf{v}_k^\top x = b_k \}$\;
	Compute the midpoint of the $k$-th boundary as
	\begin{equation*}
	    m_k = \frac{1}{|\mathcal{V}^k|} \sum_{v \in \mathcal{V}^k} v;
	\end{equation*}\\
	Solve $c^* = \arg\max_{c \in \mathbb{R}_{>0}} \{c:\mathbf{V}(c\mathbf{v}_k + m_k) \geq \mathbf{b} - \rho^*\mathbf{1}\}$. The $\rho^*$-interior point is given by $\zeta_k = c^*\mathbf{v}_k + m_k$;
	}
\caption{Computation of $\rho^*$-interior points}
\label{alg:interior_points}
\end{algorithm}

\section{Implied volatility simulation}
\label{sec:iv_simulation}

As a further description of the performance of our model, we consider the simulated time series of implied volatilities (IV) and IV smiles for multiple expiries and moneynesses. These are given in Figure \ref{fig:simulation_ivs_ts} and Figure \ref{fig:simulation_ivs_smile}.

We see that the long-maturity and short-maturity-ITM implied volatilities are similar, but the short-maturity-OTM volatility processes have noticeably different behavior. This is also reflected in the varied shapes of the volatility smiles at different points in time.

\begin{figure}[!h]
    \centering
    \includegraphics[scale=.66]{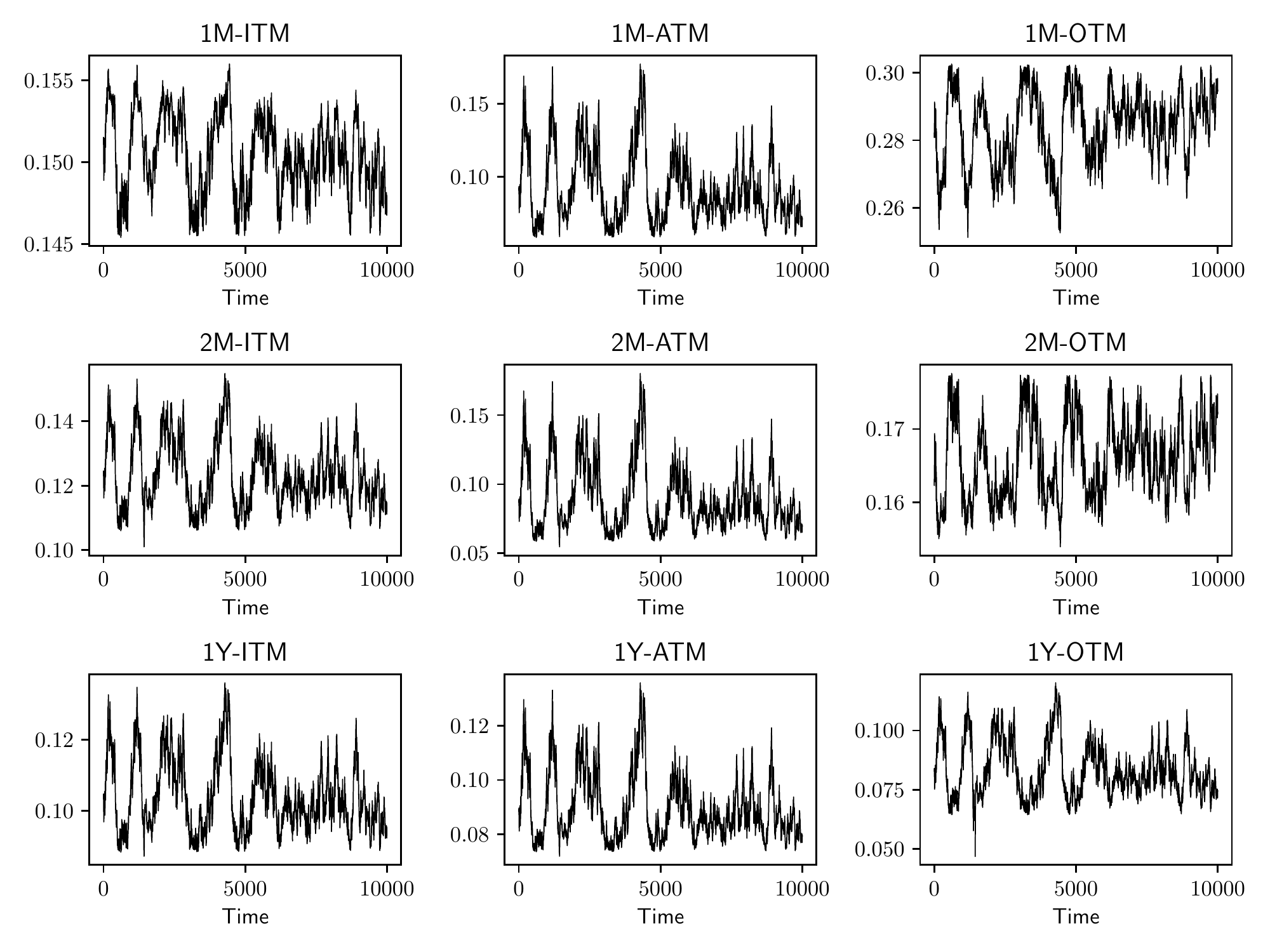}
    \caption{Simulated evolution of implied volatilities for multiple expiries and moneynesses. ATM: $m=0$; ITM: $m=-0.0943$; OTM: $m=0.1133$.}
    \label{fig:simulation_ivs_ts}
\end{figure}

\begin{figure}[!h]
    \centering
    \includegraphics[scale=.66]{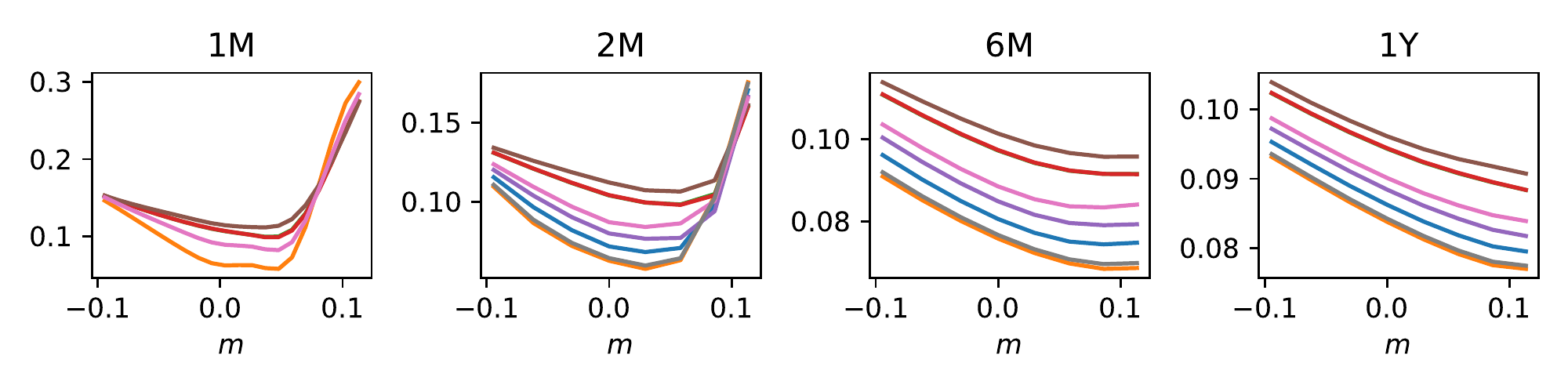}
    \caption{Simulated implied volatility smiles for 1M, 2M, 6M and 1Y options on multiple days.}
    \label{fig:simulation_ivs_smile}
\end{figure}

\section{Necessity of drift and diffusion operators}

Given the underlying drift and diffusion functions  are represented by neural networks, our approach is to shrink the diffusion and correct the drift near boundaries such the Friedman--Pinsky conditions \eqref{eq:friedman1973_1} are satisfied. These conditions are sufficient for constraining the factor process within the state space where there is no static arbitrage.

The question is whether it is necessary in practice to do the diffusion shrinking and drift correction. In particular, since the input data are already arbitrage-free, is the neural network capable of learning the correct boundary behaviours for the drift and diffusion functions without enforcing extra constraints?

To address this question, we estimate a model using exactly the same input data, neural network architecture and other technical details, except for neither shrinking the diffusion nor correcting the drift. The evolution of training and validation losses over epochs is shown in Figure \ref{fig:loss_history_uncon}. We find that, compared with Figure \ref{fig:loss_history},
\begin{enumerate}[leftmargin=*, label=(\arabic*)]
\setlength\itemsep{1pt}
    \item The losses converge faster. This is probably because that diffusion shrinking and drift correction constrain the gradients of the loss function when using backpropagation.
    \item The initial loss value is much larger. By shrinking diffusion and correcting drift at every epoch, we force the neural network training to start from and stay within the correct function subspace.
\end{enumerate}

\begin{figure}[!h]
    \centering
    \includegraphics[scale=.66]{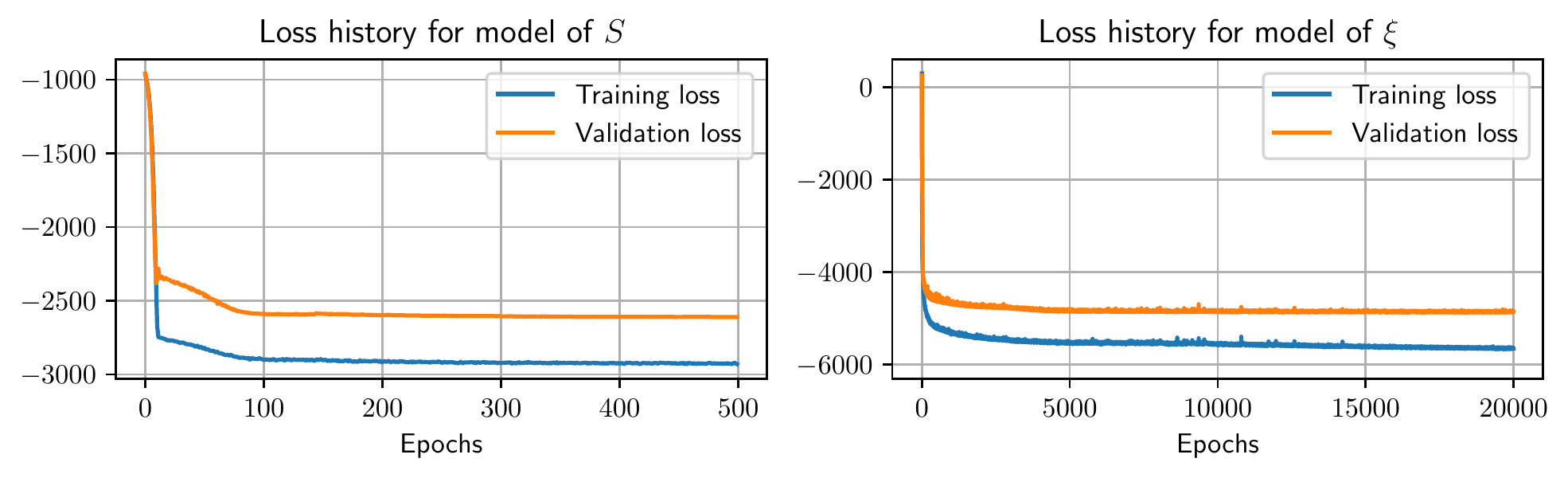}
    \caption{Evolution of training losses and validation losses (without imposing drift correction and diffusion shrinking).}
    \label{fig:loss_history_uncon}
\end{figure}
\begin{figure}[!h]
    \centering
    \includegraphics[scale=.66]{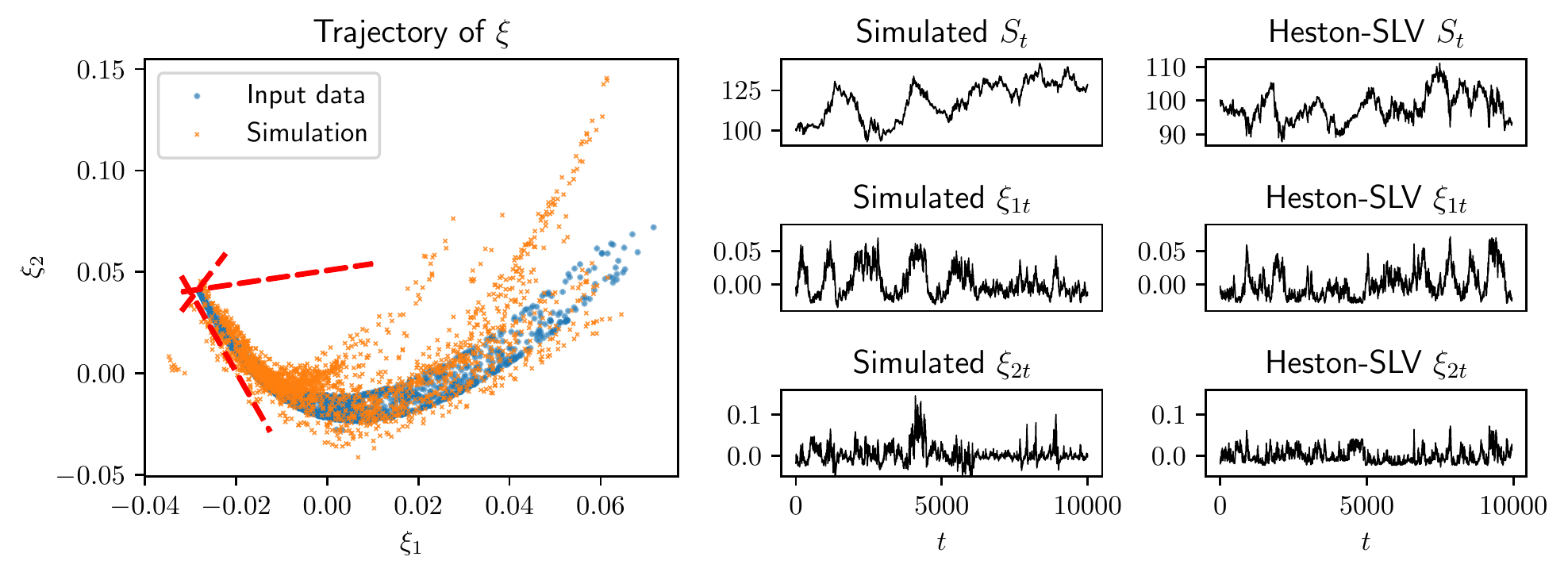}
    \caption{Simulation of $S$ and $\xi$ from the learnt neural network model (without imposing drift correction and diffusion shrinking).}
    \label{fig:simulation_trajectory_uncon}
\end{figure}

We investigate how well the learnt model can simulate out-of-sample data. As shown in Figure \ref{fig:simulation_trajectory_uncon}, the simulated factors can easily break the static arbitrage constraints (the dashed red lines), and generate very unrealistic values. Compared with the simulation results shown in Figure \ref{fig:simulation_trajectory} (where the same random seed is used), drift correction and diffusion shrinking are effective at producing the correct boundary behaviours, which is crucial for an arbitrage-free model.

\section{Neural network sensitivity analysis}
\label{sec:nn_sensitivity}

We demonstrate the sensitivity of the model estimation outcomes against a few neural network hyperparameters, including:
\begin{itemize}
\setlength\itemsep{1pt}
    \item Neural network depth: the number of hidden layers plus one (the output layer);
    \item Neural network width: the number of neurons per layer;
    \item Sparsity ratio: the fraction (between 0 and 1) of the network layers' weights that are pruned to zero;
    \item Activation function.
\end{itemize}

We call the network given in  Section \ref{sec:training_results} the \textit{benchmark} network. The values of the hyperparameters for the benchmark network for $\xi$ are listed in Table \ref{tab:params_values}. To carry out  sensitivity analysis, we will train a collection of modified networks by varying one hyperparameter from the benchmark network.

\begin{table}[!h]
    \centering
    \footnotesize
    \begin{tabular}{lccccc}
        \toprule
        \textbf{Hyperparameters} & NN depth &  NN width & Sparsity ratio & Activation function \\
        \cmidrule(lr){1-1} \cmidrule(lr){2-5}
        \textbf{Values} & 3 & 256 & 0.5 & ReLU \\
        \bottomrule
    \end{tabular}
    \caption{Hyperparameters used for the benchmark model.}
    \label{tab:hp_benchmark}
\end{table}

We compare the evolution of training losses and simulation performance between the benchmark network and the modified networks. The training losses are shown in Figure \ref{fig:sens_loss}. We have shown in Figure \ref{fig:simulation_dist} that the trained benchmark network simulates data with a similar distribution to the training data. To quantify this similarity between the distributions of simulated data and input data, we compute the Wasserstein distance of order one between the two distributions. We compare the computed distances for the benchmark network and the modified networks and list the results in Table \ref{tab:wasserstein}.

\begin{figure}[!h]
    \centering
    \includegraphics[scale=.66]{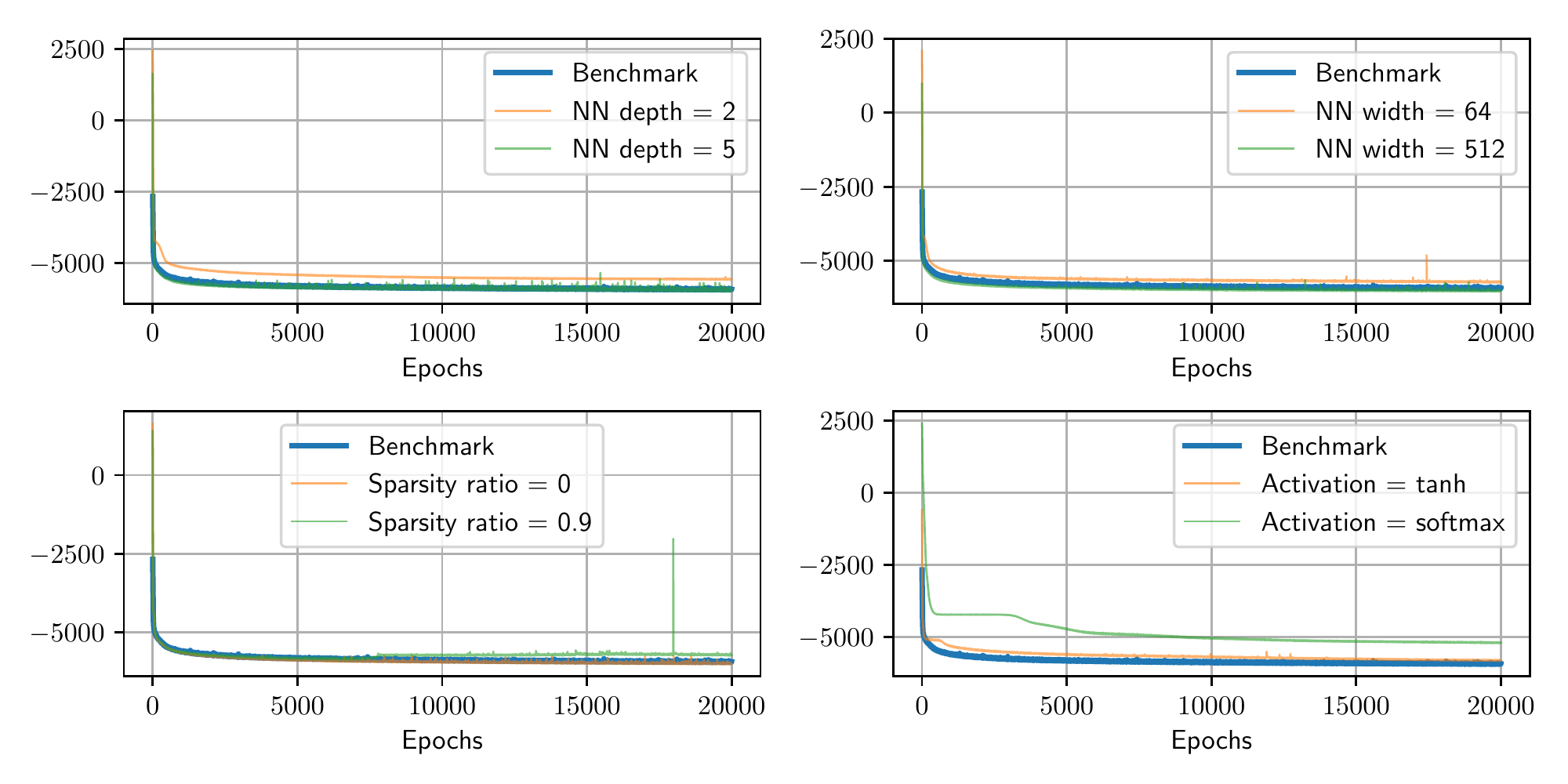}
    \caption{Evolution of training losses under different neural network architectures.}
    \label{fig:sens_loss}
\end{figure}

\begin{table}[!ht]
    \centering
    \footnotesize
    \begin{tabular}{lccc}
        \toprule
        \textbf{Neural network} & \textbf{Log-return of} $S$ & $\xi_1$ & $\xi_2$ \\
        \cmidrule(lr){1-1} \cmidrule(lr){2-4}
        Benchmark               & 1.71  & 40.44     & 20.30 \\
        NN depth = 2            & 3.79  & 176.89     & 71.99 \\
        NN depth = 5            & 2.83  & 55.07     & 52.14 \\
        NN width = 64           & 2.32  & 52.46     & 29.17 \\
        NN width = 512          & 1.82  & 58.90     & 24.97 \\
        Sparsity ratio = 0      & 1.51  & 53.92     & 57.47 \\
        Sparsity ratio = 0.9    & 3.68  & 69.69     & 105.61 \\
        Activation = tanh       & 2.17  & 46.57     & 35.26 \\
        Activation = softmax    & 2.97 & 148.58    & 462.74\\
    \bottomrule
    \end{tabular}
    \caption{Wasserstein distances ($\times 10^{-4}$) between the empirical distributions of simulated data and the input data.}
    \label{tab:wasserstein}
\end{table}

The training losses of all networks drop fairly quickly during the first 1000 epochs, and then gradually converge. Together with the simulation performance, we find that:
\begin{itemize}
\setlength\itemsep{1pt}
    \item The 2-layer shallow neural network gives higher convergent losses and much worse simulation results. Increasing the depth of the neural network, or varying its width, has a relatively small effect on performance, with mixed direction.
    \item The ReLU activation function generally produces smaller convergent losses and more similar simulations (i.e.\ smaller Wasserstein distances). The performance using a softmax activation function (as implemented by Tensorflow) was surprisingly poor.
    \item Pruning the network with moderate sparsity improves the extrapolation capability of the learnt neural network, thus resulting in better simulation results for $\xi$.
\end{itemize}

\section{Higher-dimensional models}
\label{sec:highdim_mdl}

In our numerical example, we observe that representing the call price data with more factors can further reduce reconstruction error, which improves the accuracy of simulations for option portfolios, such as VIX. However, a few challenges arise:
\begin{itemize}
\setlength\itemsep{1pt}
\item The inclusion of more factors leads to higher-dimensional models; and the number of unknown functions to estimate increases quadratically with the number of factors, increasing the chance of overfitting and making the estimates less robust.
\item There are dominating polynomial relationships between factors, specifically because the call price data are generated from a two-factor model.
\end{itemize}
These issues led to poor performance of a high dimensional version of our model, in particular, highly unrealistic simulations of prices.

To overcome these issues, we assume that some \textit{secondary} factors could be expressed as polynomials of some \textit{primary} factors plus noise. Specifically, suppose there are $d'$ primary factors and $d-d'$ secondary factors. For each secondary factor $\xi_i$, where $i=d'+1,\dots,d$, we have an independent Brownian motion $W_i$, and the model
\begin{equation}
	\xi_i = f_i(\xi_1, \dots, \xi_{d'}) + \varepsilon_i, \text{ and } \diff \varepsilon_{it} = \kappa_i (\theta_i - \varepsilon_{it}) \diff t + \varsigma_i \diff W_{it}. 
\end{equation}
In other words, noises are assumed to be mean-reverting Ornstein-Uhlenbeck (OU) processes with constant parameters $\kappa$, $\theta$ and $\varsigma$. Here $f$ is calibrated to the observed dominating polynomial relations between factors. Let $\xi^p$ be the collection of all primary factors, with drift $\mu^p: \mathbb{R}^{d'+1} \rightarrow \mathbb{R}^{d'}$ and diffusion $\sigma^p: \mathbb{R}^{d'+1} \rightarrow \mathbb{R}^{d' \times d'}$. Consequently, the dynamics for the secondary factor $\xi_i$ is
\begin{equation}
	\diff \xi_i = \mu_i (\xi^p, \xi_i ) \diff t + \sigma_i (\xi^p) \diff W^p_{t} + \varsigma_i \diff W_{it}, \text{ with } 
	\begin{cases}
	\mu_i = \kappa_i (\theta_i - \xi_i + f_i) + (\nabla f_i)^\top \mu^p +  \frac{1}{2} \sum_{j,k} \Gamma_{jk}, \\
	\sigma_i = (\nabla f_i)^\top \sigma^p,
	\end{cases}
	\label{eq:model_sec_factor}
\end{equation}
where $\Gamma = (\Gamma_{jk}) = [\sigma^p (\sigma^p)^\top] \odot \nabla^2 f_i$ is the element-wise product of the covariance matrix and the Hessian matrix of $f_i$, $W^p_t = [W_{1t} \cdots W_{d't}]^\top$ and $\nabla = [\partial/\partial \xi_1 \cdots \partial/\partial \xi_{d'}]^\top$ is the gradient operator.

\paragraph*{Factor decoding}

We keep the two factors in our numerical example as the primary ones, and append another 11 secondary factors, decoded consecutively with the objective of maximizing statistical accuracy using Algorithm \ref{alg:decode_factor}. We choose to append 11 secondary factors because this is the minimal number of factors which ensures no static arbitrage, as seen in Table \ref{tab:metrics_secondary_factor}. Each secondary factor is modelled as a cubic function of $(\xi_1, \xi_2)$ plus OU noises. In Figure \ref{fig:factor_polynomial}, we show the scattergrams for $\xi_7$ and the primary factors $(\xi_1, \xi_2)$, and compare with those for the polynomial implied values. The residual time series is also plotted.
\begin{table}[!ht]
\tiny
\centering
\begin{tabular}{lcccccccccccc}
\toprule
\multicolumn{1}{c}{\multirow{2}{*}{Metrics}} & \multicolumn{12}{c}{Number of secondary factors}  \\
\cmidrule(lr){2-13}
\multicolumn{1}{c}{} & 1 & 2 & 3 & 4 & 5 & 6 & 7 & 8 & 9 & 10 & 11 & 12 \\
\cmidrule(lr){1-1} \cmidrule(lr){2-13}
MAPE (\%) &  2.82 & 1.82  & 1.66  & 1.31  & 0.59  & 0.55  & 0.28  & 0.24 & 0.21  & 0.17  & 0.13  &  0.10  \\
PDA (\%) & 2.81  & 0.89  & 0.77  & 0.70  & 0.60  & 0.59  & 0.52  & 0.38  & 0.31  & 0.28   &  0.23  &  0.21  \\
PSAS (\%) & 27.51  & 20.88  & 18.25  & 15.33  & 12.35  & 10.92  & 7.65  & 7.45  & 2.09  & 2.01   &  0.00  & 0.00   \\
\bottomrule
\end{tabular}
\caption{MAPE, PDA and PSAS metrics when appending different number of secondary factors.}
\label{tab:metrics_secondary_factor}
\end{table}
\begin{figure}[!h]
\centering
\includegraphics[scale=0.66]{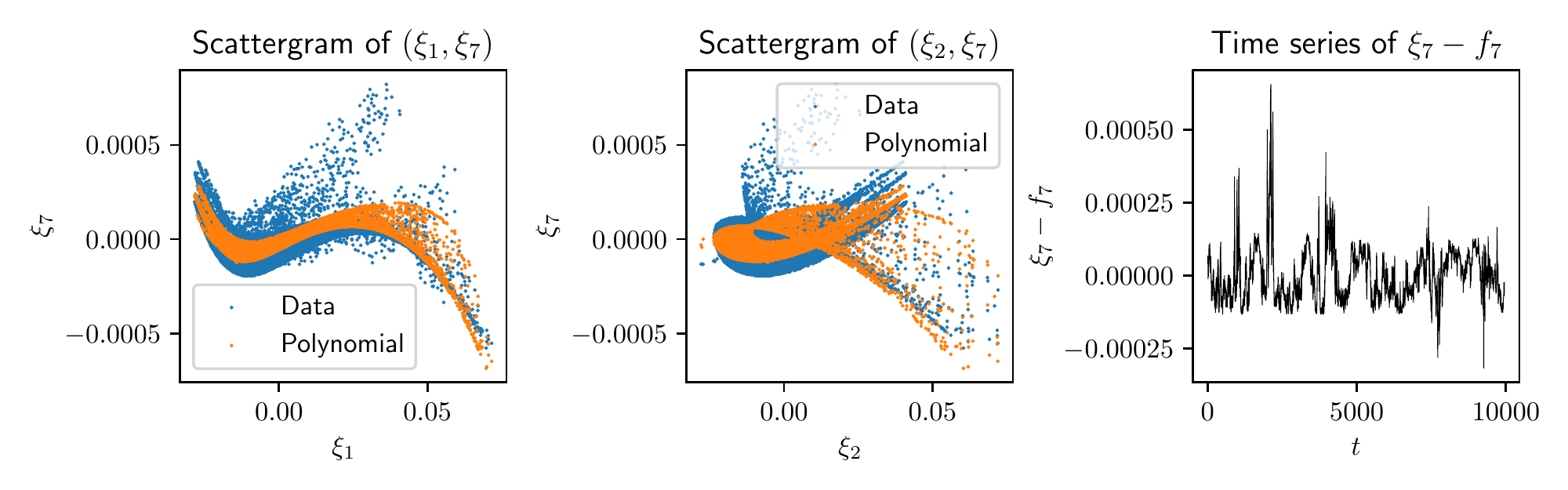}
\caption{Polynomial fitting results for $\xi_7$.}
\label{fig:factor_polynomial}
\end{figure}

\paragraph*{Model estimation strategy}

We estimate the neural-SDE model for primary factors and the OU noise processes for secondary factors separately rather than jointly. Since primary factors are much more representative of the input price data, it is crucial to learn their dynamics correctly, while joint estimation with secondary factors dilute the weight of primary factors in the objective likelihood function. Precisely, with the 2 primary factors and the 11 secondary factors, we
\begin{enumerate}[leftmargin=*, label=(\roman*)]
\item estimate a neural-SDE model for the 2 primary factors, as we did in the numerical example (Section \ref{sec:numerics});
\item calibrate a cubic function $f_i(\xi_1, \xi_2)$ for each secondary factor using the least-squares method, and compute the corresponding residual time series data $\{\varepsilon_{i t_l}\}_{l=1,\dots,L+1}$;
\item discretize the OU model for $\varepsilon_i$ using the Euler-Maruyama scheme, which yields an AR(1) model with normal noises, and fit with the residual time series data.
\end{enumerate}

\paragraph*{Arbitrage-free simulation}

Unlike the neural-SDE model for primary factors, the model for secondary factors \eqref{eq:model_sec_factor} does not necessarily rule out static arbitrage, particularly because its drift and diffusion do not satisfy Friedman and Pinsky's condition outlined in \eqref{eq:friedman1973_1}. Therefore, when forward simulating factors, there needs to be additional transformations on particularly secondary factors to ensure no-arbitrage. Rather than scaling diffusion and correcting drift of \eqref{eq:model_sec_factor}, we follow a simpler strategy, considering that secondary factors are strongly dependent on primary factors and have much smaller magnitudes.

Let $\mathcal{P}$ be the $d$-polytope factor state space where there is no static arbitrage. Suppose $\mathcal{P}^p$  is the arbitrage-free $d'$-polytope state space for primary factors, then $\mathcal{P}^p$ is the affine projection of $\mathcal{P} \subset \mathbb{R}^d$ onto $\mathbb{R}^{d'}$ given by $\mathcal{P}^p = \{ y \in \mathbb{R}^{d'} :  \exists x \in \mathcal{P},  y=[x_1 \dots x_{d'}]^\top\}$. Therefore, if we use the neural-SDE model to simulate arbitrage-free primary factors, then there must exist some secondary factors such that the full set of factors reproduces arbitrage-free call prices. Hence, we are able to simulate arbitrage-free factors by:
\begin{enumerate}[leftmargin=*, label=(\roman*)]
\item simulating primary factors using the arbitrage-free neural-SDE model;
\item for each secondary factor, computing the polynomial term from simulated primary factors, and simulating noises from the estimated AR(1) model;
\item adding up the polynomial and noise terms for each secondary factor; if needed, perturb the collection of secondary factors using the arbitrage repair algorithm \cite{Cohen2020}.
\end{enumerate}

We simulate factors from the estimated models, where the simulation results for the 2 primary factors have been presented in Section \ref{sec:outsample_test}. In Figure \ref{fig:factor_distribution}, we compare the marginal distributions of input data, simulated data and arbitrage-repaired data for the 11 secondary factors. We see that repairing arbitrage does not significantly affect the marginal distributions of most factors.

\begin{figure}[!h]
\centering
\includegraphics[scale=0.65]{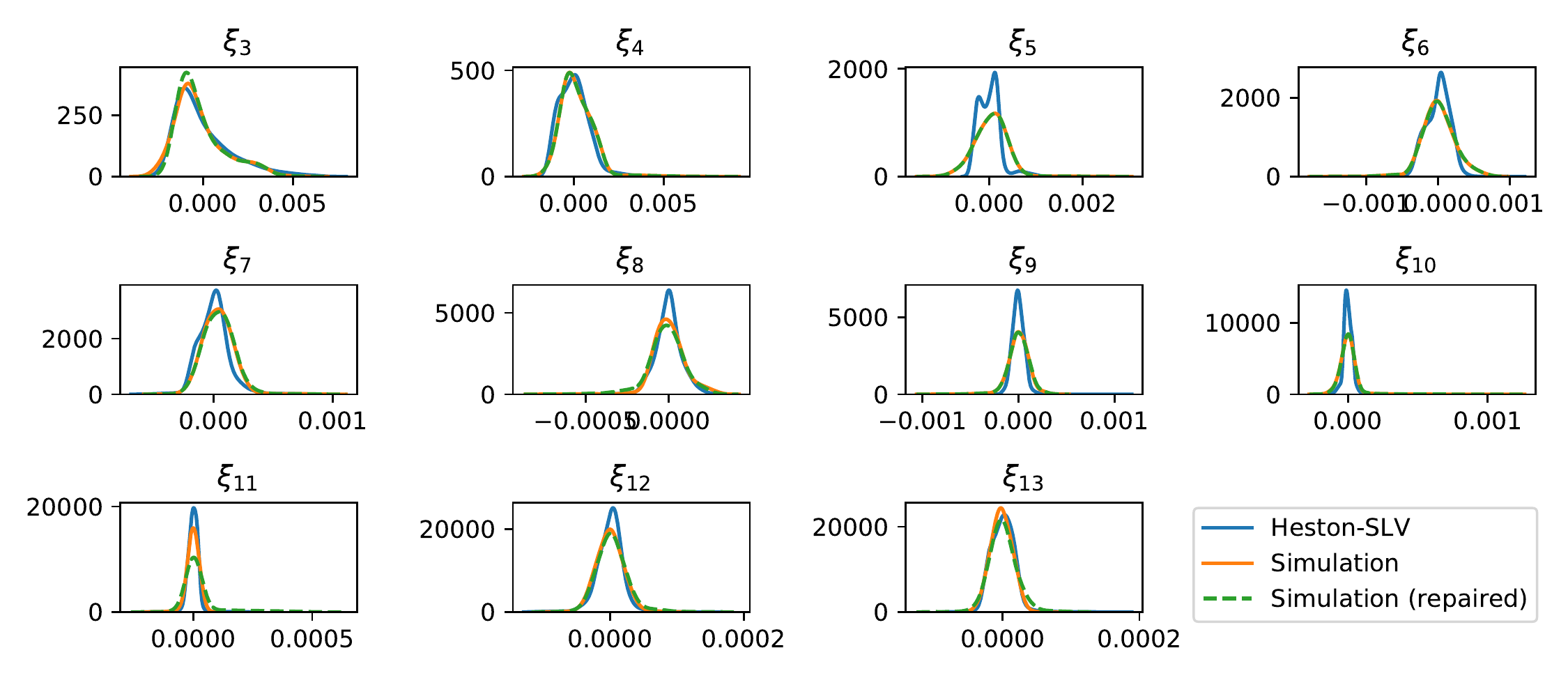}
\caption{Marginal distributions of secondary factors.}
\label{fig:factor_distribution}
\end{figure}

Finally, we compute the VIX index using the simulated factor through \eqref{eq:vix_simulation}, and compare its distribution with that of the Heston-SLV VIX. As seen in Figure \ref{fig:s_vix_full}, its marginal distribution and the joint distribution with the log-return of $S$ look reasonably close to those of the Heston-SLV VIX. Repairing arbitrage has trivial impact on the distributions.  Compared with Figure \ref{fig:s_vix1}, this higher-dimensional model greatly improves the VIX simulation, indicating that the inclusion of more factors leads to better option price reconstruction, thus a more accurate replication of VIX.
\begin{figure}[!ht]
    \centering
    \begin{subfigure}[b]{.49\textwidth}
    \centering
        \includegraphics[scale=.63]{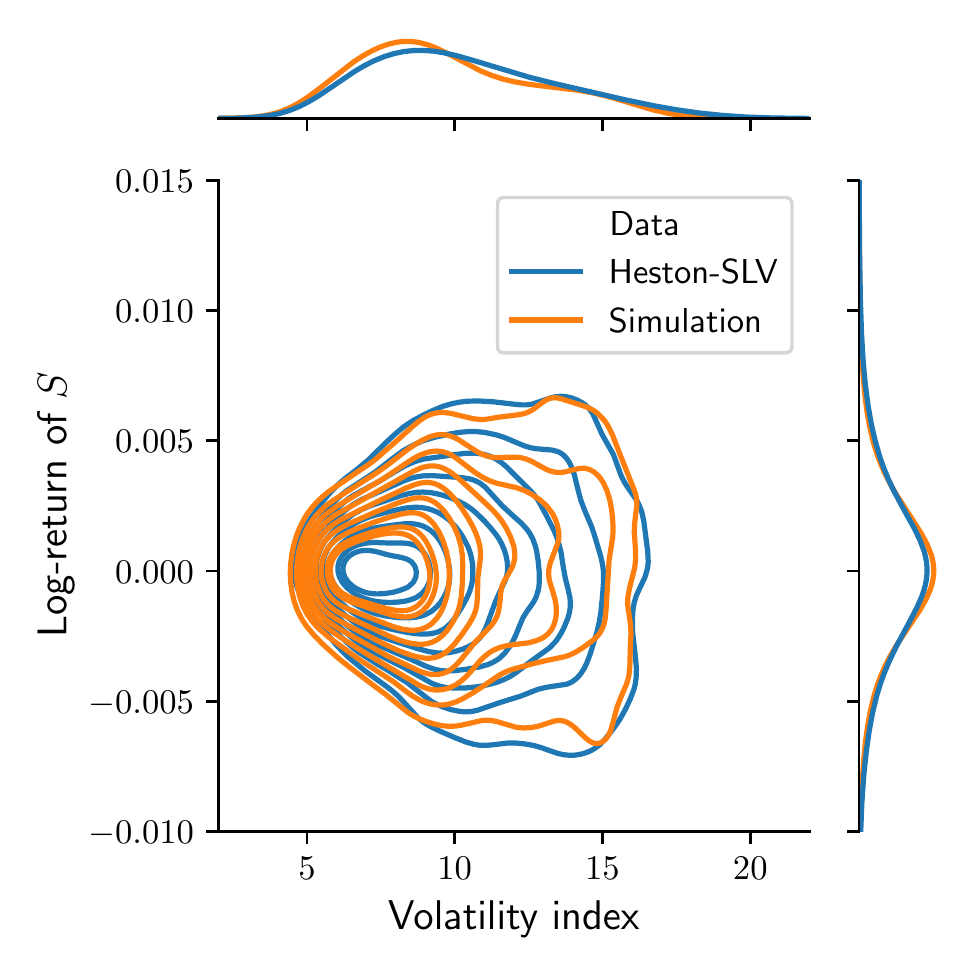}
        \caption{Heston-SLV VIX v.s. VIX calculated from the simulated 13 factors.}
        \label{fig:s_vix11}
    \end{subfigure}
    \hfill
    \begin{subfigure}[b]{.49\textwidth}
    \centering
        \includegraphics[scale=.63]{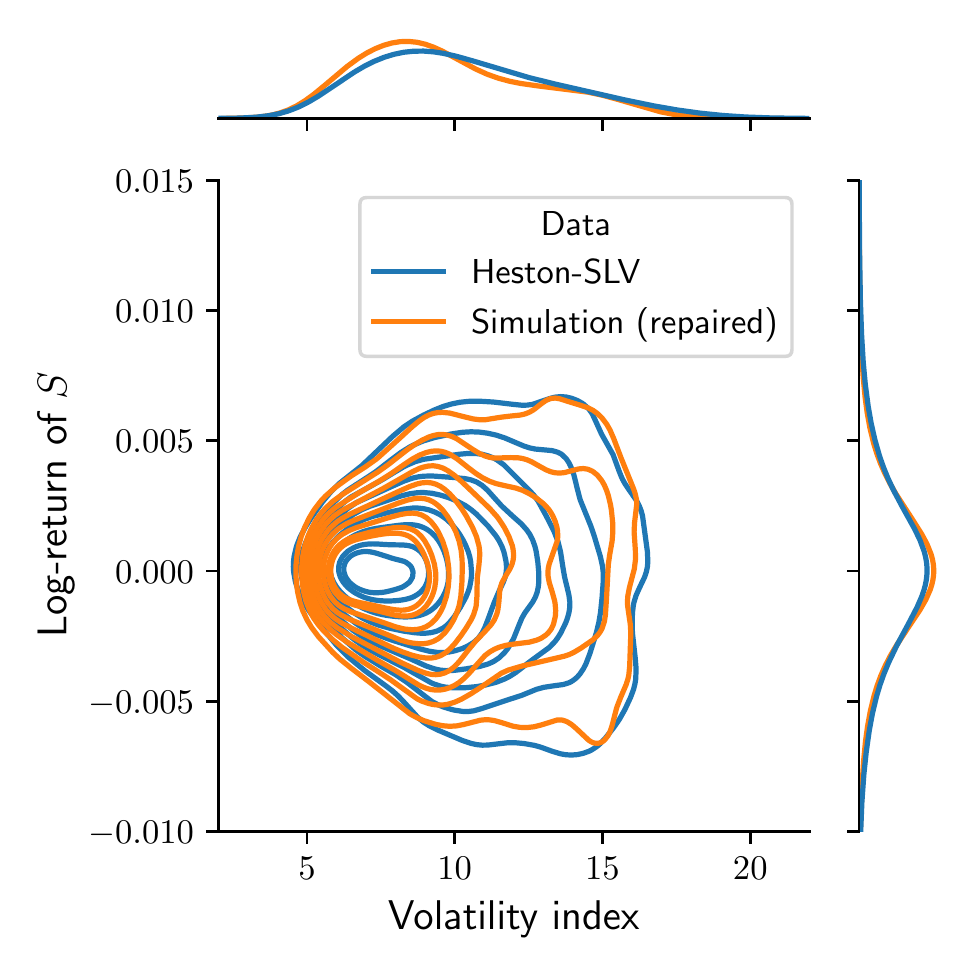}
        \caption{Heston-SLV VIX v.s. VIX calculated from the simulated 13 factors with arbitrage repair.}
        \label{fig:s_vix22}
    \end{subfigure}
    \caption{Joint distribution of the log-return of $S$ and the VIX-like volatility index.}
    \label{fig:s_vix_full}
\end{figure}

\section*{Acknowledgements}
This publication is based on work supported by the EPSRC Centre for Doctoral Training in Industrially Focused Mathematical Modelling (EP/L015803/1) in collaboration with CME Group. We thank Florian Huchede, Director of Quantitative Risk Management, and other colleagues at CME for providing valuable data access, suggestions from the business perspective, and continued support. 

Samuel N. Cohen and Christoph Reisinger acknowledge the support of the Oxford-Man Institute for Quantitative Finance, and Samuel N. Cohen also acknowledges the support of the Alan Turing Institute under the Engineering and Physical Sciences Research Council grant EP/N510129/1.

\setlength{\bibsep}{0.0pt}
\bibliographystyle{abbrv}
\bibliography{reference} 

\end{document}

%% file: figure/tikz_nn.tex
\begin{tikzpicture}[x=1.5cm, y=1.5cm, >=stealth]

\foreach \m/\l [count=\y] in {1,2,3,missing,4}
  \node [every neuron/.try, neuron \m/.try] (input-\m) at (0,2.5-\y) {};

\foreach \m [count=\y] in {1,missing,2}
  \node [every neuron/.try, neuron \m/.try ] (hidden1-\m) at (2,2-\y*1.25) {};

\foreach \m [count=\y] in {1,missing,2}
  \node [every neuron/.try, neuron \m/.try ] (hidden2-\m) at (5,2-\y*1.25) {};

\foreach \m [count=\y] in {1,empty,2}
  \node [every neuron/.try, neuron \m/.try ] (output-\m) at (7,1.5-\y) {};
  
\foreach \m [count=\y] in {1,empty,2}
  \node [trans neuron/.try, neuron \m/.try ] (transform-\m) at (9,1.5-\y) {};
\node[] at (transform-1) {$\mathcal{G}_\mu$};
\node[] at (transform-2) {$\mathcal{G}_\sigma$};
  
\foreach \m [count=\y] in {1,empty,2}
  \node [every neuron/.try, neuron \m/.try ] (target-\m) at (10.8,1.5-\y) {};

\foreach \l [count=\i] in {S, \xi_1 , \xi_2, \xi_d}
  \draw [<-] (input-\i) -- ++(-1,0)
    node [above, midway] {$\l$};

\foreach \l [count=\i] in {1,n}
  \node [above] at (hidden1-\i.north) {}; 

\foreach \l [count=\i] in {1,n}
  \node [above] at (hidden2-\i.north) {}; 

\foreach \l [count=\i] in {\hat{\mu}, \hat{\sigma}}
  \draw [->] (output-\i) -- (transform-\i)
    node [above, midway] {$\l$};
\foreach \l [count=\i] in {\hat{\mu}, \hat{\sigma}}
  \draw [->] (transform-\i) -- (target-\i) {};
  
\foreach \l [count=\i] in {\mu, \sigma}
  \draw [->] (target-\i) -- ++(1,0)
    node [above, midway] {$\l$};

\foreach \i in {1,...,4}
  \foreach \j in {1,...,2}
    \draw [->] (input-\i) -- (hidden1-\j);

\foreach \i in {1,...,2}
  \foreach \j in {1,...,2}
    \draw [->] (hidden1-\i) -- (hidden2-\j);

\foreach \i in {1,...,2}
  \foreach \j in {1,...,2}
    \draw [->] (hidden2-\i) -- (output-\j);

\node [align=center, above] at (0,2) {Input \\layer};
\node [align=center, above] at (2,2) {Hidden \\layer};
\node [align=center, above] at (5,2) {Hidden \\layer};
\node [align=center, above] at (7,1.68) {Output layer \\(underlying \\ function)};
\node [align=center, above] at (9,2) {Transform \\ layer};
\node [align=center, above] at (10.8,1.68) {Output layer \\(target \\ function)};

\node[fill=white,scale=4,inner xsep=0pt,inner ysep=5mm] at ($(hidden1-1)!.5!(hidden2-2)$) {$\dots$};

\end{tikzpicture}

%% file: figure/tikz_diffusion_operator.tex
\begin{tikzpicture}[scale=1]
    \node (v1) at (0,0) {};
    \node (v2) at (4,-1) {};
    \node (v3) at (-2,-1) {};
    \node (v4) at (1, -4) {};
    
    \draw[black, line width=2] (-.4, .1) -- (4.4, -1.1) node[midway, above, sloped] {$\mathbf{v}_1^\top y = b_1$};  
    \draw[black, line width=2] (.4, .2) -- (-2.4, -1.2) node[midway, above, sloped] {$\mathbf{v}_2^\top y = b_2$};  
    \draw[black, line width=2] (-2.2, -.8) -- (1.2, -4.2) node[midway, below, sloped] {$\mathbf{v}_3^\top y = b_3$}; 
    \draw[black, line width=2] (4.2, -.8) -- (.8, -4.2) node[midway, below, sloped] {$\mathbf{v}_4^\top y = b_4$}; 
    
    \Perp[.2]{(v1)}{(v2)}{point1}[-6mm]
    \node[left, red] at (point1) {$\mathbf{v}_1$};
    \Perp[.3]{(v3)}{(v1)}{point2}[-6mm]
    \node[left, red] at (point2) {$\mathbf{v}_2$};
    \Perp[.1]{(v4)}{(v3)}{point3}[-6mm]
    \node[left, red] at (point3) {$\mathbf{v}_3$};
    \Perp[.8]{(v2)}{(v4)}{point4}[-6mm]
    \node[right, red] at (point4) {$\mathbf{v}_4$};
    
    \node (xi) at (2,-2) [circle, fill, scale=.5, label=left:{$y$}] {};
    
    \draw [latex-] ($(v2)!(xi)!(v4)$) -- (xi) node[midway, below, sloped] {$\mathbf{q}_{1}$};
    \node (point5) at ($(v2)!8.5mm!(v1)$) {};
    \draw [latex-] ($(point5)!9mm!(xi)$) -- (xi) node[midway, below, sloped] {$\mathbf{q}_{2}$};
    
\end{tikzpicture}

%% file: figure/tikz_drift_interior.tex
\begin{tikzpicture}[scale=1]
    \node (v1) at (0,0) {};
    \node (v2) at (4,-1) {};
    \node (v3) at (-2,-1) {};
    \node (v4) at (1, -4) {};
    
    \draw[black, line width=2] (-.4, .1) -- (4.4, -1.1) node[midway, above, sloped] {$\mathbf{v}_1^\top y = b_1$};  
    \draw[black, line width=2] (.4, .2) -- (-2.4, -1.2) node[midway, above, sloped] {$\mathbf{v}_2^\top y = b_2$};  
    \draw[black, line width=2] (-2.2, -.8) -- (1.2, -4.2) node[midway, below, sloped] {$\mathbf{v}_3^\top y = b_3$}; 
    \draw[black, line width=2] (4.2, -.8) -- (.8, -4.2) node[midway, below, sloped] {$\mathbf{v}_4^\top y = b_4$}; 
    
    \draw[dashed] (-.4, -.2) -- (4-0.24, -1-0.24);
    \draw[dashed] (.4, -.1) -- (.2-2, -1-.2);
    \draw[dashed] (-2+.4, -1+.2) -- (1+.3, .3-4);
    \draw[dashed] (-.4+4, .1-1) -- (-.25+1, .25-4);
    
    \draw[stealth-stealth] (-.5, -2.5) -- (-.2, -2.2) node[midway, above, sloped] {$\rho^*$};
    
    \node (z1) at (.8,-3) [circle, fill, scale=.2, label=above:{$\zeta_1$}] {};
    \node (z2) at (1.2,-3.2) [circle, fill, scale=.2, label=above:{$\zeta_2$}] {};
    \node (z3) at (.7, -.7) [circle, fill, scale=.2, label=left:{$\zeta_3$}] {};
    \node (z4) at (-.3, -.6) [circle, fill, scale=.2, label=below:{$\zeta_4$}] {};
    
\end{tikzpicture}

%% file: main.bbl
\begin{thebibliography}{10}

\bibitem{quantlib}
{QuantLib}, a free/open-source library for quantitative finance.
\newblock \url{https://www.quantlib.org/}.
\newblock Accessed: 2021-04-13.

\bibitem{tensorflow2015-whitepaper}
M.~Abadi, et al.
\newblock {TensorFlow}: Large-scale machine learning on heterogeneous systems,
  2015.
\newblock Software available from tensorflow.org.

\bibitem{Bain2021}
A.~Bain, M.~Mariapragassam, and C.~Reisinger.
\newblock Calibration of local-stochastic and path-dependent volatility models
  to vanilla and no-touch options.
\newblock {\em The Journal of Computational Finance}, 24(4), 2021.

\bibitem{BlackScholes1973}
F.~Black and M.~Scholes.
\newblock The pricing of options and corporate liabilities.
\newblock {\em Journal of Political Economy}, 81(3):637--54, 1973.

\bibitem{breeden1978}
D.~T. Breeden and R.~H. Litzenberger.
\newblock Prices of state-contingent claims implicit in option prices.
\newblock {\em The Journal of Business}, 51(4):621--51, 1978.

\bibitem{Carmona2007}
R.~Carmona.
\newblock {\em HJM: A Unified Approach to Dynamic Models for Fixed Income,
  Credit and Equity Markets}, pages 1--50.
\newblock Paris--Princeton Lectures on Mathematical Finance 2004. Springer,
  Berlin Heidelberg, 2007.

\bibitem{Carmona2009}
R.~Carmona and S.~Nadtochiy.
\newblock Local volatility dynamic models.
\newblock {\em Finance and Stochastics}, 13(1):1--48, 2009.

\bibitem{Carmona2012}
R.~Carmona and S.~Nadtochiy.
\newblock Tangent {L\'evy} market models.
\newblock {\em Finance and Stochastics}, 16(1):63--104, 2012.

\bibitem{caron1989}
R.~J. Caron, J.~F. McDonald, and C.~M. Ponic.
\newblock A degenerate extreme point strategy for the classification of linear
  constraints as redundant or necessary.
\newblock {\em Journal of Optimization Theory and Applications},
  62(2):225--237, 1989.

\bibitem{CGMY2003}
P.~Carr, H.~Geman, D.~Madan, and M.~Yor.
\newblock {Stochastic volatility for L{\'e}vy processes}.
\newblock {\em Mathematical Finance}, 13(3):345--382, 2003.

\bibitem{Carr2005}
P.~Carr and D.~B. Madan.
\newblock A note on sufficient conditions for no arbitrage.
\newblock {\em Finance Research Letters}, 2(3):125--130, 2005.

\bibitem{cartis2019}
C.~Cartis, J.~Fiala, B.~Marteau, and L.~Roberts.
\newblock Improving the flexibility and robustness of model-based
  derivative-free optimization solvers.
\newblock {\em ACM Transactions on Mathematical Software}, 45(3):Article 32,
  2019.

\bibitem{vix}
CBOE.
\newblock {White paper: CBOE volatility index}, 2019.
\newblock \url{https://cdn.cboe.com/resources/vix/vixwhite.pdf}, Accessed:
  2021-04-08.

\bibitem{Chataigner2020}
M.~Chataigner, S.~Cr\'epey, and M.~Dixon.
\newblock Deep local volatility.
\newblock {\em Risks}, 8(3), 2020.

\bibitem{Cohen2020}
S.~N. Cohen, C.~Reisinger, and S.~Wang.
\newblock Detecting and repairing arbitrage in traded option prices.
\newblock {\em Applied Mathematical Finance}, 27(5):345--373, 2020.

\bibitem{Cont2002}
R.~Cont and J.~{da Fonseca}.
\newblock Dynamics of implied volatility surfaces.
\newblock {\em Quantitative Finance}, 2(1):45--60, 2002.

\bibitem{Cont20022}
R.~Cont, J.~{da Fonseca}, and V.~Durrleman.
\newblock Stochastic models of implied volatility surfaces.
\newblock {\em Economic Notes}, 31(2):361--377, 2002.

\bibitem{cousot2007}
L.~Cousot.
\newblock {Conditions on option prices for absence of arbitrage and exact
  calibration}.
\newblock {\em Journal of Banking \& Finance}, 31(11):3377--3397, 2007.

\bibitem{Cuchiero2020}
C.~Cuchiero, W.~Khosrawi, and J.~Teichmann.
\newblock A generative adversarial network approach to calibration of local
  stochastic volatility models.
\newblock {\em Risks}, 8(4), 2020.

\bibitem{Daglish2007}
T.~Daglish, J.~Hull, and W.~Suo.
\newblock {Volatility surfaces: theory, rules of thumb, and empirical
  evidence}.
\newblock {\em Quantitative Finance}, 7(5):507--524, 2007.

\bibitem{DavisObloj2008}
M.~Davis and J.~Ob\l{}\'{o}j.
\newblock Market completion using options.
\newblock {\em Advances in Mathematics of Finance}, 83:49--60, 2008.

\bibitem{davis2007}
M.~H.~A. Davis and D.~G. Hobson.
\newblock The range of traded option prices.
\newblock {\em Mathematical Finance}, 17(1):1--14, 2007.

\bibitem{Delbaen1994}
F.~Delbaen and W.~Schachermayer.
\newblock A general version of the fundamental theorem of asset pricing.
\newblock {\em Mathematische Annalen}, 300(1):463--520, 1994.

\bibitem{Derman1994}
E.~Derman and I.~Kani.
\newblock Riding on a smile.
\newblock {\em Risk}, 7, 1994.

\bibitem{Derman1998}
E.~Derman and I.~Kani.
\newblock Stochastic implied trees: Arbitrage pricing with stochastic term and
  strike structure of volatility.
\newblock {\em International Journal of Theoretical and Applied Finance},
  01(01):61--110, 1998.

\bibitem{Dugas2009}
C.~Dugas, Y.~Bengio, F.~B\'{e}lisle, C.~Dadeau, and R.~Garcia.
\newblock Incorporating functional knowledge in neural networks.
\newblock {\em Journal of Machine Learning Research}, 10:1239--1262, 2009.

\bibitem{Dupire1994}
B.~Dupire.
\newblock Pricing with a smile.
\newblock {\em Risk Magazine}, 7:18--20, 1994.

\bibitem{Fengler2009}
M.~Fengler.
\newblock Arbitrage-free smoothing of the implied volatility surface.
\newblock {\em Quantitative Finance}, 9:417--428, 2009.

\bibitem{Fengler2015}
M.~R. Fengler and L.-Y. Hin.
\newblock Semi-nonparametric estimation of the call-option price surface under
  strike and time-to-expiry no-arbitrage constraints.
\newblock {\em Journal of Econometrics}, 184(2):242--261, 2015.

\bibitem{friedman1973}
A.~Friedman and M.~A. Pinsky.
\newblock Asymptotic stability and spiraling properties for solutions of
  stochastic equations.
\newblock {\em Transactions of the American Mathematical Society},
  186:331--358, 1973.

\bibitem{Gatheral2014}
J.~Gatheral and A.~Jacquier.
\newblock Arbitrage-free {SVI} volatility surfaces.
\newblock {\em Quantitative Finance}, 14(1):59--71, 2014.

\bibitem{gierjatowicz2020robust}
P.~Gierjatowicz, M.~Sabate-Vidales, D.~{\v S}i{\v s}ka, {\L}.~Szpruch, and {\v
  Z}.~{\v Z}uri{\v c}.
\newblock Robust pricing and hedging via neural {SDEs}, 2020.
\newblock arXiv:2007.04154.

\bibitem{Hagan2002}
P.~S. Hagan, D.~Kumar, A.~S. Lesniewski, and D.~E. Woodward.
\newblock Managing smile risk.
\newblock {\em Wilmott Magazine}, 1:84--108, 2002.

\bibitem{Hagan2014}
P.~S. Hagan, D.~Kumar, A.~S. Lesniewski, and D.~E. Woodward.
\newblock {Arbitrage-free SABR}.
\newblock {\em Wilmott}, 69:60--75, 2014.

\bibitem{harrison1979}
J.~M. Harrison and D.~Kreps.
\newblock Martingales and arbitrage in multiperiod securities markets.
\newblock {\em Journal of Economic Theory}, 20(3):381--408, 1979.

\bibitem{HJM1992}
D.~Heath, R.~Jarrow, and A.~Morton.
\newblock Bond pricing and the term structure of interest rates: A new
  methodology for contingent claims valuation.
\newblock {\em Econometrica}, 60(1):77--105, 1992.

\bibitem{Heston1993}
S.~L. Heston.
\newblock A closed-form solution for options with stochastic volatility with
  applications to bond and currency options.
\newblock {\em Review of Financial Studies}, 6:327--343, 1993.

\bibitem{Hutzenthaler2012}
M.~Hutzenthaler, A.~Jentzen, and P.~E. Kloeden.
\newblock Strong convergence of an explicit numerical method for {SDEs} with
  nonglobally {L}ipschitz continuous coefficients.
\newblock {\em The Annals of Applied Probability}, 22(4):1611--1641, 2012.

\bibitem{itkin2019deep}
A.~Itkin.
\newblock Deep learning calibration of option pricing models: some pitfalls and
  solutions, 2019.
\newblock arXiv:1906.03507.

\bibitem{Schonbucher1999}
P.~J.~Schonbucher.
\newblock A market model for stochastic implied volatility.
\newblock {\em Philosophical Transactions of the Royal Society A: Mathematical,
  Physical and Engineering Sciences}, 357:2071--2092, 1999.

\bibitem{Jacod2010}
J.~Jacod and P.~Protter.
\newblock Risk-neutral compatibility with option prices.
\newblock {\em Finance and Stochastics}, 14(2):285--315, 2010.

\bibitem{Jex1999}
M.~Jex, R.~Henderson, and D.~Wang.
\newblock Pricing exotics under the smile.
\newblock {\em Risk}, pages 72--75, November 1999.

\bibitem{Kahale2003}
N.~Kahale.
\newblock An arbitrage-free interpolation of volatilities.
\newblock {\em Risk Magazine}, 17:102--106, 2004.

\bibitem{Kallsen2013}
J.~Kallsen and P.~Krühner.
\newblock On a {Heath--Jarrow--Morton} approach for stock options.
\newblock {\em Finance and Stochastics}, 19:583--615, 2013.

\bibitem{Karatzas2007}
I.~Karatzas and C.~Kardaras.
\newblock The num\'eraire portfolio in semimartingale financial models.
\newblock {\em Finance and Stochastics}, 11:447--493, 2007.

\bibitem{kellerer1972}
H.~G. Kellerer.
\newblock {Markov-Komposition} und eine {Anwendung} auf {Martingale}.
\newblock {\em Mathematische Annalen}, 198:99--122, 1972.

\bibitem{Kreps1981}
D.~Kreps.
\newblock Arbitrage and equilibrium in economies with infinitely many
  commodities.
\newblock {\em Journal of Mathematical Economics}, 8(1):15--35, 1981.

\bibitem{Motzkin1953}
T.~Motzkin, H.~Raiffa, G.~L. Thompson, and R.~M. Thrall.
\newblock The double description method.
\newblock In {\em Contributions to the Theory of Games II}. Princeton
  University Press, 1953.

\bibitem{Piterbarg2006}
V.~Piterbarg.
\newblock Markovian projection method for volatility calibration.
\newblock SSRN preprint 906473, 2006.

\bibitem{Reiswich2009}
D.~Reiswich and U.~Wystup.
\newblock {FX} volatility smile construction.
\newblock CPQF Working Paper Series~20, Frankfurt School of Finance and
  Management, Centre for Practical Quantitative Finance (CPQF), 2009.
\newblock Available at \url{https://EconPapers.repec.org/RePEc:zbw:cpqfwp:20};
  accessed 2021-05-20.

\bibitem{Schwarz2017}
D.~C. Schwarz.
\newblock Market completion with derivative securities.
\newblock {\em Finance and Stochastics}, 21(1):263--284, 2017.

\bibitem{Schweizer2008}
M.~Schweizer and J.~Wissel.
\newblock Arbitrage-free market models for option prices: the multi-strike
  case.
\newblock {\em Finance and Stochastics}, 12(4):469--505, 2008.

\bibitem{Schweizer2007}
M.~Schweizer and J.~Wissel.
\newblock Term structures of implied volatilities: Absence of arbitrage and
  existence results.
\newblock {\em Mathematical Finance}, 18(1):77--114, 2008.

\bibitem{Szpruch2018}
{\L}.~Szpruch and X.~Zhang.
\newblock V-integrability, asymptotic stability and comparison theorem of
  explicit numerical schemes for {SDEs}.
\newblock {\em Mathematics of Computations}, 87:755--783, 2018.

\bibitem{Venttsel1965}
A.~Ventzel'.
\newblock On equations of the theory of conditional {Markov} processes.
\newblock {\em Theory of Probability and its Applications}, 10(2):357--360,
  1965.

\bibitem{Wissel2008}
J.~S. Wissel.
\newblock {\em Arbitrage-free market models for liquid options}.
\newblock PhD thesis, ETH Z\"urich, 2008.

\end{thebibliography}
